\documentclass[11pt]{article}

\usepackage{enumitem} 
\setlist{noitemsep,topsep=0pt,parsep=0pt} 

\usepackage[margin=1cm]{caption} 
\usepackage{subcaption}
\usepackage{lineno}
\usepackage{amsmath,amsthm}
\usepackage{graphicx}
\usepackage{amsfonts}
\usepackage{enumitem}
\usepackage{mathtools}
\usepackage{empheq}
\usepackage{float}
\usepackage{tikz}
\usetikzlibrary{positioning, shapes.misc}
\usepackage{mathtools} 

\usepackage{tikz}
\usetikzlibrary{arrows,backgrounds,calc,fit,decorations.pathreplacing,decorations.markings,shapes.geometric}

\tikzset{every fit/.append style=text badly centered}

\usepackage{tyson}

\usepackage[bookmarks=true,hypertexnames=false]{hyperref}
\hypersetup{colorlinks=true, citecolor=blue, linkcolor=red, urlcolor=blue}

\newcommand{\Holant}{\operatorname{Holant}}
\newcommand{\PlHolant}{\operatorname{Pl-Holant}}
\newcommand{\holant}[2]{\ensuremath{\Holant\left(#1\mid #2\right)}}
\newcommand{\plholant}[2]{\ensuremath{\PlHolant\left(#1\mid #2\right)}}

\newcommand{\numP}{{\rm \#P}}

\newcommand{\smm}[4]{\left[ \begin{smallmatrix} #1 & #2 \\ #3 & #4 \end{smallmatrix}\right]}
\newcommand{\smmv}[2]{\left[ \begin{smallmatrix} #1 \\ #2 \end{smallmatrix}\right]}
\newcommand{\smmh}[2]{\left[ \begin{smallmatrix} #1 & #2 \end{smallmatrix}\right]}

\newcommand{\rlm}{\lambda/\mu}

\theoremstyle{definition}
\newtheorem{remark}{Remark}

\def\partIPlCSP2SecNum{5}
\def\partIPl2CSPSecPageNum{24}

\tikzstyle{internal} = [draw, fill, shape=circle]
\tikzstyle{external} = [shape=circle]
\tikzstyle{square}   = [draw, fill, rectangle]
\tikzstyle{triangle} = [draw, fill, regular polygon, regular polygon sides=3, inner sep=3pt]
\tikzstyle{pentagon} = [draw, fill, regular polygon, regular polygon sides=5, inner sep=2pt, minimum size=14pt]

\title{Planar 3-way Edge
Perfect Matching Leads to A Holant Dichotomy}

\author{
Jin-Yi Cai%
\thanks{Department of Computer Sciences, University of Wisconsin--Madison.}\\
\footnotesize \texttt{jyc@cs.wisc.edu}
\and
Austen Z. Fan$^*$\\
\footnotesize \texttt{afan@cs.wisc.edu}
}

\date{} 

\begin{document}
\pagenumbering{arabic}
\maketitle

\begin{abstract}
We prove a complexity dichotomy theorem
for  a class of Holant problems on planar 3-regular bipartite graphs. The complexity dichotomy states that for every weighted constraint function $f$  defining the problem (the weights can 
even be negative),
the problem is either
computable in polynomial time if $f$ satisfies a tractability
criterion, or  \#P-hard otherwise.
One particular problem in this 
problem space 
is  a long-standing open problem of
Moore and Robson~\cite{MooreR01} on  counting Cubic Planar X3C.
The dichotomy resolves this problem by showing that it is \numP-hard.
Our proof relies on the machinery of 
signature theory developed in the study
of Holant problems. An essential ingredient in our proof of the main dichotomy theorem is
a pure graph-theoretic result: Excepting
 some trivial cases, every  
3-regular plane graph has
a planar 3-way edge perfect matching.
The proof technique of this graph-theoretic  result is a combination of
algebraic and combinatorial methods.

The  P-time  tractability criterion of the dichotomy is explicit.
Other than the known classes of tractable constraint functions
 (degenerate, affine, product type, matchgates-transformable) we also identify a new infinite set of P-time computable
planar Holant problems; however, its
tractability is not by a direct holographic
transformation to matchgates, but by a combination of 
this method and a global argument. The complexity dichotomy
states that everything else in this Holant class is \#P-hard.
\end{abstract}
\newpage

\section{Introduction}\label{sec:intro}

Holant problems are also known as edge-coloring models.
They can express a broad class of counting problems,
such as counting matchings,
perfect matchings
({\sc \#PM}),
proper edge-colorings, cycle coverings,
and a host of counting orientation problems such as counting
Eulerian orientations or the six-vertex model.
Every counting constraint satisfaction problem (\#CSP)  can be expressed
as a Holant problem. 
On the other hand,
Freedman, Lov\'asz and Schrijver~\cite{freedman2007reflection} proved that the prototypical Holant problem
\#PM cannot be expressed 
as a graph homomorphism function (vertex-coloring model)
by any real valued constraint function.
This is  true even for 
complex valued constraint functions~\cite{CaiG19}. 

Some problems are \#P-hard in general, yet computable on
planar graphs. 
The problem \#PM is such a problem~\cite{Valiant79,Jerrum87}.
A most fascinating algorithm---the FKT algorithm~\cite{kasteleyn1967graph,kasteleyn1961statistics, temperley1961dimer}---computes \#PM in polynomial time (FP, polynomial-time computable functions)
for planar graphs. 
Valiant  introduced  holographic algorithms
which are non-parsimonious
reductions to the FKT
 algorithm, 
 placing  many planar counting problems  in FP
 that 
 seemed to be intractable.
To understand these algorithms
 a signature theory was developed and the Holant
 framework was introduced.
 Stated in this signature theory,
  Valiant's holographic algorithms boil down to
  what constraint functions (signatures) can be realized by the so-called matchgate signatures
  under a holographic transformation.
  Delineating the precise
boundary of FP tractability for these problems has been a central focus in the classification theory of
  counting problems~\cite{Backens21,CaiF17,CaiFGW22,GuoW20,CaiLX17,Valiant08}.
 A general theme has emerged:
 for  very broad classes of counting problems, one can
 classify every problem in the class to be of
 exactly one of three types: (1) 
 FP, (2) \#P-hard in general but FP
 on planar graphs, or (3)  \#P-hard on planar graphs. Furthermore, for all \#CSP on Boolean variables (which includes vertex models),  Valiant's holographic algorithm
  is a universal
 algorithm~\cite{CaiF17} that solves problems in (2)\footnote{However, this is not true for Holant problems in general~\cite{CaiFGW22}. If one recalls that
the FKT algorithm solves  \#PM for planar graphs, which is
\emph{the} prototypical  Holant problem but not a vertex model,
this is particularly intriguing.}. 
In this paper, we prove that for a class of bipartite
Holant problems, this three-way classification \emph{holds}.
However,
 there are \emph{two methods} 
for planar tractability in type (2): In addition to
holographic
transformations to matchgates, there is another type which
combines this transformation with a global argument. Either method alone is
not, but \emph{together} they do form, a
universal strategy for planar tractability.

%


We briefly define  
Holant
problems on Boolean variables. An
input is a signature grid $\Omega$
consisting of a graph $G = (V,E)$ with 
each vertex $v$  labeled by a constraint function $f_v$
(also called a signature).
The Holant problem is to compute a sum-of-product
$\operatorname{Holant}(\Omega) = 
\sum_{\sigma: E \rightarrow \{0, 1\}} 
\prod_{v\in V}f_v(\sigma|_{E(v)})$,
where $E(v)$ denotes the incident edges  of $v$.
E.g., \#PM is the counting problem where each $f_v$ is the 0-1 valued \textsc{Exact-One} function.
In planar Holant problems, denoted by $\operatorname{Pl-Holant}$,
$G$ is required to be  planar, 
and $f_v$ takes inputs from 
$E(v)$ which is given a cyclic order starting from some edge  (specified by $\Omega$).


In this paper, we study a class of Holant problems
whose input graphs are planar, 3-regular and bipartite.
More precisely, let $f(x, y, z)= [f_0, f_1, f_2, f_3]$
be any ternary constraint function  which takes value $f_i \in  {\mathbb{Q}}$, if the input has Hamming weight $i$.  We allow both positive and negative
values.
 We study  \plholant{f}{(=_3)},
the Holant problem on planar, 3-regular bipartite
graphs where LHS vertices are assigned $f$ and RHS
vertices are assigned a ternary equality $(=_3)$.
Without  planarity, a complexity dichotomy was proved for 
 these bipartite Holant problem in~\cite{CaiFL21}.
 Planarity plus regularity add considerable difficulty.



One can think of them as counting problems on
3-regular 3-uniform hypergraphs, or set systems where
every subset has cardinality 3 and every element appears in 3 subsets.
The planarity refers to its (bipartite) incidence  graph.
These include some well studied problems.
One long-standing open problem raised by Moore and Robson in~\cite{MooreR01}
is counting
\texttt{Cubic-Planar-X3C}, 
(X3C stands for \textsc{Exact-3-Cover}), or equivalently Cubic Planar Monotone 1-in-3 SAT. Expressed as a Holant problem it is 
\plholant{[0,1,0,0]}{(=_3)}, 
where $[0,1,0,0]$ is the ternary \textsc{Exact-One} function.
Schaefer~\cite{Schaefer78} proved that Monotone 1-in-3 SAT is NP-complete. 
Lichtenstein~\cite{Lichtenstein82} first considered the complexity of many
 planar  problems,
    and Laroche~\cite{Laroche93} proved that
Planar Monotone 1-in-3 SAT is NP-complete.
Monotone 1-in-3 SAT is the same as X3C.
Dyer and Frieze~\cite{DyerF86} proved the NP-completeness of  
Planar X3C and 3DM where 
 each element is in either 2 or 3 subsets (of
cardinality 3).
Moore and Robson~\cite{MooreR01}, in a reduction
using ingenious combinatorial gadgets, showed that
this problem remains NP-complete when 
each element is in exactly 3 subsets. 
However, they 
noted that they  were not able to conclude the \numP-hardness of its counting version, which is 
precisely
 \plholant{[0,1,0,0]}{(=_3)}, while all previous NP-complete proofs listed here
 do extend to \#P-hardness for its counting version.
We observe that  these proofs are combinatorial, and they become
 increasingly  more delicate with  planarity and regularity restrictions.

Our proof is carried out using
 the machinery of
signature theory developed in the study
of Holant problems. These are algebraic proofs
which show that the underlying combinatorial constructions succeed.
This machinery demonstrates the power of using {\em algebraic} method 
to prove complexity results which are combinatorial in nature.
Indeed, this is exactly in the spirit of Valiant's holographic algorithms
which use arithmetic cancellations to achieve reductions
that are globally valid for counting, but solutions do not
correspond in a 1-1 fashion (i.e.,
non-parsimonious reductions).


One difficulty in working with 3-regular bipartite Holant problems is the severe limitation on the gadgets that  can be possibly constructed. One can  show that
 on either side of the bipartite problem,
 every constructible gadget defines 
a constraint function having  arity 
a multiple of 3. So in particular, one cannot directly
produce unary signatures, or binary signatures on either side.
One \emph{can} produce  ``straddled'' signatures that take some input variables
from one side and some 
from the other. Typically a ``degenerate'' signature
is not very useful in the proof of a dichotomy theorem. A counter-intuitive
idea from~\cite{CaiFL21} is to utilize straddled and degenerate
signatures, to ``virtually'' produce unary signatures.
This idea led to a complexity dichotomy for these bipartite
counting problems in the setting that ignores planarity.
But the essence of Valiant's  holographic algorithm
and the study of Holant problems is to account for planar
tractability, and
we know there are problems in this class
that are \#P-hard in general but in FP
on planar graphs.

In this paper we settle that by proving a planar complexity dichotomy.
A major technical challenge 
is how to ``virtually'' produce   unary signatures in a planar way.
We prove
a pure graph-theoretic result that says that, except
in some trivial cases, every  
3-regular plane graph~\footnote{
A 3-regular graph is also called a \textit{cubic} graph. Properties of cubic planar graphs have been studied extensively~\cite{HoltonM88, AldredBHM00, NoyRR20, HeckmanT06, Scheim74}.
} has
a planar 3-way edge perfect matching (P3EM). 
We use it as an essential ingredient to the proof of the dichotomy.
This result should be of independent interest.
%
The proof technique to prove this  matching theorem is a combination of
algebraic and combinatorial methods.
This theorem lets us virtually
``manufacture'' and then ``absorb''
unary signatures in the \#P-hardness reduction. 
This allows us to carry out the needed \#P-hardness
 reductions in a planar way.



\section{Preliminaries and Our Main Theorem}
A (symmetric) constraint function
(a.k.a. signature) of arity $n$
is $f = [f_0, f_1, \ldots, f_n]$, where $f_i$ denotes the function
value  on inputs of Hamming weight $i$.
E.g., the  ternary \textsc{Exact-One} function
is  $[0,1,0,0]$, and
the ternary \textsc{Equality}
function $\left(=_{3} \right)$
is  $[1,0,0,1]$.
In this paper we consider the following 
set of $\operatorname{Holant}$ problems, denoted by 
$\plholant{f}{(=_3)}$,
where $f$ is a ternary function. An input is a signature grid $\Omega$ consisting of a planar 3-regular bipartite graph $G = (U,V,E)$, where each vertex in $U$ is assigned 
$f= [f_0,f_1,f_2,f_3]$ with values $f_i \in \mathbb{Q}$, and each vertex in $V$ is assigned $\left(=_{3} \right)$. The Holant problem
is to compute
 
$$
\operatorname{Holant}\left(\Omega\right) = \sum_{\sigma: E \rightarrow \{ 0,1\}} \prod_{u \in U} f\left(\sigma |_{E(u)}\right) \prod_{v \in V} \left( =_{3} \right) \left(\sigma |_{E(v)}\right).
$$


For clarity, we shall call vertices in $U$ are on the left hand side (LHS) and vertices in $V$ are on the right hand side (RHS). 
%
We can write a  signature of arity $n$ as a vector in 
$\mathbb{Q}^{2^n}$ indexed in lexicographical order. 
A 
(symmetric) 
signature $f$ is  {\em degenerate} if there exists a unary signature $u \in \mathbb{C}^2$ such that $f = u^{\otimes n}$, the $n$th tensor power.

The main result of this paper is the following dichotomy theorem:
\begin{theorem}\label{thm:main}
\plholant{f}{(=_3)} where $f = [f_0,f_1,f_2,f_3]$ and  $f_i \in \mathbb{Q}$ {\rm (}$0 \le i \le 3${\rm )} is \numP-hard except in the following cases, for which the problem is in $\operatorname{FP}$:
(1) $f$ is degenerate;
(2) $f = [a, 0, 0, b]$, for some $a, b$;
(3) $f = [a,0,\pm a,0], [0,a,0,\pm a], [a,-a,-a,a], [a,a,-a,-a]$ for some $a$;
(4) $f = [a, b, b, a]$ or $f =[a,b,-b,-a]$, for some $a, b$;
(5) $f= [3a+b,-a-b, -a+b, 3a-b]$ for some $a, b$.
Without the
planar restriction, the 
problem 
$\operatorname{Holant}(f\, |\,(=_3))$ remains in $\operatorname{FP}$ in cases (1), (2) and (3), but \#P-hard in cases
(4) and (5).
\end{theorem}

In case (1), the signature $f$  decomposes into
three unary signatures. In case (2), $f$ is a generalized
equality. In  case (3), $f$ is in the affine class. In case (4), the Holant problem is transformable to planar \#PM with matchgates
(see more details about these tractable classes in~\cite{cai2017complexity}). In case (5), the planar P-time tractability is 
 \emph{neither} by  Valiant's holographic reduction alone,
  \emph{nor} entirely independent from it. Rather it is by a 
  combination of a holographic reduction together with
  a global argument. 

As mentioned in Section~\ref{sec:intro},
counting \texttt{Cubic-Planar-X3C} is just
\plholant{[0,1,0,0]}{(=_3)}, the counting problem of Moore and Robson~\cite{MooreR01}.
It clearly belongs to this class.
It is also equivalent to Cubic Planar Monotone 1-in-3 SAT. 
By Theorem~\ref{thm:main}, it is \#P-complete.

To see that case (5) is planar tractable,
we prove 
for any $a$ and $b$,
the value of \plholant{f}{(=_3)} for $f = [3a+b, -a-b, -a+b, 3a-b]$ on any planar signature grid exactly equals   the value of \plholant{[0,2a,0,0]}{[0,1,0,0]} on the same signature grid,
and thus can be computed by the FKT algorithm
for counting perfect matchings. Indeed, by a holographic transformation
using $H = \left[\begin{smallmatrix}1 & 1 \\ 1 & -1\end{smallmatrix}\right]$ we have the following sequence of equivalences:
\begin{eqnarray*}
       \plholant{f}{(=_3)} & \equiv_{T} & \plholant{fH^{\otimes 3}}{(H^{-1})^{\otimes 3}(=_3)}\\  &  \equiv_{T} & \plholant{[0,0, 2a,2b]}{[1,0,1,0]}\\& \equiv_{T} & \plholant{[0,0,2a,0]}{[0,0,1,0]} \\  &\equiv_{T} & \plholant{[0,2a,0,0]}{[0,1,0,0]} 
\end{eqnarray*}
where the third equivalence follows from the observation that for each nonzero
term in the Holant sum, every vertex on the LHS has at least two of three edges
assigned 1 (from $[0, 0, 2a, 2b]$), meanwhile every vertex on the RHS
has at most two of three edges assigned 1 (from $[1, 0, 1, 0]$). The graph being
bipartite and 3-regular, the number of vertices on both sides must equal, thus
every vertex has exactly two incident edges assigned 1.

  An example of a planar tractable
  problem that belongs to case (5) is as follows.
 It can be viewed as a covering problem on 3-uniform hypergraphs of degree 3.
We say  $(X, {\cal S})$  is a 3-regular $k$-uniform set system
(or 3-regular $k$-uniform hypergraph),
if ${\cal S}$ consists of a family of sets $S \subset X$ each of
size $|S| =k$, and every $x \in X$ is in exactly 3 sets.
If $k=2$ this is just an ordinary 3-regular graph
(where the $2$-subsets are ordinary edges).  We consider
 3-regular $3$-uniform set systems.
 We say ${\cal S'}$ is a \emph{leafless partial cover} 
 if every $x \in \bigcup_{S \in {\cal S'}} S$ belongs to more than
 one set $S  \in {\cal S'}$. We say $x$ is \emph{lightly covered}
 if $|\{S \in {\cal S'} : x \in S\}|$ is 2,
 and \emph{heavily covered} if this number is 3.
  
  \vspace{.1in}
\noindent$\mathbf{Problem:}$  \texttt{ Weighted-Leafless-Partial-Cover}.

\noindent$\mathbf{Input:}$ A 3-regular $3$-uniform set system
 $(X, {\cal S})$.

\noindent$\mathbf{Output:}$ $\sum_{\cal S'} (-1)^l 2^h$,
where the sum is over all leafless partial covers  ${\cal S'}$,
and $l$ (resp. $h$) is the number of $x \in X$ that are
{lightly covered} (resp. {heavily covered}).
  \vspace{.1in}

Expressed in the Holant framework this problem is just
\holant{f}{(=_3)}, where $f = [1, 0, -1, 2]$.
This problem belngs to case (5)  with $a=1/2$ and $b= - 1/2$.

\begin{figure}[ht]
    \centering
    \begin{tikzpicture}[scale=0.5]
        \draw[very thick] (-4,0)--(0,2);
        \draw[very thick] (-4,0)--(-2,0);
        \draw[very thick] (-4,0)--(0,-2);

        \draw[very thick] (4,0)--(0,2);
        \draw[very thick] (4,0)--(2,0);
        \draw[very thick] (4,0)--(0,-2);

        \draw[very thick] (0,1)--(0,2);
        \draw[very thick] (0,1)--(2,0);
        \draw[very thick] (0,1)--(-2,0);

        \draw[very thick] (0,-1)--(0,-2);
        \draw[very thick] (0,-1)--(2,0);
        \draw[very thick] (0,-1)--(-2,0);
        
        \filldraw [blue] (0,2) circle (4pt);
        \filldraw [blue] (0,-2) circle (4pt);
        \filldraw [blue] (2,0) circle (4pt);
        \filldraw [blue] (-2,0) circle (4pt);
        
        \filldraw [red] (-4,0) circle (4pt);
        \filldraw [red] (4,0) circle (4pt);
        \filldraw [red] (0,-1) circle (4pt);
        \filldraw [red] (0,1) circle (4pt);
    \end{tikzpicture}
    \caption{A small instance for the problem \texttt{Weighted-Leafless-Partial-Cover}}
    \label{problem example}
\end{figure}
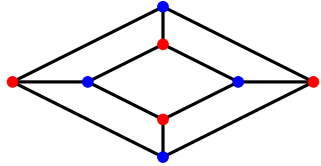

Figure~\ref{problem example} illustrates a small instance of this problem. Blue dots represent elements $x$ and red dots represent the family of sets $S$. An element $x$ is contained in a set $S$ if and only if the blue dot for $x$ is connected to the red dot for $S$. It is not hard to see that there are exactly 6 leafless partial covers, which are $\emptyset$, any family of 3 sets (there are 4 of them), and the family of all 4 sets. Therefore, the Holant value of this instance is $1 + 4(-1)^32^1+2^4 = 9$. 
One can also verify that there are 
exactly 9 distinct perfect matchings in the graph
in Figure~\ref{problem example}.

Therefore, it is known that cases (1)--(5) are
in $\operatorname{FP}$. The main claim lies in that all other cases are \numP-hard over planar graphs.
The cases (4) and (5) capture precisely those problems
that are \#P-hard on general but in FP on planar graphs; neither case alone does that.




A gadget in this paper, such as those illustrated in Figure~\ref{f1} and Figure~\ref{a nice gadget},
is a planar 3-regular bipartite graph $G = (U, V, E_{\rm in}, E_{\rm out})$ with
internal edges $E_{\rm in}$ and dangling edges $E_{\rm out}$.
There can be $m$ dangling edges internally incident
to vertices from $U$ and $n$ dangling edges internally incident
to vertices from $V$. These $m+n$ dangling edges 
correspond to Boolean variables $x_1, \ldots, x_m, y_1, \ldots,  y_n$
and the gadget defines
a signature
\[f(x_1, \ldots, x_m, y_1, \ldots,  y_n)
=
\sum_{\sigma: E_{\rm in} \rightarrow \{ 0,1\}} \prod_{u \in U} f\left(\widehat{\sigma} |_{E(u)}\right) \prod_{v \in V} \left( =_{3} \right) \left(\widehat{\sigma}  |_{E(v)}\right),
\]
where $\widehat{\sigma} $ denotes the extension
of $\sigma$ by the assignment on the  dangling edges.
The variables $x_1, \ldots, x_m$ (respectively, $y_1, \ldots,  y_n$) are called 
LHS (respectively, RHS) variables and are to be connected 
externally to RHS (respectively, LHS) signatures
in $\plholant{f}{(=_3)}$.


Gadgets that are constructible for \plholant{f}{(=_3)} is severely limited due to  planarity and bipartiteness.
Suppose $g$ is the signature
of a gadget construction with all of its variables  on the LHS.
Then by a simple counting argument, the arity of $g$ must be a multiple of 3.
The same is true for a  gadget construction
with all of its variables  on the RHS.
In particular, one cannot hope to produce any unary or binary signature on either side. However, being able to have unary or binary signatures at hand has been proven to be very useful in studying Holant problems.

To tackle this difficulty,
{\em straddled} gadgets were introduced~\cite{CaiFL21}
which have both LHS and RHS  variables. For example, the gadget $G_1$ in Figure~\ref{f1}, after we place $f$ 
on the square vertex and  $(=_3)$ 
on  the circle vertex, 
has one variable on the LHS (the dangling edge that connects to  a square)  and one variable on the RHS (the dangling edge that connects to  a circle).
We list the values of 
a signature $f$ in a {\em signature matrix} $M_f$ where the row(s) $R$ and column(s) $C$ correspond to assignments,
\newcommand{\cupdot}{\mathbin{\mathaccent\cdot\cup}}
in lexicographic order, of input variables $X = R \cupdot C$.
 We may identify $f$ with $M_f$. 
 When two signatures  $M_f$ and $M_g$
 are composed by merging the dangling edges
 of the column variables of $f$ with the row variables of $g$, 
 the signature matrix of the resulting signature is the matrix product $M_fM_g$.
In our paper, the composition must respect the bipartiteness and planarity. Also, note that if a straddled gadget has $m$ dangling edges to be connected to RHS and $n$ dangling edges to be connected to LHS, then $m-n \equiv 0 \bmod 3$.

One crucial
idea in~\cite{CaiFL21} is to interpolate \emph{degenerate} straddled binary signatures and use them as two unary signatures;
one of which is desired and the other is to be grouped together  to form an easily computable positive constant, which does not affect the complexity. 
%
%
However, the ``grouped together'' process destroys the planar structure and thus the reduction  fails  for planar graphs. However,
we can make it work  for planar graphs  if we can group
these leftover unaries three at a time within each face. 
This is where planar 3-way edge matching (P3EM) comes in.
Our theorem on P3EM will allow us to do that.

More formally, let $G = (V, E)$ be an undirected plane graph, i.e., a planar graph
with a given planar embedding. We allow $G$ to be a multi-graph, i.e.,
parallel edges  and self-loops are allowed. A planar 
3-way edge matching (P3EM) is a partition of $E$ into a collection $M$
of 3-edge subsets $E = \bigcup_{t \in M} t$ such that we can
add one vertex $v_t$ for each $t \in M$ and connect $v_t$ to
the three edges of $E$ in $t$ so that the resulting graph is
still a plane graph. In Section~\ref{section:P3EM}, we prove that a P3EM always exists for any plane 3-regular graph (except for some trivial cases) and, moreover, can be constructed in polynomial time. An often-used technique in dealing with plane graphs is first taking a spanning tree of the dual graph and picking a root node, e.g., the node associated to the outer face. Starting from a leaf, one argues that some invariant property can be ``propagated'' through the tree until finally reaching the root. This technique is used in~\cite{MooreR01} as well as in the proof of previous dichotomies concerning planarity~\cite{CaiF17, CaiFGW22}. However, this technique does not work in this case. New techniques have to be invented.
The proof of our P3EM theorem is a mixture of algebra and combinatorics, and it should be of independent interest.

\section{Planar 3-way Edge Matching (P3EM)}\label{section:P3EM}
We begin with the following lemma.

\begin{lemma}\label{lem:P3EM-equivalce}
A 3-regular plane graph $G$ has a P3EM iff
there is an assignment that assigns 
each edge to an adjacent face
so that the number of edges assigned to each face is $0 \bmod 3$.
\end{lemma}
\begin{proof}
If $E = \bigcup_{t \in M} t$ 
is a P3EM, then for  every  3-edge subset $t \in M$,
the point $v_t$ belongs to a face 
adjacent to all three edges in $t$. This gives the assignment of $E$.

Conversely, suppose there is such an assignment $\sigma$ for $G$,
and we first assume $G$ is a connected plane graph.
Then a partition of $E$ into 3-edge subsets
can be obtained by collecting  consecutive triples from edges that are assigned toward any face $F$ along a cyclic traversal of the boundary of $F$. 
This produces  a P3EM for $G$. 

Now suppose
$G$ is disconnected and we consider $G$ 
as given on the sphere $S^2$. There is a simple closed curve $S$
(homeomorphic to a circle $S^1$) disjoint from $G$
separating $S^2$ into two discs $D_1$ and $D_2$, and
separating  $G$  into two nonempty
disjoint plane graphs $G_1 \cup G_2$, with $G_i \subset D_i$ ($i=1,2$).
$S$ is contained in  some face $F$. For   $i=1,2$,
every face other than $F$ in $D_i$ is assigned
by $\sigma$ edges  from $E(G_i)$ only, and the number of which is $0 \bmod 3$. 
Since $G_i$ is 3-regular, we have $|E(G_i)| \equiv 0 \bmod 3$.
Hence the number of edges from $E(G_i)$ assigned to $F$ is also $0 \bmod 3$.
Thus the restriction of $\sigma$ to $E(G_i)$ is an edge assignment  for $G_i$ that satisfies the stipulation
in the lemma statement. Formally, for $G_1$ we can remove $G_2$,  then $F$ becomes
an extended face $\widehat{F}$ containing $D_2$, and
we get an assignment $\sigma_1$ from $E(G_1)$
to the set of faces of $G$ in $D_1$ together with  $\widehat{F}$.
For  $G_1$ we can contract $D_2$ to a single point, and $\widehat{F}$ becomes essentially the intersection of $F$ with $D_1$. 
Similarly we have an assignment $\sigma_2$ for $G_2$.

 By induction, we have a P3EM for both $G_1$ and $G_2$, made up of triples of edges of $G_1$ 
and $G_2$ separately.  
Due to the contraction of $D_2$ for $G_1$ the triples $t$ assigned to $F$
from $G_1$ correspond to points $v_t$ inside $D_1$. The same statement is true for $G_2$. This removes any potential
interference with planarity when putting the two P3EMs together to form a P3EM for $G$.
%
%
\end{proof}

Thus to prove the existence of a P3EM of $G$ we will 
prove the existence  of such an assignment. It also
follows that $G$ has a P3EM iff each connected 
component of $G$ does.

Plane graphs $G_i= (V_i, E_i)$, $(i=1,2)$, are  planarly isomorphic if there exists an 1-1 correspondence  of $V_1$ and $V_2$ such that
it induces a 1-1 correspondence  of the edges and faces
by incidence.
Clearly, having a P3EM is a property preserved by planar isomorphism.

\begin{theorem}\label{thm:P3EM}
Every 3-regular plane graph, except for those containing a connected component
$K_4$ or the multi-graph  $M_{2,3}$ on 2 vertices with 3 parallel edges, admits a planar 
3-way edge matching, and one can be found in polynomial time. 
\end{theorem}

We note that P3EM indeed does not exist for the two exceptional graphs. Also, 
up to planar isomorphism there is only one
plane embedding for these two graphs, as well as
all graphs depicted in Figure~\ref{base case}, 
which will serve as our induction base cases.
In Figure~\ref{base case}, edges of the same color form a  triple. 

\begin{figure}[H]
  \centering
  \begin{subfigure}[b]{0.3\textwidth}
  \centering
  \begin{tikzpicture}[scale=0.5]
  \filldraw [black] (-1,0) circle (5pt);
  \filldraw [black] (1,0) circle (5pt);
  \draw[red] (-1,0) -- (1,0);
  \draw[red] (-1.5,0) circle (0.5);
  \draw[red] (1.5,0) circle (0.5);
  \filldraw [black] (0,-1) circle (0.01pt);
  \end{tikzpicture}
  \captionsetup{justification=centering}
  \caption{}
  \label{base case 3}
  \end{subfigure}
  \begin{subfigure}[b]{0.3\textwidth}
  \centering
  \begin{tikzpicture}[scale=0.5]
  \filldraw [black] (-2,0) circle (5pt);
  \filldraw [black] (2,0) circle (5pt);
  \filldraw[black] (0,1) circle (5pt);
  \filldraw (0,2) circle (5pt);
  \draw[red] (-2,0) -- (2,0);
  \draw[red] (-2,0) .. controls (-1.5,0.75) and (-0.75,1) .. (0,1);
  \draw[red] (2,0) .. controls (1.5,0.75) and (0.75,1) .. (0,1);
  \draw[blue] (-2,0) .. controls (-1,-1) and (1,-1) .. (2,0);
  \draw[blue] (0,1) -- (0,2);
  \draw[blue] (0,2.5) circle (0.5);
  \end{tikzpicture}
  \captionsetup{justification=centering}
  \caption{}
  \label{base case 4}
  \end{subfigure}
  \begin{subfigure}[b]{0.3\textwidth}
  \centering
  \begin{tikzpicture}[scale=0.5]
  \filldraw [black] (-3,0) circle (5pt);
  \filldraw [black] (3,0) circle (5pt);
  \filldraw[black] (0,1) circle (5pt);
  \filldraw (0,0) circle (5pt);
  \filldraw (0,3) circle (5pt);
  \filldraw (0,4.5) circle (5pt);
  \draw[red] (-3,0) -- (3,0);
  \draw[red] (3,0) -- (0,4.5);
  \draw[blue] (0,1.5) circle (0.5);
  \draw[blue] (0,0) -- (0,1);
  \draw[blue] (3,0) -- (0,3);
  \draw[orange] (-3,0) -- (0,3);)
  \draw[orange] (0,3) -- (0,4.5);
  \draw[orange] (-3,0) -- (0,4.5);
  \end{tikzpicture}
  \captionsetup{justification=centering}
  \caption{}
  \label{base case 5}
  \end{subfigure}
    \begin{subfigure}[b]{0.3\textwidth}
  \centering
  \begin{tikzpicture}[scale=0.5]
  \filldraw [black] (-3,0) circle (5pt);
  \filldraw [black] (-1,0) circle (5pt);
  \filldraw [black] (1,0) circle (5pt);
  \filldraw [black] (3,0) circle (5pt);
  \draw[red] (-3,0) -- (-1,0);
  \draw[red] (-1,0) .. controls (-0.5,0.75) and (0.5,0.75) .. (1,0);
  \draw[blue] (-1,0) .. controls (-0.5,-0.75) and (0.5,-0.75) .. (1,0);
  \draw[blue] (1,0) -- (3,0);
  \draw[red] (-3,0) .. controls (-1,2) and (1,2) .. (3,0);
  \draw[blue] (-3,0) .. controls (-1,-2) and (1,-2) .. (3,0);
  \end{tikzpicture}
  \captionsetup{justification=centering}
  \caption{}
  \label{base case 1}
  \end{subfigure}
  \begin{subfigure}[b]{0.3\textwidth}
  \centering
  \begin{tikzpicture}[scale=0.5]
  \filldraw [black] (0,0) circle (5pt);
  \filldraw [black] (-3,-1) circle (5pt);
  \filldraw [black] (3,-1) circle (5pt);
  \filldraw [black] (0,2) circle (5pt);
  \filldraw [black] (0,3) circle (5pt);
  \filldraw [black] (0,4) circle (5pt);
  \draw[red] (-3,-1) -- (0,4);
  \draw[red] (0,4) -- (0,3);
  \draw[red] (0,2) .. controls (-0.75,2) and (-0.75,3) .. (0,3);
  \draw[blue] (0,2) .. controls (0.75,2) and (0.75,3) .. (0,3);
  \draw[blue] (0,0) -- (0,2);
  \draw[blue] (0,4) -- (3,-1);
  \draw[orange] (-3,-1) -- (3,-1);
  \draw[orange] (-3,-1) -- (0,0);
  \draw[orange] (3,-1) -- (0,0);
  \end{tikzpicture}
  \captionsetup{justification=centering}
  \caption{}
  \label{base case 2}
  \end{subfigure}
  \begin{subfigure}[b]{0.3\textwidth}
  \centering
  \begin{tikzpicture}[scale=0.5]
  \filldraw [black] (0,1) circle (5pt);
  \filldraw [black] (-1,-1) circle (5pt);
  \filldraw [black] (1,-1) circle (5pt);
  \filldraw [black] (0,2.5) circle (5pt);
  \filldraw [black] (-3,-2) circle (5pt);
  \filldraw [black] (3,-2) circle (5pt);
  \draw[red] (0,1) -- (-1,-1);
  \draw[orange] (-1,-1) -- (1,-1);
  \draw[blue] (1,-1) -- (0,1);
  \draw[red] (0,1) -- (0,2.5);
  \draw[orange] (-1,-1) -- (-3,-2);
  \draw[blue] (1,-1) -- (3,-2);
  \draw[red] (-3,-2) -- (0,2.5);
  \draw[blue] (3,-2) -- (0,2.5);
  \draw[orange] (-3,-2) -- (3,-2);
  \end{tikzpicture}
  \captionsetup{justification=centering}
  \caption{}
  \label{base case 6}
  \end{subfigure}
  \begin{subfigure}[b]{0.3\textwidth}
  \centering
  \begin{tikzpicture}[scale=0.5]
  \filldraw [black] (0,1) circle (5pt);
  \filldraw [black] (0,-1) circle (5pt);
  \filldraw [black] (1,0) circle (5pt);
  \filldraw [black] (-1,0) circle (5pt);
  \filldraw [black] (3,0) circle (5pt);
  \filldraw [black] (-3,0) circle (5pt);
  \draw[red] (0,1) -- (0,-1);
  \draw[blue] (0,1) -- (1,0);
  \draw[orange] (1,0) -- (0,-1);
  \draw[red] (0,-1) -- (-1,0);
  \draw[red] (-1,0) -- (0,1);
  \draw[blue] (1,0) -- (3,0);
  \draw[orange] (-1,0) -- (-3,0);
  \draw[blue] (-3,0) .. controls (-1.5,2.5) and (1.5,2.5) .. (3,0);
  \draw[orange] (-3,0) .. controls (-1.5,-2.5) and (1.5,-2.5) .. (3,0);
  \end{tikzpicture}
  \captionsetup{justification=centering}
  \caption{}
  \label{base case 7}
  \end{subfigure}
  \begin{subfigure}[b]{0.3\textwidth}
  \centering
  \begin{tikzpicture}[scale=0.5]
  \filldraw [black] (0,1) circle (5pt);
  \filldraw [black] (0,-1) circle (5pt);
  \filldraw [black] (1,0) circle (5pt);
  \filldraw [black] (-1,0) circle (5pt);
  \filldraw [black] (3,0) circle (5pt);
  \filldraw [black] (-3,0) circle (5pt);
  \filldraw (0,2) circle (5pt);
  \filldraw (0,3.5) circle (5pt);
  \draw[red] (0,1) -- (0,-1);
  \draw[blue] (0,1) -- (1,0);
  \draw[orange] (1,0) -- (0,-1);
  \draw[red] (0,-1) -- (-1,0);
  \draw[red] (-1,0) -- (0,1);
  \draw[blue] (1,0) -- (3,0);
  \draw[orange] (-1,0) -- (-3,0);
  \draw[green] (-3,0) -- (0,2);
  \draw[blue] (3,0) -- (0,2);
  \draw[green] (0,2) -- (0,3.5);
  \draw[green] (-3,0) -- (0,3.5);
  \draw[orange] (3,0) -- (0,3.5);
  \end{tikzpicture}
  \captionsetup{justification=centering}
  \caption{}
  \label{base case 8}
  \end{subfigure}
  \captionsetup{justification=centering}
  \caption{Base cases in the proof of Theorem~\ref{thm:P3EM}}
  \label{base case}
\end{figure}
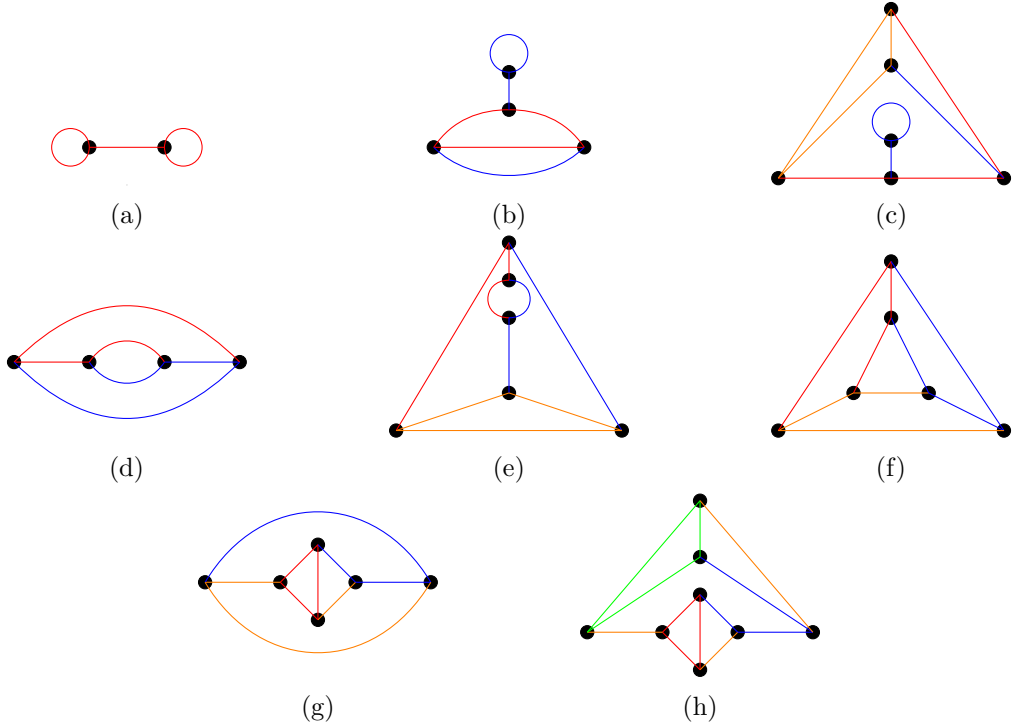

\begin{proof}
By Lemma~\ref{lem:P3EM-equivalce}, it suffices to prove the case when $G$ is connected, as putting together P3EMs for each connected component gives a P3EM for $G$.

We  prove Theorem~\ref{thm:P3EM} by induction.
Our induction hypothesis is as follows: if $|V(G)| \le k$ for some integer $k$ and it is not one of the two exceptions, then it admits a P3EM. 
Given a larger graph $G$, 
we try to reduce it to a smaller graph $G'$ such that we obtain
a P3EM for $G$ from that of $G'$;
whenever the reduction step produces one of
the two exceptions, we will give a P3EM directly to the original graph.

\begin{figure}[ht]
  \centering
  \begin{subfigure}[b]{0.4\textwidth}
  \centering
  \begin{tikzpicture}[scale=0.5]
  \filldraw [black] (0,0) circle (5pt);
  \node[below] at (0,0) {$A$};
  \filldraw [black] (1,0) circle (5pt);
  \node[below] at (1,0) {$B$};
  \filldraw [black] (2,1) circle (5pt);
  \node[above right] at (2,1) {$C$};
  \filldraw [black] (2,-1) circle (5pt);
  \node[below right] at (2,-1) {$D$};
  \draw (-0.5,0) circle (0.5);
  \draw (0,0) -- (1,0);
  \draw (1,0) -- (2,1);
  \draw (1,0) -- (2,-1);
  \draw (2,1) -- (2,2);
  \draw (2,1) -- (3,1);
  \draw (2,-1) -- (3,-1);
  \draw (2,-1) -- (2,-2);
  \draw[thick, dotted] (2,2) -- (2,3);
  \draw[thick, dotted] (3,1) -- (4,1);
  \draw[thick, dotted] (3,-1) -- (4,-1);
  \draw[thick, dotted] (2,-2) -- (2,-3);
  \node[left] at (-0.75,0) {$e_1$};
  \node[above] at (0.5,0) {$e_2$};
  \node[below right] at (1.25,1) {$e_4$};
  \node[above right] at (1.25,-1) {$e_3$};
  \end{tikzpicture}
  \captionsetup{justification=centering}
  \caption{self loop in the original graph}
  \end{subfigure}
  \begin{subfigure}[b]{0.4\textwidth}
  \centering
  \begin{tikzpicture}[scale=0.5]
  \filldraw [black] (0,1) circle (5pt);
  \node[left] at (0,1) {$C$};
  \filldraw [black] (0,-1) circle (5pt);
  \node[left] at (0,-1) {$D$};
  \draw (0,1) -- (0,-1);
  \draw (0,1) -- (1,1);
  \draw (0,1) -- (0,2);
  \draw (0,-1) -- (0,-2);
  \draw (0,-1) -- (1,-1);
  \draw[thick, dotted] (0,2) -- (0,3);
  \draw[thick, dotted] (1,1) -- (2,1);
  \draw[thick, dotted] (0,-2) -- (0,-3);
  \draw[thick, dotted] (1,-1) -- (2,-1);
  \node[right] at (0,0) {$e^*$};
  \end{tikzpicture}
  \captionsetup{justification=centering}
  \caption{resulting graph}
  \end{subfigure}
  \captionsetup{justification=centering}
  \caption{Transforming self-loops}
  \label{self loop}
\end{figure}
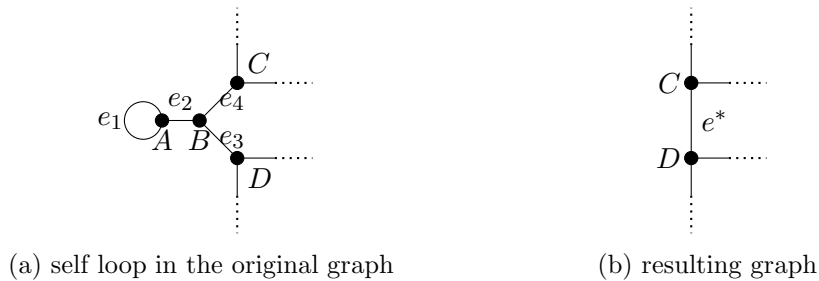

We first show it suffices to consider simple planar 3-regular graphs.

 If $G$ has a self-loop, then locally it has a fragment depicted in Figure~\ref{self loop}, unless it is 
 planarly isomorphic to the base case Figure~\ref{base case 3}, which we
 directly give a P3EM. Now perform the transformation depicted in Figure~\ref{self loop}. If the resulting graph is one of the two exceptions, then the original graph is planarly isomorphic to the base cases Figure~\ref{base case 4} or Figure~\ref{base case 5}, for which we give their P3EMs directly. 
 Otherwise, by induction hypothesis, there exists a P3EM $M$ for the resulting graph. If the edge $e^*$ in Figure~\ref{self loop}b is mapped rightwards ({\it resp.} leftwards) in $M$, then we obtain a P3EM for the original graph by simulating $e_4$ as $e^*$ to be mapped rightwards ({\it resp.} leftwards) and connecting $e_1$, $e_2$ and $e_3$.
Note that the self-loop transformation is valid regardless whether there is a self-loop at the vertices $C$ or $D$.

\begin{figure}[ht]
  \centering
  \begin{subfigure}[b]{0.4\textwidth}
  \centering
  \begin{tikzpicture}[scale=0.5]
  \draw[thick, dotted] (-4.5,1.5) -- (-4,1);
  \draw[thick, dotted] (-4.5,-1.5) -- (-4,-1);
  \draw (-4,1) -- (-3,0);
  \draw (-4,-1) -- (-3,0);
  \filldraw [black] (-3,0) circle (5pt);
  \node[below] at (-3,0) {$A$};
  \draw (-3,0) -- (-1,0);
  \node[above] at (-2,0) {$e_1$};
  \filldraw [black] (-1,0) circle (5pt);
  \node[below] at (-1,0) {$B$};
  \draw (-1,0) .. controls (-0.5,0.75) and (0.5,0.75) .. (1,0);
  \node[above] at (0,0.5) {$e_2$};
  \draw (-1,0) .. controls (-0.5,-0.75) and (0.5,-0.75) .. (1,0);
  \node[below] at (0,-0.5) {$e_3$};
  \filldraw [black] (1,0) circle (5pt);
  \node[below] at (1,0) {$C$};
  \draw (1,0) -- (3,0);
  \node[above] at (2,0) {$e_4$};
  \filldraw [black] (3,0) circle (5pt);
  \node[below] at (3,0) {$D$};
  \draw (3,0) -- (4,1);
  \draw (3,0) -- (4,-1);
  \draw[thick, dotted] (4,1) -- (4.5,1.5);
  \draw[thick, dotted] (4,-1) -- (4.5,-1.5);
  \end{tikzpicture}
  \captionsetup{justification=centering}
  \caption{double edges in the original graph}
  \end{subfigure}
  \begin{subfigure}[b]{0.4\textwidth}
  \centering
  \begin{tikzpicture}[scale=0.5]
  \draw[thick, dotted] (-2.5,1.5) -- (-2,1);
  \draw[thick, dotted] (-2.5,-1.5) -- (-2,-1);
  \draw (-2,1) -- (-1,0);
  \draw (-2,-1) -- (-1,0);
  \filldraw [black] (-1,0) circle (5pt);
  \node[below] at (-1,0) {$A$};
  \draw (-1,0) -- (1,0);
  \node[above] at (0,0) {$e^*$};
  \filldraw [black] (1,0) circle (5pt);
  \node[below] at (1,0) {$D$};
  \draw (1,0) -- (2,1);
  \draw (1,0) -- (2,-1);
  \draw[thick, dotted] (2,1) -- (2.5,1.5);
  \draw[thick, dotted] (2,-1) -- (2.5,-1.5);
  \end{tikzpicture}
  \captionsetup{justification=centering}
  \caption{resulting graph}
  \end{subfigure}
  \captionsetup{justification=centering}
  \caption{Transforming double edges}
  \label{double edge}
\end{figure}
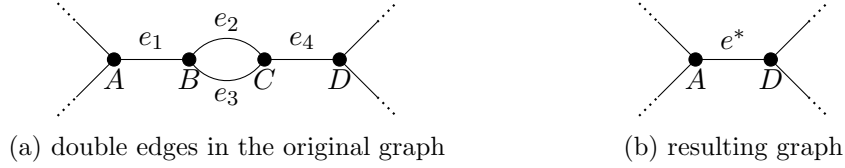

If $G$ has parallel edges between two vertices, say $B$ and $C$,
since  $G$ is 3-regular and 
$G$ is not planarly isomorphic to  $M_{2,3}$,
there must be exactly two edges between them.
If $B$ and $C$ have a common neighbor, say $A$, we may delete $B, C$ and their
incident edges, and add a self-loop at $A$. The resulting
graph is not one of the exceptions, by induction, it has a P3EM. One can then easily
obtain a P3EM for $G$.  
Now suppose the third edges from $B$ and $C$ are $e_1= \{A,B\}$ and $e_4 =\{C, D\}$
with $A \ne D$, as depicted in Figure~\ref{double edge}.
We likewise perform a transformation which ``deletes" $B$ and $C$  with double edges, 
and merge $e_1$ and $e_4$ to a single edge $e^* = \{A, D\}$.  If the resulting graph is one of the two exceptions, the original graph is isomorphic to the base cases Figure~\ref{base case 1} or Figure~\ref{base case 2} which we give their P3EMs directly. Otherwise, by induction hypothesis, there exists a P3EM $M$ for the resulting graph. If the edge $e^*$ is mapped upward ({\it resp.} downward) in $M$ (Figure~\ref{double edge}b), then we obtain a P3EM for the original graph by simulating $e_2$ ({\it resp.} $e_3$) as $e^*$ and connecting $e_3$ ({\it resp.} $e_2$), $e_1$ and $e_4$. 

Below we assume  $G$ is simple, i.e.,\ without parallel edges or self-loops.

\begin{figure}[ht]
  \centering
  \begin{subfigure}[b]{0.4\textwidth}
  \centering
  \begin{tikzpicture}[scale=0.5]
  \filldraw [black] (0,1) circle (5pt);
  \node[right] at (0,1) {$A$};
  \filldraw [black] (-1,-1) circle (5pt);
  \node[below] at (-1,-1) {$B$};
  \filldraw [black] (1,-1) circle (5pt);
  \node[below] at (1,-1) {$C$};
  \filldraw [black] (0,2.5) circle (5pt);
  \node[right] at (0,2.5) {$D$};
  \filldraw [black] (-2,-2) circle (5pt);
  \node[below left] at (-2,-2) {$E$};
  \filldraw [black] (2,-2) circle (5pt);
  \node[below right] at (2,-2) {$F$};
  \draw (0,1) -- (-1,-1);
  \draw (-1,-1) -- (1,-1);
  \draw (1,-1) -- (0,1);
  \draw (0,1) -- (0,2.5);
  \draw (-1,-1) -- (-2,-2);
  \draw (1,-1) -- (2,-2);
  \draw[thick, dotted] (-2,-2) -- (-3,-2);
  \draw[thick, dotted] (-2,-2) -- (-2,-3);
  \draw[thick, dotted] (2,-2) -- (3,-2);
  \draw[thick, dotted] (2,-2) -- (2,-3);
  \draw[thick, dotted] (0,2.5) -- (-0.5, 3.5);
  \draw[thick, dotted] (0,2.5) -- (0.5, 3.5);
  \node[left] at (0.25,1.75) {$e_1$};
  \node[above left] at (-1.25,-2) {$e_2$};
  \node[above right] at (1.25, -2) {$e_3$};
  \end{tikzpicture}
  \captionsetup{justification=centering}
  \caption{triangle in the original graph}
  \label{triangle 1}
  \end{subfigure}
  \begin{subfigure}[b]{0.4\textwidth}
  \centering
  \begin{tikzpicture}[scale=0.5]
  \filldraw [black] (0,0) circle (5pt);
  \filldraw [black] (0,1.5) circle (5pt);
  \node[right] at (0,1.5) {$D$};
  \filldraw [black] (-1,-1) circle (5pt);
  \node[below left] at (-1,-1) {$E$};
  \filldraw [black] (1,-1) circle (5pt);
  \node[below right] at (1,-1) {$F$};
  \draw (0,0) -- (0,1.5);
  \draw (0,0) -- (-1,-1);
  \draw (0,0) -- (1,-1);
  \draw[thick, dotted] (0,1.5) -- (-0.5,2.5);
  \draw[thick, dotted] (0,1.5) -- (0.5,2.5);
  \draw[thick, dotted] (-1,-1) -- (-2,-1);
  \draw[thick, dotted] (-1,-1) -- (-1,-2);
  \draw[thick, dotted] (1,-1) -- (2,-1);
  \draw[thick, dotted] (1,-1) -- (1,-2);
  \node[left] at (0.25, 0.75) {$e_1$};
  \node[above left] at (-0.25,-0.75) {$e_2$};
  \node[above right] at (0.25,-0.75) {$e_3$};
  \end{tikzpicture}
  \captionsetup{justification=centering}
  \caption{resulting graph}
  \label{triangle 2}
  \end{subfigure}
  \\
  \begin{subfigure}[b]{0.4\textwidth}
  \centering
  \begin{tikzpicture}[scale=0.5]
  \filldraw [black] (0,1) circle (5pt);
  \node[above] at (0,1) {$A$};
  \filldraw [black] (0,-1) circle (5pt);
  \node[below] at (0,-1) {$B$};
  \filldraw [black] (1,0) circle (5pt);
  \node[below] at (1,0) {$C$};
  \filldraw [black] (-1,0) circle (5pt);
  \node[below] at (-1,0) {$D$};
  \filldraw [black] (2,0) circle (5pt);
  \node[below] at (2,0) {$F$};
  \filldraw [black] (-2,0) circle (5pt);
  \node[below] at (-2,0) {$E'$};
  \draw (0,1) -- (0,-1);
  \draw (0,1) -- (1,0);
  \draw (1,0) -- (0,-1);
  \draw (0,-1) -- (-1,0);
  \draw (-1,0) -- (0,1);
  \draw (1,0) -- (2,0);
  \draw (-1,0) -- (-2,0);
  \draw[thick, dotted] (2,0) -- (2.75,0.75);
  \draw[thick, dotted] (2,0) -- (2.75,-0.75);
  \draw[thick, dotted] (-2,0) -- (-2.75, 0.75);
  \draw[thick, dotted] (-2,0) -- (-2.75, -0.75);
  \end{tikzpicture}
  \captionsetup{justification=centering}
  \caption{triangle in the original graph ($D=E$)}
  \label{triangle 3}
  \end{subfigure}
  \begin{subfigure}[b]{0.4\textwidth}
  \centering
  \begin{tikzpicture}[scale=0.5]
  \filldraw [black] (1,0) circle (5pt);
  \node[below] at (1,0) {$F$};
  \filldraw [black] (-1,0) circle (5pt);
  \node[below] at (-1,0) {$E'$};
  \draw (-1,0) -- (1,0);
  \draw[thick, dotted] (1,0) -- (2,1);
  \draw[thick, dotted] (1,0) -- (2,-1);
  \draw[thick, dotted] (-1,0) -- (-2,1);
  \draw[thick, dotted] (-1,0) -- (-2,-1);
  \end{tikzpicture}
  \captionsetup{justification=centering}
  \caption{resulting graph ($D=E$)}
  \label{triangle 4}
  \end{subfigure}
  \captionsetup{justification=centering}
  \caption{Transforming triangles}
  \label{triangle}
\end{figure}
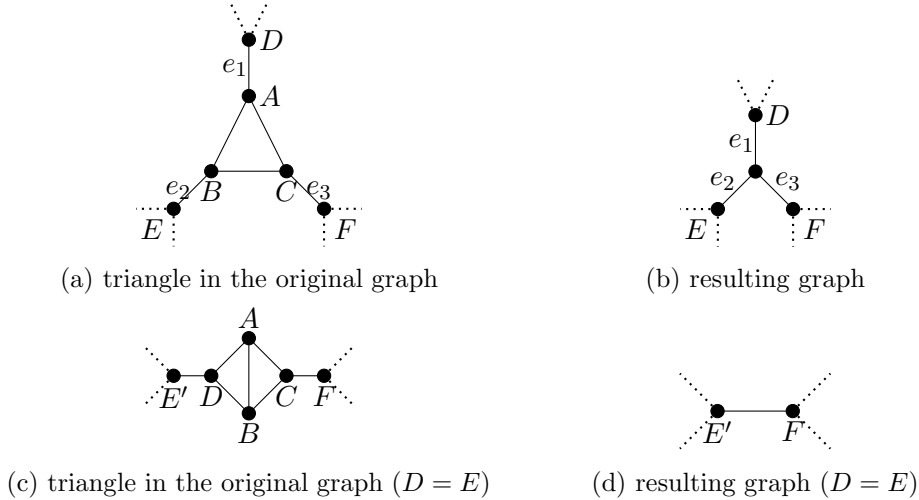

Next we consider the case when $G$ contains 
a triangle face as depicted in  Figure~\ref{triangle 1}. Since it is a simple graph, all three edges $e_1, e_2, e_3$ are distinct, and
$D, E, F \not \in \{A, B, C\}$. If the vertices $D, E, F$ are all distinct,  then we perform the transformation from Figure~\ref{triangle 1} to Figure~\ref{triangle 2}. By induction hypothesis, the resulting graph admits a P3EM unless it is $K_4$ (in this case, the resulting graph cannot be $M_{2,3}$ since it has more than two vertices). If the resulting graph is $K_4$, then the original graph is (or planarly isomorphic to)
the base case Figure~\ref{base case 6} for which we give a P3EM directly. If the resulting graph is not $K_4$ and hence admits a P3EM $M$ by our induction hypothesis, then the original graph can simulate $M$ by connecting the edges of the triangle face internally. We now consider the case when the vertices $D,E,F$ are not all distinct. Since the original graph is not $K_4$,  the vertices $D,E,F$ are not all the same vertex. Without loss of generality we assume $D=E\neq F$. See Figure~\ref{triangle 3} for an illustration. 
$D$ has an incident vertex $E' \ne A, B$.  Suppose $E' \ne F$. Then we perform the transformation illustrated in Figure~\ref{triangle 4}. Similarly as above, if the resulting graph is not one of the exceptions, then we can easily simulate the P3EM (which is given by the induction) in the resulting graph.
If the resulting graph Figure~\ref{triangle 4} is planarly isomorphic to the exceptions $M_{2,3}$ or $K_4$, then the original graph is planarly isomorphic to the base cases Figure~\ref{base case 7} or Figure~\ref{base case 8}, which we give a P3EM directly. Finally suppose $E'=F$, then we transform the original graph by deleting  the vertices  $A, B, C, D$ 
and their incident edges, and form a self-loop at $E'= F$. The resulting graph has fewer vertices
(and has a self-loop and so it is not $M_{2,3}$ or $K_4$), and so
by the induction hypothesis 
 it admits a P3EM. It is easy to verify that a P3EM in the resulting graph, as before, can be simulated in the original graph. 
In the following we may  assume $G$ is simple without triangle faces.

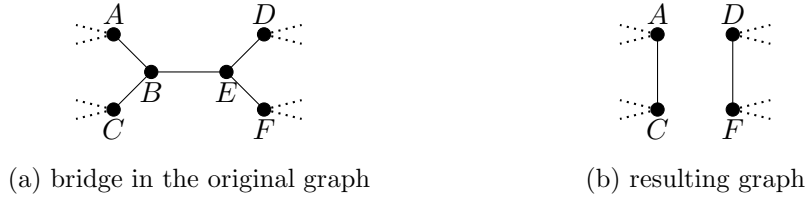
\begin{figure}[ht]
  \centering
  \begin{subfigure}[b]{0.4\textwidth}
  \centering
  \begin{tikzpicture}[scale=0.5]
  \filldraw [black] (-1,0) circle (5pt);
  \node[below] at (-1,0) {$B$};
  \filldraw [black] (-2,1) circle (5pt);
  \node[above] at (-2,1) {$A$};
  \filldraw [black] (-2,-1) circle (5pt);
  \node[below] at (-2,-1) {$C$};
  \filldraw [black] (1,0) circle (5pt);
  \node[below] at (1,0) {$E$};
  \filldraw [black] (2,1) circle (5pt);
  \node[above] at (2,1) {$D$};
  \filldraw [black] (2,-1) circle (5pt);
  \node[below] at (2,-1) {$F$};
  \draw (-1,0) -- (1,0);
  \draw (-1,0) -- (-2,1);
  \draw (-1,0) -- (-2,-1);
  \draw (1,0) -- (2,1);
  \draw (1,0) -- (2,-1);
  \draw[thick, dotted] (-2,1) -- (-3,1.25);
  \draw[thick, dotted] (-2,1) -- (-3,0.75);
  \draw[thick, dotted] (-2,-1) -- (-3, -1.25);
  \draw[thick, dotted] (-2,-1) -- (-3,-0.75);
  \draw[thick, dotted] (2,1) -- (3,1.25);
  \draw[thick, dotted] (2,1) -- (3,0.75);
  \draw[thick, dotted] (2,-1) -- (3,-0.75);
  \draw[thick, dotted] (2,-1) -- (3,-1.25);
  \end{tikzpicture}
  \captionsetup{justification=centering}
  \caption{bridge in the original graph}
  \end{subfigure}
  \begin{subfigure}[b]{0.4\textwidth}
  \centering
  \begin{tikzpicture}[scale=0.5]
  \filldraw [black] (-1,1) circle (5pt);
  \node[above] at (-1,1) {$A$};
  \filldraw [black] (-1,-1) circle (5pt);
  \node[below] at (-1,-1) {$C$};
  \filldraw [black] (1,1) circle (5pt);
  \node[above] at (1,1) {$D$};
  \filldraw [black] (1,-1) circle (5pt);
  \node[below] at (1,-1) {$F$};
  \draw (-1,-1) -- (-1,1);
  \draw (1,1) -- (1,-1);
  \draw[thick, dotted] (-1,1) -- (-2, 1.25);
  \draw[thick, dotted] (-1,1) -- (-2,0.75);
  \draw[thick, dotted] (-1,-1) -- (-2, -0.75);
  \draw[thick, dotted] (-1,-1) -- (-2, -1.25);
  \draw[thick, dotted] (1,1) -- (2,1.25);
  \draw[thick, dotted] (1,1) -- (2,0.75);
  \draw[thick, dotted] (1,-1) -- (2,-0.75);
  \draw[thick, dotted] (1,-1) -- (2,-1.25);
  \end{tikzpicture}
  \captionsetup{justification=centering}
  \caption{resulting graph}
  \end{subfigure}
  \captionsetup{justification=centering}
  \caption{Transforming a bridge}
  \label{bridge}
\end{figure}

Next we consider the case when the graph contains a  bridge, i.e.,
an edge whose removal disconnects the graph (see Figure~\ref{bridge}a). 
This means that the same face $f$  is on both sides of the bridge
$\{B,E\}$ in $G$. 
Perform the transformation illustrated in Figure~\ref{bridge}. The resulting graph has two disconnected components which are not isomorphic to any exception case. Indeed, neither is isomorphic to $M_{2,3}$ since $G$ has no
parallel edges, and if it were
$K_4$ then the original graph $G$ contains a triangle face. 
Thus by induction there are P3EMs, $M_1$ and $M_2$, for the two components respectively. We will use the  edges $\{B,C\}$ and $\{E,F\}$ 
to simulate $\{A,C\}$ and $\{D,F\}$ respectively, and  match
the three edges $\{A,B\}, \{B,E\}$ and $\{E,D\}$ directly. Then we obtain a
P3EM for $G$ from $M_1$ and $M_2$. 
Note that
both $\{B,C\}$ and $\{E,F\}$ are on the face $f$.
Since  $\{B, E\}$ is a bridge, regardless of how  $\{A,C\}$ and $\{D,F\}$ are matched respectively by $M_1$ and $M_2$,
we can substitute $\{B,C\}$ and $\{E,F\}$ for them respectively.
When viewed in a spherical
embedding, we may assume both  $\{A,C\}$ and $\{D,F\}$ are on the outer face
for the two disconnected components. 
 Then  in $G$ the substitution of   $\{B,C\}$ for $\{A,C\}$,
 and  $\{E,F\}$ 
for $\{D,F\}$, gives 
 a total number of edges assigned to the face $f$ in $G$
to be the sum of the corresponding numbers assigned by $M_1$ and $M_2$, and thus this total number is 
$\equiv 0 \pmod{3} $. Below we assume the graph has no bridges.

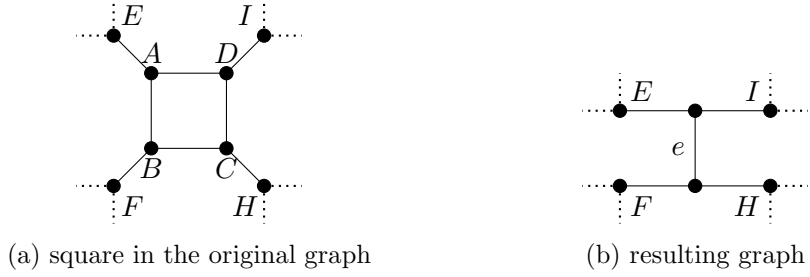
\begin{figure}[ht]
  \centering
  \begin{subfigure}[b]{0.4\textwidth}
  \centering
  \begin{tikzpicture}[scale=0.5]
  \filldraw [black] (-1,1) circle (5pt);
  \node[above] at (-1,1) {$A$};
  \filldraw [black] (-1,-1) circle (5pt);
  \node[below] at (-1,-1) {$B$};
  \filldraw [black] (1,-1) circle (5pt);
  \node[below] at (1,-1) {$C$};  
  \filldraw [black] (1,1) circle (5pt);
  \node[above] at (1,1) {$D$};
  \filldraw [black] (-2,2) circle (5pt);
  \node[above] at (-1.5,2) {$E$};
  \filldraw [black] (-2,-2) circle (5pt);
  \node[below] at (-1.5,-2) {$F$};
  \filldraw [black] (2,-2) circle (5pt);
  \node[below] at (1.5,-2) {$H$};
  \filldraw [black] (2,2) circle (5pt);
  \node[above] at (1.5,2) {$I$};
  \draw (-1,1) -- (-1,-1);
  \draw (-1,-1) -- (1,-1);
  \draw (1,-1) -- (1,1);
  \draw (1,1) -- (-1,1);
  \draw (-1,1) -- (-2,2);
  \draw (-1,-1) -- (-2,-2);
  \draw (1,-1) -- (2,-2);
  \draw (1,1) -- (2,2);
  \draw[thick, dotted] (-2,2) -- (-3,2);
  \draw[thick, dotted] (-2,2) -- (-2,3);
  \draw[thick, dotted] (-2,-2) -- (-3,-2);
  \draw[thick, dotted] (-2,-2) -- (-2,-3);
  \draw[thick, dotted] (2,-2) -- (3,-2);
  \draw[thick, dotted] (2,-2) -- (2,-3);
  \draw[thick, dotted] (2,2) -- (3,2);
  \draw[thick, dotted] (2,2) -- (2,3);
  \end{tikzpicture}
  \captionsetup{justification=centering}
  \caption{square in the original graph}
  \end{subfigure}
  \begin{subfigure}[b]{0.4\textwidth}
  \centering
  \begin{tikzpicture}[scale=0.5]
  \filldraw [black] (0,1) circle (5pt);
  \filldraw [black] (0,-1) circle (5pt);
  \filldraw [black] (2,1) circle (5pt);
  \node[above left] at (2,1) {$I$};
  \filldraw [black] (2,-1) circle (5pt);
  \node[below left] at (2,-1) {$H$};
  \filldraw [black] (-2,1) circle (5pt);
  \node[above right] at (-2,1) {$E$};
  \filldraw [black] (-2,-1) circle (5pt);
  \node[below right] at (-2,-1) {$F$};
  \draw (-2,1) -- (2,1);
  \draw (-2,-1) -- (2,-1);
  \draw (0,1) -- (0,-1);
  \node[left] at (0,0) {$e$};
  \draw[thick, dotted] (-2,1) -- (-3,1);
  \draw[thick, dotted] (-2,1) -- (-2,2);
  \draw[thick, dotted] (-2,-1) -- (-3,-1);
  \draw[thick, dotted] (-2,-1) -- (-2,-2);
  \draw[thick, dotted] (2,1) -- (3,1);
  \draw[thick, dotted] (2,1) -- (2,2);
  \draw[thick, dotted] (2,-1) -- (3,-1);
  \draw[thick, dotted] (2,-1) -- (2,-2);
  \end{tikzpicture}
  \captionsetup{justification=centering}
  \caption{resulting graph}
  \end{subfigure}
  \captionsetup{justification=centering}
  \caption{Transforming a square}
  \label{square}
\end{figure}

We show next that if $G$ has a square face, then it admits a P3EM. 
See Figure~\ref{square} for an illustration. Since it is simple,
$E$ is distinct from $B, D$. Also $E \ne C$, for otherwise
$\{B, F\}$ or $\{D, I\}$ would be a bridge.
By the same reason, none of the vertices $E, F, H, I$ can be from $\{A, B, C, D\}$.
If $E = H$ we again would have a bridge.
Also $E \ne F, I$, because $G$ has no triangle face.
It follows that all $A, B, C, D, E, F, H, I$ are distinct.
Now we perform the transformation in Figure~\ref{square}, replacing the
square $ABCD$ by an edge $e$.  It is clearly not one of the exception graphs 
(it has  at least 6 vertices).
By induction hypothesis, the resulting graph admits a P3EM $M$. If $e$ in the resulting graph is mapped leftwards,
then  we can simulate $M$ in $G$ by mapping $\{A,B\}$ leftwards, and 
match $\{A,D\}$, $\{D,C\}$ and $\{C,B\}$ inside the square.
Similarly, if $e$  is mapped rightwards, then we use $\{D, C\}$ in its place,
and match $\{A,D\}$, $\{A, B\}$ and $\{B, C\}$ inside the square.
This gives a P3EM for $G$. Below we assume $G$ has no square faces.

\begin{figure}[ht]
  \centering
  \begin{subfigure}[b]{0.4\textwidth}
  \centering
  \begin{tikzpicture}[scale=0.5]
  \filldraw [black] (0,1.5) circle (5pt);
  \node[above] at (0,1.5) {$A$};
  \filldraw [black] (0,-1.5) circle (5pt);
  \node[below] at (0,-1.5) {$B$};
  \filldraw [black] (-1.5,1.5) circle (5pt);
  \node[above] at (-1.5,1.5) {$C$};
  \filldraw [black] (-1.5,-1.5) circle (5pt);
  \node[below] at (-1.5,-1.5) {$D$};
  \filldraw [black] (1.5,1.5) circle (5pt);
  \node[above] at (1.5,1.5) {$E$};
  \filldraw [black] (1.5,-1.5) circle (5pt);
  \node[below] at (1.5,-1.5) {$F$};
  \draw (0,1.5) -- (0,-1.5);
  \draw (0,1.5) -- (-1.5,1.5);
  \draw (0,1.5) -- (1.5,1.5);
  \draw (0,-1.5) -- (-1.5,-1.5);
  \draw (0,-1.5) -- (1.5, -1.5);
  \draw[thick, dotted] (-1.5,1.5) .. controls (-3,1) and (-3,-1) .. (-1.5,-1.5);
  \draw[thick, dotted] (1.5,1.5) .. controls (3,1) and (3,-1) .. (1.5,-1.5);
  \node at (-1.25,0) {$1$};
  \node at (1.25,0) {$2$};
  \end{tikzpicture}
  \captionsetup{justification=centering}
  \caption{chord in the original graph}
  \label{chord_1}
  \end{subfigure}
  \begin{subfigure}[b]{0.4\textwidth}
  \centering
  \begin{tikzpicture}[scale=0.5]
  \filldraw [black] (0,0) circle (5pt);
  \node[left] at (0,-0.5) {$E'$};
  \filldraw [black] (0,1.5) circle (5pt);
  \node[above] at (0,1.5) {$A$};
  \filldraw [black] (0,-1.5) circle (5pt);
  \node[below] at (0,-1.5) {$B$};
  \filldraw [black] (-1.5,1.5) circle (5pt);
  \node[above] at (-1.5,1.5) {$C$};
  \filldraw [black] (-1.5,-1.5) circle (5pt);
  \node[below] at (-1.5,-1.5) {$D$};
  \filldraw [black] (1.5,0) circle (5pt);
  \node[below] at (1.5,0) {$F'$};
  \filldraw [black] (3,-1.5) circle (5pt);
  \node[below] at (3,-1.5) {$F$};
  \filldraw [black] (3,1.5) circle (5pt);
  \node[above] at (3,1.5) {$E$};
  \draw (0,0) -- (0,1.5);
  \draw (0,0) -- (0,-1.5);
  \draw (0,1.5) -- (-1.5,1.5);
  \draw (0,-1.5) -- (-1.5,-1.5);
  \draw (0,-1.5) -- (1.5,0);
  \draw (0,0) -- (1.5, 0);
  \draw (0,1.5) -- (1.5,0);
  \draw[thick, dotted] (-1.5,1.5) .. controls (-3,0.75) and (-3,-0.75) .. (-1.5,-1.5);
  \draw (3,1.5) -- (3,-1.5);
  \draw[thick, dotted] (3,1.5) .. controls (4, 0.75) and (4,-0.75) .. (3,-1.5);
  \end{tikzpicture}
  \captionsetup{justification=centering}
  \caption{resulting graph}
  \end{subfigure}
  \captionsetup{justification=centering}
  \caption{Transforming a chord}
  \label{chord}
\end{figure}
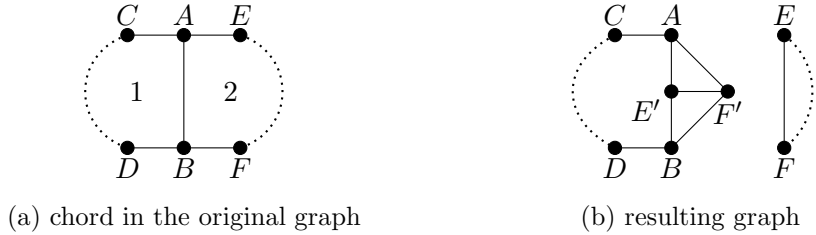

We now consider the case when the graph contains a chord. Let $\mathcal{C}$ be the boundary of the external face of the plane graph. Since we can now assume $G$ is bridgeless, $\mathcal{C}$ is a simple cycle. We say it contains a chord  if there exist two vertices on $\mathcal{C}$ that are joined by an edge that is not in $\mathcal{C}$. See Figure~\ref{chord}a for an illustration. 
Since  $\mathcal{C}$  is the outer boundary,
any chord must connect inside of $\mathcal{C}$.
Let $\{A, B\}$ be a chord, and let $C, E$ and $D, F$ be their neighbors on $\mathcal{C}$.
We note that there is no edge connecting $\{C,D\}$ or $\{E,F\}$ since  $G$ has no square
face. In Figure~\ref{chord}a 
we mark the part of $G$ to the left, respectively to the right, of (but including) the edge $\{A,B\}$ 
as region 1, respectively  region 2. 
By planarity, the only edges connecting 
 regions 1 and  2 are those incident to $A$ or $B$. 
Denote the numbers of edges in regions 1 and 2 by $E_1$ and $E_2$
(both including $\{A,B\}$). Then $E_1 + E_2  \equiv 1 \pmod{3}$. Perform the transformation illustrated in Figure~\ref{chord} and we obtain a resulting graph that is disconnected and its two components are 3-regular plane graphs. Since in the original graph, $E$ and $F$ are not adjacent, there exists another vertex distinct from $E$ and $F$ in region 2. Thus, the left side component in the resulting graph has at least one vertex fewer than $G$, and by induction hypothesis it admits a P3EM $M'_1$. Similarly, the right side component also admits a P3EM $M'_2$ (it cannot be $M_{2,3}$ since $G$ is simple,
 nor  $K_4$ since that will imply $G$ has a triangle face). 
The edge $\{E',F'\}$ is 
assigned either inside the triangle $AE'F'$ or inside $BE'F'$.
In the former case, 
the edges $\{B, E'\}$ and $\{B, F'\}$ must be assigned outside the triangle $BE'F'$.
In the latter case,  the edges $\{A, E'\}$ and $\{A, F'\}$ are assigned outside the triangle $AE'F'$. In either case, exactly one of the edges $\{A, E'\}$ or $\{B, E'\}$ is assigned leftwards and one of the edges $\{A, F'\}$ or $\{B, F'\}$ is assigned outside. Thus, there are $N_1 \equiv 2 \pmod{3}$ edges in the path  $(A,C,\dots, D, B)$ along the cycle  $\mathcal{C}$ assigned outside. Similarly, let $N_2$ denote the number of edges assigned outside along the external face in $M'_2$, then $N_2 \equiv 0 \pmod{3}$.
We now construct a 3DEM in $G$. In region 1, the
edge $\{A,B\}$ will be assigned leftwards (taking its place as either
$\{A, E'\}$ or $\{B, E'\}$ which was assigned leftwards in $M'_1$).
All other edges will be assigned
in the same way as in $M'_1$. 
 In region 2 we assign all edges, other than $\{A,B\}$, as follows.
 The edge $\{B,F\}$ will be assigned outside (taking its place as either
 $\{A, F'\}$ or $\{B, F'\}$ which was assigned outside in $M'_1$);
 $\{A, E\}$ will be assigned as $\{E, F\}$ in $M'_2$;
all other edges 
will be assigned the same way as in $M'_2$. By doing so, all internal faces
are assigned  $0 \pmod{3}$ edges, as in the case of $M'_1$ and $M'_2$. For the external face, note that there are also $N_1 + 1 + N_2 \equiv 0 \pmod{3}$ edges assigned to it in total. Thus our construction gives a valid P3EM in $G$.

To summarize, we can assume now that the  3-regular plane graph $G$ is  simple, without triangle and square faces, bridgeless and  chordless. We now show that the graph must have a pentagon face. Let $v,e,f$ denote the number of vertices, edges and faces (including the external one) of $G$, respectively. Since $G$ is 3-regular, we have $3v=2e$. Suppose the minimum number of edges around any face is $n$, then $2e \geq nf$. 
By Euler's formula, we have $2 = v-e+f \le (2/n - 1/3) e$, and thus
$n < 6$. Since $G$ is simple and
without triangle and square faces, we have $n=5$.

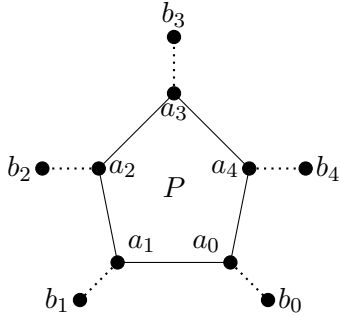
\begin{figure}[ht]
  \centering
  \begin{tikzpicture}[scale=0.5]
  \node at (0,0.5) {$P$};
  \filldraw [black] (1.5,-1.5) circle (5pt);
  \node[above left] at (1.5,-1.5) {$a_0$};
  \filldraw [black] (-1.5,-1.5) circle (5pt);
  \node[above right] at (-1.5,-1.5) {$a_1$};
  \filldraw [black] (-2,1) circle (5pt);
  \node[right] at (-2,1) {$a_2$};
  \filldraw [black] (2,1) circle (5pt);
  \node[left] at (2,1) {$a_4$};
  \filldraw [black] (0,3) circle (5pt);
  \node[below] at (0,3) {$a_3$};
  \filldraw [black] (2.5,-2.5) circle (5pt);
  \node[right] at (2.5,-2.5) {$b_0$};
  \filldraw [black] (-2.5,-2.5) circle (5pt);
  \node[left] at (-2.5,-2.5) {$b_1$};
  \filldraw [black] (-3.5,1) circle (5pt);
  \node[left] at (-3.5,1) {$b_2$};
  \filldraw [black] (3.5,1) circle (5pt);
  \node[right] at (3.5,1) {$b_4$};
  \filldraw [black] (0,4.5) circle (5pt);
  \node[above] at (0,4.5) {$b_3$};
  \draw (1.5,-1.5) -- (-1.5,-1.5);
  \draw (-1.5, -1.5) -- (-2,1);
  \draw (-2,1) -- (0,3);
  \draw (0,3) -- (2,1);
  \draw (2,1) -- (1.5,-1.5);
  \draw[thick, dotted] (1.5,-1.5) -- (2.5,-2.5);
  \draw[thick, dotted] (-1.5,-1.5) -- (-2.5,-2.5);
  \draw[thick, dotted] (0,3) -- (0,4.5);
  \draw[thick, dotted] (-2,1) -- (-3.5,1);
  \draw[thick, dotted] (2,1) -- (3.5,1);
  \end{tikzpicture}  
  \captionsetup{justification=centering}
   \caption{Pentagon in the original graph}
   \label{pentagon_path_1}
\end{figure}

Now fix a pentagon face $P$ in the graph
with vertices $\{a_0, a_1, \ldots, a_4\}$. See Figure~\ref{pentagon_path_1} for an illustration. Since $G$ is 3-regular, simple and bridgeless,
there is a neighbor
$b_i$   of $a_i$ distinct from $\{a_0, a_1, \ldots, a_4\}$, for every $0\le i \le 4$.
For example,  $b_0 \ne a_1$ by simplicity.
If $b_0=a_2$, i.e., if $\{a_0,a_2\}$ were an edge, then
there would be a bridge $\{a_1, b_1\}$.
Indeed, the edge $\{a_0,a_2\}$ must lie outside of the pentagon face $P$.
If one traverses the edges $\{a_2,a_1\}, \{a_1,a_0\}$ with $P$ to its left,
then follows with the edge $\{a_0,a_2\}$, one gets a cycle which separates the part
of $G$ that contains $P$ from the part of $G$ that connects to $a_1$ via the edge $\{a_1,b_1\}$ (note that all three neighbors of each vertex in $\{a_0,a_1,a_2\}$ are accounted for and thus no other adjacent edge exists to the right of this cycle). So,
deleting $\{a_1,b_1\}$  disconnects $G$, and thus
$\{a_1,b_1\}$  is a bridge. Thus, $b_0 \not =a_2$ as $G$ is bridgeless. See Figure~\ref{fig:b0neqa2} for an illustration.
By symmetry, $b_0 \ne a_3, a_4$ as well.
By the same reason, 
$b_i \not \in \{a_0, a_1, \ldots, a_4\}$,
 for all $0\le i \le 4$.

 \begin{figure}[ht]
  \centering
  \begin{tikzpicture}[scale=0.5]
  \node at (0,0.5) {$P$};
  \filldraw [black] (1.5,-1.5) circle (5pt);
  \node[above left] at (1.5,-1.5) {$a_0$};
  \filldraw [black] (-1.5,-1.5) circle (5pt);
  \node[above right] at (-1.5,-1.5) {$a_1$};
  \filldraw [black] (-2,1) circle (5pt);
  \node[right] at (-2,1) {$a_2$};
  \filldraw [black] (2,1) circle (5pt);
  \node[left] at (2,1) {$a_4$};
  \filldraw [black] (0,3) circle (5pt);
  \node[below] at (0,3) {$a_3$};
  \filldraw [black] (-2.5,-2.5) circle (5pt);
  \node[above] at (-2.5,-2.5) {$b_1$};
  \draw (1.5,-1.5) -- (-1.5,-1.5);
  \draw (-1.5, -1.5) -- (-2,1);
  \draw (-2,1) -- (0,3);
  \draw (0,3) -- (2,1);
  \draw (2,1) -- (1.5,-1.5);
  \draw (-1.5,-1.5) -- (-2.5,-2.5);
  \draw[thick, dotted] (-2.5,-2.5) -- (-2.5,-3.5);
  \draw[thick, dotted] (-2.5,-2.5) -- (-3.5,-2.5);
  \draw (-2,1) .. controls (0,8) and (5,1) .. (1.5,-1.5);
  \end{tikzpicture}
  \captionsetup{justification=centering}
   \caption{Why $b_0 \neq a_2$}
   \label{fig:b0neqa2}
\end{figure}
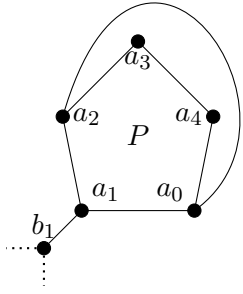

Furthermore, since $G$ has no triangle or square
faces
 and is bridgeless, we claim that without loss of generality
$b_i$ are all distinct ($ 0 \le i\le 4$). To see that, we first note that
 $b_0 \ne b_1, b_4$, for otherwise
there would be a triangle face.
Next we deal with the case  $b_0 = b_2$ or $b_0 = b_3$.
By symmetry suppose $b_0 = b_2$. There is
a third adjacent vertex $b$ of $b_0$, other than $a_0, a_2$. Consider the
cycle $C = (a_2, a_1, a_0, b_0=b_2, a_2)$ with
$P$ to its left, which defines two simply connected
regions in the spherical embedding of $G$. Call the region that contains $P$
the  interior region. If the edge $\{b_0, b\}$
is in the interior region, then $\{a_1, b_1\}$
is a bridge, by the same proof for $b_0 \ne a_2$. So we may assume $\{b_0, b\}$
lies in the exterior region of  the cycle $C$. (See Figure~\ref{trouble}.)

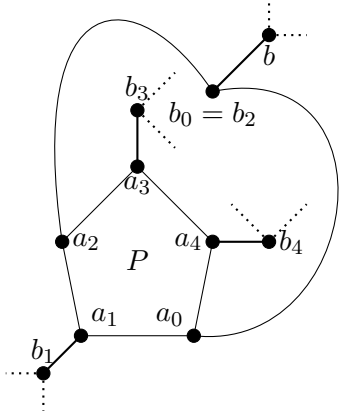
\begin{figure}[ht]
  \centering
  \begin{tikzpicture}[scale=0.5]
  \node at (0,0.5) {$P$};
  \filldraw [black] (1.5,-1.5) circle (5pt);
  \node[above left] at (1.5,-1.5) {$a_0$};
  \filldraw [black] (-1.5,-1.5) circle (5pt);
  \node[above right] at (-1.5,-1.5) {$a_1$};
  \filldraw [black] (-2,1) circle (5pt);
  \node[right] at (-2,1) {$a_2$};
  \filldraw [black] (2,1) circle (5pt);
  \node[left] at (2,1) {$a_4$};
  \filldraw [black] (0,3) circle (5pt);
  \node[below] at (0,3) {$a_3$};
  \filldraw [black] (-2.5,-2.5) circle (5pt);
  \node[above] at (-2.5,-2.5) {$b_1$};
  \filldraw [black] (3.5,1) circle (5pt);
  \node[right] at (3.5,1) {$b_4$};
  \filldraw [black] (0,4.5) circle (5pt);
  \node[above] at (0,4.5) {$b_3$};
  \filldraw [black] (2,5) circle (5pt);
  \node[below] at (2,5) {$b_0=b_2$};
  \filldraw [black] (3.5,6.5) circle (5pt);
  \node[below] at (3.5,6.5) {$b$};
  \draw (1.5,-1.5) -- (-1.5,-1.5);
  \draw (-1.5, -1.5) -- (-2,1);
  \draw (-2,1) -- (0,3);
  \draw (0,3) -- (2,1);
  \draw (2,1) -- (1.5,-1.5);
  \draw[thick] (-1.5,-1.5) -- (-2.5,-2.5);
  \draw[thick] (0,3) -- (0,4.5);
  \draw[thick] (2,1) -- (3.5,1);
  \draw (-2,1) .. controls (-3,8) and (0,8) .. (2,5);
  \draw (1.5,-1.5) .. controls (6,-2) and (7,6) .. (2,5);
  
  \draw[thick] (2,5) -- (3.5,6.5);
  \draw[thick, dotted] (0,4.5) -- (1,5.5);
  \draw[thick, dotted] (0,4.5) -- (1,3.5);
  \draw[thick, dotted] (3.5,1) -- (2.5,2);
  \draw[thick, dotted] (3.5,1) -- (4.5,2);
  \draw[thick, dotted] (-2.5,-2.5) -- (-3.5,-2.5);
  \draw[thick, dotted] (-2.5,-2.5) -- (-2.5,-3.5);
  \draw[thick, dotted] (3.5,6.5) -- (3.5,7.5);
  \draw[thick, dotted] (3.5,6.5) -- (4.5,6.5);
  \end{tikzpicture}  
  \captionsetup{justification=centering}
   \caption{$\{b_0, b\}$
lies in the exterior region of  the cycle $C$}
   \label{trouble}
\end{figure}

Clearly $b \ne b_1$, for otherwise there is a  
square face.
Now there is a face $\Delta$ bounded by the cycle that
contains the edge $\{a_1, a_0\}$ on the opposite
side of $P$. The bounding cycle contains the
path $(b_1, a_1, a_0, b_0, b)$, followed by a path  
$\pi = (b=x_0, x_1, \ldots, x_\ell=b_1)$ of $\ell \ge 1$
edges
from $b$ back to $b_1$. Here the first edge $\{b, x_1\}$
is  the right branch we take when we go from
$b_0$ to $b$,
and the last edge $\{x_{\ell-1}, b_1\}$ 
is the left branch we take if we go from
$a_1$ to $b_1$.
Similarly there is another face $\Delta'$ bounded by the cycle that
contains the edge $\{a_1, a_2\}$ on the opposite
side of $P$. The bounding cycle contains the
path $b_1, a_1, a_2, b_0, b$ followed by 
a path $\pi'$ of $\ell' \ge 1$
edges from $b$ back to $b_1$. (See Figure~\ref{G-two-bounding-faces}.) 

 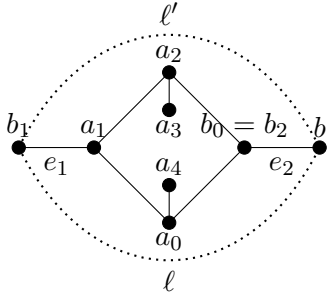
\begin{figure}[ht]
  \centering
  \begin{tikzpicture}[scale=0.5]
  \filldraw [black] (0,2) circle (5pt);
  \node [above] at (0,2) {$a_2$};
  \filldraw [black] (0,-2) circle (5pt);
  \node [below] at (0,-2) {$a_0$};
  \filldraw [black] (-2,0) circle (5pt);
  \node [above] at (-2,0) {$a_1$};
  \filldraw [black] (2,0) circle (5pt);
  \node [above] at (2,0) {$b_0=b_2$};
  \filldraw [black] (4,0) circle (5pt);
  \node [above] at (4,0) {$b$};
  \filldraw [black] (-4,0) circle (5pt);
  \node [above] at (-4,0) {$b_1$};

  \node[below] at (-3,0) {$e_1$};
  \node[below] at (3,0) {$e_2$};

  \filldraw [black] (0,1) circle (5pt);
  \node [below] at (0,1) {$a_3$};

  \filldraw [black] (0,-1) circle (5pt);
  \node [above] at (0,-1) {$a_4$};

  \draw (0,2) -- (2,0);
  \draw (2,0) -- (0,-2);
  \draw (0,-2) -- (-2,0);
  \draw (-2,0) -- (0,2);
  \draw(-2,0) -- (-4,0);
  \draw (2,0) -- (4,0);

  \draw (0,-1) -- (0,-2);
  \draw (0,1) -- (0,2);

  \draw[thick, dotted] (-4,0) .. controls (-2,4) and (2,4) ..(4,0);
  \draw[thick, dotted] (-4,0) .. controls (-2,-4) and (2,-4) .. (4,0);

  \node [above] at (0,3) {$\ell'$};
  \node [below] at (0,-3) {$\ell$};

  \end{tikzpicture}
  \captionsetup{justification=centering}
   \caption{When $b_0 = b_2$}
   \label{G-two-bounding-faces}
\end{figure}

We will now define two auxiliary graphs $G_1$ and $G_2$.
$G_1$ consists of the cycle $C$ and its interior 
region, augmented by a single edge $e^* = \{a_1, b_0\}$.
$G_2$ consists of everything in $G$ properly
exterior to the cycle $C$ (i.e.,
not containing $C$ and its interior)
with the two edges $e_1 =
\{a_1, b_1\}$ and $e_2 = \{b, b_0\}$ replaced by 
one new edge $e_{12} = \{a_1, b_0\}$. (See Figure~\ref{G1G2}). 
$G_1$ is not one of the exceptional graphs since it contains a pentagon.
If $G_2$ were $M_{2,3}$ then the two paths $\pi$ and $\pi'$
denoted by the dotted lines
with labels $\ell$ and $\ell'$  both consist of a single edge
and are present in $G$, contradicting $G$ being simple.
If $G_2$ were $K_4$ then there are four triangle faces (on the spherical embedding),
two of which must be present in $G$,  contradicting $G$ having
no triangle faces.
So, by induction, both $G_1$ and $G_2$
have a P3EM. Note that $G_1$ contains two
triangle faces separated by the edge $e^*$.
Any P3EM of $G_1$ assigns $e^*$ to one of these two
triangle faces which implies that all three edges
of this triangle face must be assigned to this face.
Thus, in Figure~\ref{G1}
the four edges $\{a_1, a_2\}, \{a_2, b_0\}, \{a_1, a_0\}, \{a_0, b_0\}$ must be assigned all up
or all down, according to whether $e^*$
is assigned down or up, respectively. 
In $G_2$, the edge $e_{12}$ is assigned either up or down
to the two adjacent faces. If $e_{12}$ is assigned up, then
along the path $\pi$ (resp. $\pi'$) of $\ell$ (resp. $\ell'$) 
edges there are
$0 \pmod 3$ (resp. $2 \pmod 3$) edges assigned toward the face
that $e_{12}$ is on its boundary.
If $e_{12}$ is assigned  down, then the opposite happens, i.e.,
$2 \pmod 3$ (resp. $0 \pmod 3$ ) edges of $\pi$ (resp. $\pi'$) 
 are assigned toward the  face that $e_{12}$ is on its boundary.

We now define an edge assignment  that will be a P3EM for $G$.
Every edge in $G$ other than $e_1$ and $e_2$ belongs to
exactly one of $G_1$ or $G_2$. We assign these edges according to
the assignment in  $G_1$ or $G_2$ respectively.
This satisfies the requirement of P3EM for every face other than
$\Delta$ and $\Delta'$ in $G$. For the
assignment on  $e_1$ and $e_2$, there are
four cases according to how $e^*$ in $G_1$ and $e_{12}$ in $G_2$
are assigned.
The first case is both $e^*$ in $G_1$  and $e_{12}$ in $G_2$ are
assign up, and we assign $e_1$ down and $e_2$ up in $G$.
Then,  there are a total of
3 edges $e_1 = \{b_1, a_1\}$,
$\{a_1, a_0\}$ and $\{a_0, b_0\}$ and $0 \pmod 3$  edges of $\pi$
assigned toward $\Delta$.
Also  there are a total of
1 edge $e_2 = \{b_0, b\}$,
and $2 \pmod 3$  edges of $\pi'$
assigned toward 
$\Delta'$.
The  second case is when $e^*$ and $e_{12}$  are assigned
respectively  up and down, and we assign $e_1$ and $e_2$ both
down in $G$. 
Then,  there are a total of
4 edges $e_1 = \{b_1, a_1\}$,
$\{a_1, a_0\}$, $\{a_0, b_0\}$  and $e_2 = (b_0, b)$,
and $2 \pmod 3$  edges of $\pi$
assigned toward $\Delta$, making it $0 \pmod 3$ altogether.
Also  there are no edge among these four
and  $0 \pmod 3$  edges of $\pi'$
assigned toward 
$\Delta'$. The other two cases are similar.
We have proved that a P3EM exists for $G$.

Hence we may assume that $b_0, \ldots, b_4$ are all distinct.


\begin{figure}[ht]
  \centering
  \begin{subfigure}[b]{0.4\textwidth}
  \centering
  \begin{tikzpicture}[scale=0.5]
  \filldraw [black] (0,2) circle (5pt);
  \node [above] at (0,2) {$a_2$};
  \filldraw [black] (0,-2) circle (5pt);
  \node [below] at (0,-2) {$a_0$};
  \filldraw [black] (-2,0) circle (5pt);
  \node [above] at (-2,0) {$a_1$};
  \filldraw [black] (2,0) circle (5pt);
  \node [right] at (2,0) {$b_0=b_2$};

  \draw (0,2) -- (2,0);
  \draw (2,0) -- (0,-2);
  \draw (0,-2) -- (-2,0);
  \draw (-2,0) -- (0,2);
  \draw (-2,0) .. controls (-1,-4) and (1,-4) .. (2,0);

  \node [below] at (0,-3) {$e^*$};
  \end{tikzpicture}
  \captionsetup{justification=centering}
  \caption{$G_1$}
  \label{G1}
  \end{subfigure}
  \begin{subfigure}[b]{0.4\textwidth}
  \centering
  \begin{tikzpicture}[scale=0.5]
  \filldraw [black] (4,0) circle (5pt);
  \node [above] at (4,0) {$b$};
  \filldraw [black] (-4,0) circle (5pt);
  \node [above] at (-4,0) {$b_1$};
  \draw(-4,0) -- (4,0);
  \draw[thick, dotted] (-4,0) .. controls (-2,4) and (2,4) ..(4,0);
  \draw[thick, dotted] (-4,0) .. controls (-2,-4) and (2,-4) .. (4,0);
  \node [above] at (0,0) {$e_{12}$};
  \node [above] at (0,3) {$l'$};
  \node [below] at (0,-3) {$l$};
  \end{tikzpicture}
  \captionsetup{justification=centering}
  \caption{$G_2$}
  \label{G2}
  \end{subfigure}
  \captionsetup{justification=centering}
  \caption{$G_1$ and $G_2$}
  \label{G1G2}
\end{figure}
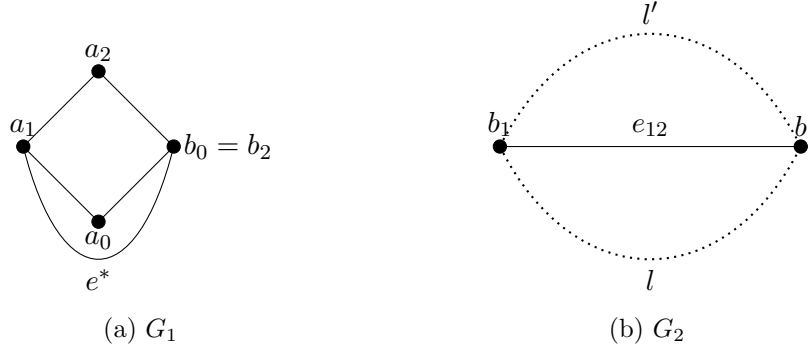

We claim that  there is a simple path connecting $b_i$ and $b_{i+1}$ for each $0 \leq i \leq 4$ (where $b_5 = b_0$), and furthermore
the cycle $(a_i, b_i, \ldots, b_{i+1}, a_{i+1}, a_i)$ using this path
is the  boundary of a face.
Define an $a_i$-R path as follows: start from $a_i$ and take the first edge  $\{a_i, b_i\}$, and then at every new vertex (of degree 3) choose
the right branch for the next vertex, until we encounter 
a previously visited vertex on this walk, or  one of $\{a_0, a_1 , \ldots, a_4\}$, then stop. For notational simplicity we 
consider the case for $i=0$;  all other cases are the same. Suppose the $a_0$-R path is $\{x_0, x_1, ..., x_k, ..., x_m\}$, where $x_0 = a_0$ and $ x_1 = b_0$.
First we claim $x_m  \neq a_0$.  If it were, then the step before would have been $a_1, a_4$ or $b_1 =x_0$, but then  the $a_0$-R path should have stopped at $x_{m-1}$. Next we claim that $x_m \in \{a_1,a_2,a_3,a_4\}$.
Indeed, if $x_m \not \in  \{a_1,a_2,a_3,a_4\}$, then it is a 
previously visited vertex $x_k$ on this walk, with $k \ge 1$.
Then $x_{k-1}$ exists.
Moreover,  $(x_k, x_{k+1}, \dots, x_m)$ is a cycle which is formed by
always taking the right branch at the next vertex. The last edge $(x_{m-1}, x_m)$ (which is $(x_{m-1}, x_k)$)
must be the left branch edge when coming from the direction $(x_{k-1}, x_k)$,
thus the traversal of the cycle  $(x_k, x_{k+1}, \dots, x_m)$ is counterclockwise.
Thus 
 the edge $\{x_{k-1}, x_k\}$ is a bridge, a contradiction.
 See Figure~\ref{fig:simple paths between b_i's} for an illustration.

 \begin{figure}[ht]
  \centering
  \begin{tikzpicture}
  \filldraw [black] (0,0) circle (1pt);
  \node[left] at (0,0) {$a_0 = x_0$};
  \filldraw [black] (0,1) circle (1pt);
  \node[left] at (0,1) {$b_0 = x_1$};
  \draw (0,0) -- (0,1);

  \filldraw [black] (0.75,1.75) circle (1pt);
  \node[left] at (0.75,1.75) {$x_2$};
  \draw (0,1) -- (0.75,1.75);

  \draw (0,1) -- (-0.75,1.75);
  \draw (0.75,1.75) -- (0.5,2.25);
  \draw[thick, dotted] (0.75,1.75) -- (3,3);
  \draw[very thick] (3,3) -- (4,4);
  \draw (3,3) -- (2,3.5);
  \filldraw [black] (3,3) circle (1pt);
  \node[below] at (3,3) {$x_{k-1}$};
  \node[above left] at (4,4) {$x_k=x_m$};
  \draw (4,4) -- (4.5,3.5);
  \draw (4,4) -- (4.5,4.5);
  \filldraw [black] (4,4) circle (1pt);
  \filldraw [black] (4.5,3.5) circle (1pt);
  \filldraw [black] (4.5,4.5) circle (1pt);
  \node[above] at (4.5,4.5) {$x_{m-1}$};
  \node[below] at (4.5,3.5) {$x_{k+1}$};
  \draw[thick, dotted] (4.5,4.5) .. controls (6,5) and (6,3) .. (4.5,3.5);
  \draw (4.5,4.5) -- (5,4.25);
  \draw (4.5,3.5) -- (5,3.75);
  \draw (5.5,3.7) -- (5,3.85);
  \draw (5.5,4.3) -- (5,4);
  \filldraw [black] (5.5,3.7) circle (1pt);
  \filldraw [black] (5.5,4.3) circle (1pt);
  \end{tikzpicture}
  \captionsetup{justification=centering}
   \caption{Simple paths between $b_i$'s}
   \label{fig:simple paths between b_i's}
\end{figure}
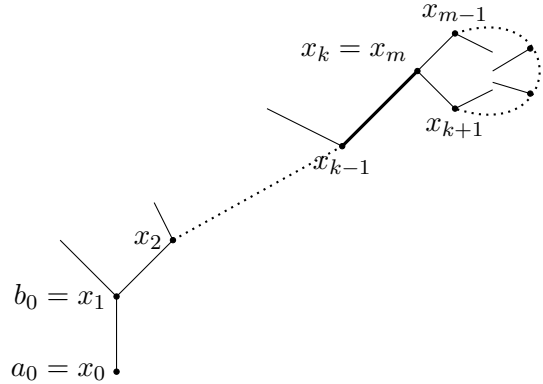
 
 Next we claim that $x_m = a_1$, and $x_{m-1} = b_1$, (see  Figure~\ref{pentagon_path_1}).
 We prove this by eliminating the possibilities
 $x_m \in \{a_2, a_3, a_4\}$. Suppose  $x_m = a_4$. It follows from the definition 
 of the $a_0$-R path that
 $x_{m-1} \not \in \{a_0, a_1, \ldots, a_4\}$,
 being the step before $x_m$, and then the only way to reach $x_m = a_4$
is $x_{m-1} = b_4$. Since this  $a_0$-R path always  takes the
right branch, viewing the plane graph on a spherical embedding we can 
consider the face to the right of this  $a_0$-R path
as the outer face and then the edge $\{a_0, a_4\}$ is a chord.
However, by our assumption $G$ is chordless. 
Now suppose $x_m = a_3$. Then consider the $a_2$-R path.
By planarity and the fact that one single face
borders the right hand side of the 
{$a_0$-R}
path which ends in $a_3$,
the $a_2$-R path cannot end in $a_4$, and therefore
it must end in $a_1$. However, considering the indices mod 5 this is exactly the same
situation with the  $a_0$-R path ending in $a_4$, another contradiction.
Finally,  if the    $a_0$-R path ends in  $x_m = a_2$, then the
 $a_1$-R path would violate planarity, or produce a bridge. 
We conclude that  $x_m = a_1$. And then 
 it follows that $x_{m-1} = b_1$, and we have
 a face with boundary $(a_0, b_0, \ldots, b_1, a_1, a_0)$ 
 from this $a_0$-R path.
 The same is true for all $a_i$-R paths. 

 In other words, we now have a pentagon face $P$ depicted as in Figure~\ref{fig:final pentagon}.

 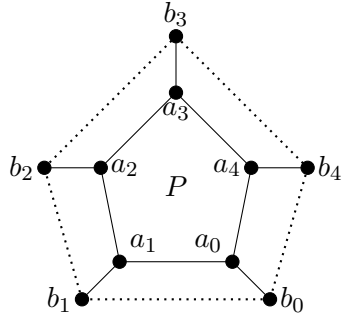
\begin{figure}[ht]
  \centering
  \begin{tikzpicture}[scale=0.5]
  \node at (0,0.5) {$P$};
  \filldraw [black] (1.5,-1.5) circle (5pt);
  \node[above left] at (1.5,-1.5) {$a_0$};
  \filldraw [black] (-1.5,-1.5) circle (5pt);
  \node[above right] at (-1.5,-1.5) {$a_1$};
  \filldraw [black] (-2,1) circle (5pt);
  \node[right] at (-2,1) {$a_2$};
  \filldraw [black] (2,1) circle (5pt);
  \node[left] at (2,1) {$a_4$};
  \filldraw [black] (0,3) circle (5pt);
  \node[below] at (0,3) {$a_3$};
  \filldraw [black] (2.5,-2.5) circle (5pt);
  \node[right] at (2.5,-2.5) {$b_0$};
  \filldraw [black] (-2.5,-2.5) circle (5pt);
  \node[left] at (-2.5,-2.5) {$b_1$};
  \filldraw [black] (-3.5,1) circle (5pt);
  \node[left] at (-3.5,1) {$b_2$};
  \filldraw [black] (3.5,1) circle (5pt);
  \node[right] at (3.5,1) {$b_4$};
  \filldraw [black] (0,4.5) circle (5pt);
  \node[above] at (0,4.5) {$b_3$};
  \draw (1.5,-1.5) -- (-1.5,-1.5);
  \draw (-1.5, -1.5) -- (-2,1);
  \draw (-2,1) -- (0,3);
  \draw (0,3) -- (2,1);
  \draw (2,1) -- (1.5,-1.5);
  \draw (1.5,-1.5) -- (2.5,-2.5);
  \draw (-1.5,-1.5) -- (-2.5,-2.5);
  \draw (0,3) -- (0,4.5);
  \draw (-2,1) -- (-3.5,1);
  \draw (2,1) -- (3.5,1);
  \draw[thick, dotted] (2.5,-2.5) -- (-2.5,-2.5);
  \draw[thick, dotted] (-2.5,-2.5) -- (-3.5,1);
  \draw[thick, dotted] (-3.5,1) -- (0,4.5);
  \draw[thick, dotted] (0,4.5) -- (3.5,1);
  \draw[thick, dotted] (3.5,1) -- (2.5,-2.5);
  \end{tikzpicture}  
  \captionsetup{justification=centering}
   \caption{Pentagon in the original graph}
   \label{fig:final pentagon}
\end{figure}

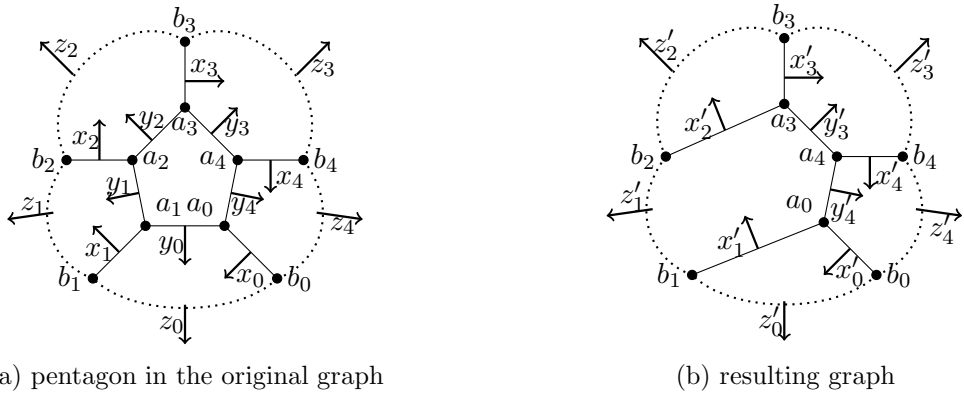
\begin{figure}[ht]
  \centering
  \begin{subfigure}[b]{0.5\textwidth}
  \centering
  \begin{tikzpicture}[scale=0.35]
  \filldraw [black] (1.5,-1.5) circle (5pt);
  \node[above left] at (1.5,-1.5) {$a_0$};
  \filldraw [black] (-1.5,-1.5) circle (5pt);
  \node[above right] at (-1.5,-1.5) {$a_1$};
  \filldraw [black] (-2,1) circle (5pt);
  \node[right] at (-2,1) {$a_2$};
  \filldraw [black] (2,1) circle (5pt);
  \node[left] at (2,1) {$a_4$};
  \filldraw [black] (0,3) circle (5pt);
  \node[below] at (0,3) {$a_3$};
  \filldraw [black] (3.5,-3.5) circle (5pt);
  \node[right] at (3.5,-3.5) {$b_0$};
  \filldraw [black] (-3.5,-3.5) circle (5pt);
  \node[left] at (-3.5,-3.5) {$b_1$};
  \filldraw [black] (-4.5,1) circle (5pt);
  \node[left] at (-4.5,1) {$b_2$};
  \filldraw [black] (4.5,1) circle (5pt);
  \node[right] at (4.5,1) {$b_4$};
  \filldraw [black] (0,5.5) circle (5pt);
  \node[above] at (0,5.5) {$b_3$};
  \draw (1.5,-1.5) -- (-1.5,-1.5);
  \draw (-1.5, -1.5) -- (-2,1);
  \draw (-2,1) -- (0,3);
  \draw (0,3) -- (2,1);
  \draw (2,1) -- (1.5,-1.5);
  \draw (1.5,-1.5) -- (3.5,-3.5);
  \draw (-1.5,-1.5) -- (-3.5,-3.5);
  \draw (0,3) -- (0,5.5);
  \draw (-2,1) -- (-4.5,1);
  \draw (2,1) -- (4.5,1);
  \draw[thick, dotted] (3.5,-3.5) .. controls (2,-5) and (-2,-5) .. (-3.5,-3.5);
  \draw[thick, dotted] (-3.5,-3.5) .. controls (-5,-3) and (-6,-1).. (-4.5,1);
  \draw[thick, dotted] (-4.5,1) .. controls (-6,4) and (-2,7) .. (0,5.5);
  \draw[thick, dotted] (4.5,1) .. controls (6,4) and (2,7) .. (0,5.5);
  \draw[thick, dotted] (3.5,-3.5) .. controls (5,-3) and (6,-1) .. (4.5,1); 
  \draw[->, line width=0.3mm] (2.5, -2.5) -- (1.5,-3.5);
  \draw[->, line width=0.3mm] (-2.5,-2.5) -- (-3.5,-1.5);
  \draw[->, line width=0.3mm] (-3.25,1) -- (-3.25,2.55);
  \draw[->, line width=0.3mm] (0,4) -- (1.5,4);
  \draw[->, line width=0.3mm] (3.25,1) -- (3.25,-0.25);
  \node at (2.5,-3.5) {$x_0$};
  \node at (-3.25,-2.5) {$x_1$};
  \node at (-3.75,1.75) {$x_2$};
  \node at (0.75,4.5) {$x_3$};
  \node at (4,0.25) {$x_4$};
  \draw[->, line width=0.3mm] (0,-1.5) -- (0,-3);
  \draw[->, line width=0.3mm] (-1.75,-0.25) -- (-3,-0.5);
  \draw[->, line width=0.3mm] (1.75,-0.25) -- (3,-0.5);
  \draw[->, line width=0.3mm] (-1.25,1.75) -- (-2.25,2.75);
  \draw[->, line width=0.3mm] (1,2) -- (2,3);
  \node at (-0.5,-2.25) {$y_0$};
  \node at (-2.5,0) {$y_1$};
  \node at (-1.25,2.5) {$y_2$};
  \node at (2,2.25) {$y_3$};
  \node at (2.25,-0.75) {$y_4$};
  \draw[->, line width=0.3mm] (0,-4.5) -- (0,-6);
  \draw[->, line width=0.3mm] (-5,-1) -- (-6.75,-1.25);
  \draw[->, line width=0.3mm] (5,-1) -- (6.75,-1.25);
  \draw[->, line width=0.3mm] (-4.25,4.25) -- (-5.5,5.5);
  \draw[->, line width=0.3mm] (4.25,4.25) -- (5.5,5.5);
  \node at (-0.5,-5.25) {$z_0$};
  \node at (-5.75,-0.75) {$z_1$};
  \node at (6,-1.5) {$z_4$};
  \node at (-4.5,5.25) {$z_2$};
  \node at (5.25,4.5) {$z_3$};
  \end{tikzpicture}
  \captionsetup{justification=centering}
  \caption{pentagon in the original graph}
  \label{pentagon_1}
  \end{subfigure}
  \begin{subfigure}[b]{0.45\textwidth}
  \centering
  \begin{tikzpicture}[scale=0.35]
  \filldraw [black] (1.5,-1.5) circle (5pt);
  \node[above left] at (1.5,-1.5) {$a_0$};
  \filldraw [black] (2,1) circle (5pt);
  \node[left] at (2,1) {$a_4$};
  \filldraw [black] (0,3) circle (5pt);
  \node[below] at (0,3) {$a_3$};
  \filldraw [black] (3.5,-3.5) circle (5pt);
  \node[right] at (3.5,-3.5) {$b_0$};
  \filldraw [black] (-3.5,-3.5) circle (5pt);
  \node[left] at (-3.5,-3.5) {$b_1$};
  \filldraw [black] (-4.5,1) circle (5pt);
  \node[left] at (-4.5,1) {$b_2$};
  \filldraw [black] (4.5,1) circle (5pt);
  \node[right] at (4.5,1) {$b_4$};
  \filldraw [black] (0,5.5) circle (5pt);
  \node[above] at (0,5.5) {$b_3$};
  \draw (0,3) -- (2,1);
  \draw (2,1) -- (1.5,-1.5);
  \draw (1.5,-1.5) -- (3.5,-3.5);
  \draw (0,3) -- (0,5.5);
  \draw (2,1) -- (4.5,1);
  \draw[thick, dotted] (3.5,-3.5) .. controls (2,-5) and (-2,-5) .. (-3.5,-3.5);
  \draw[thick, dotted] (-3.5,-3.5) .. controls (-5,-3) and (-6,-1).. (-4.5,1);
  \draw[thick, dotted] (-4.5,1) .. controls (-6,4) and (-2,7) .. (0,5.5);
  \draw[thick, dotted] (4.5,1) .. controls (6,4) and (2,7) .. (0,5.5);
  \draw[thick, dotted] (3.5,-3.5) .. controls (5,-3) and (6,-1) .. (4.5,1); 
  \draw (-3.5,-3.5) -- (1.5,-1.5); 
  \draw (-4.5,1) -- (0,3); 
  \draw[->, line width=0.3mm] (2.5, -2.5) -- (1.5,-3.5);
  \draw[->, line width=0.3mm] (0,4) -- (1.5,4);
  \draw[->, line width=0.3mm] (3.25,1) -- (3.25,-0.25);
  \draw[->, line width=0.3mm] (-1,-2.5) -- (-1.5,-1.25); 
  \draw[->, line width=0.3mm] (-2.25,2) -- (-2.75,3.25); 
  \node at (2.5,-3.35) {$x'_0$};
  \node at (-2,-2.25) {$x'_1$};
  \node at (-3.25,2.25) {$x'_2$};
  \node at (0.75,4.65) {$x'_3$};
  \node at (4,0.25) {$x'_4$};
  \draw[->, line width=0.3mm] (1.75,-0.25) -- (3,-0.5);
  \draw[->, line width=0.3mm] (1,2) -- (2,3);
  \node at (2,2.25) {$y'_3$};
  \node at (2.25,-1) {$y'_4$};
  \draw[->, line width=0.3mm] (0,-4.5) -- (0,-6);
  \draw[->, line width=0.3mm] (-5,-1) -- (-6.75,-1.25);
  \draw[->, line width=0.3mm] (5,-1) -- (6.75,-1.25);
  \draw[->, line width=0.3mm] (-4.25,4.25) -- (-5.5,5.5);
  \draw[->, line width=0.3mm] (4.25,4.25) -- (5.5,5.5);
  \node at (-0.5,-5.25) {$z'_0$};
  \node at (-5.75,-0.5) {$z'_1$};
  \node at (6,-1.75) {$z'_4$};
  \node at (-4.5,5.25) {$z'_2$};
  \node at (5.25,4.5) {$z'_3$};
  \end{tikzpicture}
  \captionsetup{justification=centering}
  \caption{resulting graph}
  \label{pentagon_2}
  \end{subfigure}
  \captionsetup{justification=centering}
  \caption{Transforming a pentagon}
  \label{pentagon}
\end{figure}
We now perform the transformation as illustrated in Figure~\ref{pentagon}. 
The transformed graph ({\bf b}) on the right is not $M_{2,3}$ or $K_4$
by vertex count. By induction there is a P3EM $M'$ on the transformed graph.
We use  Boolean variables $x'_i$ ($0 \leq i \leq 4$), 
and $y'_3$, $y'_4$ to denote the assignment on those 7 edges in Figure~\ref{pentagon}b, 
such that the variable is  1 
 if the  corresponding  edge is assigned to the face indicated by its arrow, and is 0 if 
 it is assigned to the face on the other side.
 We also use
nonnegative integer
 variables $z'_i$ ($0 \leq i \leq 4$) to denote the number  of edges
 assigned to the  side indicated along the simple path $b_i$ to $b_{i+1}$.
Now we define a P3EM $M$ on $G$ using $M'$ as follows.
 All edges in $G$ that are not incident to $a_0, a_1, a_2, a_3$ or $a_4$ 
 will retain their assignment as in $M'$. These include all edges
 on the path $b_i$ to $b_{i+1}$ (and all edges beyond these simple paths that are not
 depicted in Figure~\ref{pentagon}a.)
 In particular, if $z_i$ ($0 \leq i \leq 4$) is the number  of edges
 assigned to the  side indicated along the simple path $b_i$ to $b_{i+1}$ in $G$,
 then $z_i = z'_i$.
For the 10 edges incident to at least one of $a_0, a_1, a_2, a_3$ or $a_4$ in Figure~\ref{pentagon}a we will use  Boolean variables $x_i$ and $y_i$ ($0 \leq i \leq 4$) to denote
the  assignment of $M$ on $G$, such that the variable is  1 
 if the  corresponding  edge is assigned to the face indicated by its arrow, and is 0 
 otherwise.
 
 A moment's reflection will convince the reader that
 $M$ is a P3EM on $G$ iff we can assign Boolean 0-1 variables 
 $x_i$ and $y_i$ ($0 \leq i \leq 4$) that satisfy the following equation system $(\Sigma)$,
 where $\overline{x} = 1- x \in \{0, 1\}$ denotes the negation of
 the  Boolean variable  $x$.
{\small
\begin{empheq}[left=(\Sigma)~~~\empheqlbrace]{align}
  &x_0 + y_0 + \overline{x_1} \equiv x'_0 + \overline{x'_1} \pmod{3} \notag
  \\
  &x_2 + y_2 + \overline{x_3} \equiv x'_2 + \overline{x'_3} \pmod{3} \notag
  \\
  &x_3 + y_3 + \overline{x_4} \equiv x'_3 + y'_3 + \overline{x'_4} \pmod{3} \notag
  \\
  &x_4 + y_4 + \overline{x_0} \equiv x'_4 + y'_4 + \overline{x'_0} \pmod{3} \notag
  \\
  &x_1 + y_1 + \overline{x_2} \equiv x'_1 + \overline{x'_2} + \overline{y'_3} + \overline{y'_4} \pmod{3} \notag
  \\
  &\sum\limits_{i=0}^{4} \overline{y_i} \equiv 0 \pmod{3} \notag
\end{empheq}
}

We note that, while this equation system consists of all linear equations
mod 3, it is not an ordinary linear equation system over $\mathbb{Z}_3$;
the complicating factor is that all the variables must take Boolean values in $\{0, 1\}$.
Somewhat miraculously, we show that for any Boolean values of
$x'_i$ ($0 \leq i \leq 4$) 
and $y'_3$, $y'_4$, we  can always solve the equation system $(\Sigma)$
for the Boolean variables 
 $x_i$ and $y_i$ ($0 \leq i \leq 4$). 
 
 If $(y'_3, y'_4) \neq (0,0)$, then we set $x_i = x'_i$ for $0 \leq i \leq 4$, $y_3 = y'_3$, $y_4 = y'_4$, $y_0=y_2 =0$,  and
$y_1 = \overline{y'_3} + \overline{y'_4} \in \{ 0, 1\}$. (Note that
 $(y'_3, y'_4) \neq (0,0)$ is used to obtain
$\overline{y'_3} + \overline{y'_4} \in \{ 0, 1\}$.)
One can check that this assignment
solves $(\Sigma)$.  

Now suppose $(y'_3, y'_4) = (0,0)$. The system of equations now becomes
{\small
\begin{empheq}[left=(\Sigma')~~~\empheqlbrace]{align}
  &x_0 + y_0 + \overline{x_1} \equiv x'_0 + \overline{x'_1} \pmod{3} \notag
  \\
  &x_2 + y_2 + \overline{x_3} \equiv x'_2 + \overline{x'_3} \pmod{3} \notag
  \\
  &x_3 + y_3 + \overline{x_4} \equiv x'_3 + \overline{x'_4} \pmod{3} \notag
  \\
  &x_4 + y_4 + \overline{x_0} \equiv x'_4 + \overline{x'_0} \pmod{3} \notag
  \\
  &x_1 + y_1 + \overline{x_2} \equiv x'_1 + \overline{x'_2} + 2 \pmod{3} \notag
  \\
  &\sum\limits_{i=0}^{4}  \overline{y_i} \equiv 0 \pmod{3} \notag
\end{empheq}
}

If $x'_1 = 0$, then we set $x_1 = 1$, and $x_i = x'_i$ for $0 \leq i \leq 4$, $i \neq 1$,
and set $y_2 = y_3 = y_4 = 0$, $y_0 = y_1 = 1$.
One can check that this assignment
solves  $(\Sigma')$. If $(x'_1, x'_2) = (1,1)$, then we set
$x_2 = 0$, and $x_i = x'_i$ for $0 \leq i \leq 4$,  $i \neq 2$, and set
$y_0 = y_3 = y_4 = 0$, $y_1 = y_2 = 1$. This  solves $(\Sigma')$. Thus it remains to consider the case when $(x'_1, x'_2) = (1,0)$. 
There remain eight cases, each corresponding to an assignment  
$(x'_0, x'_3, x'_4) \in \{0, 1\}^3$.

At this point we have $(y'_3, y'_4) = (0,0)$ in addition to $(x'_1, x'_2) = (1,0)$, so 
now we are  in a situation in Figure~\ref{pentagon}b where $x'_1, x'_2, y'_3$ and $y'_4$ are all pointing into the face bounded by the cycle $(b_1, a_0, a_4, a_3, b_2, \ldots, b_1)$. By 
a reflection along the $\{a_4,b_4\}$-axis, we only need to consider four cases,
with $x'_4 =0$.  These four cases are explicitly given in
Figure~\ref{geometry}, where we use a double arrow to indicate
an actual assignment of the corresponding edge into the face indicated.
For example, the following figure deals with the case  $(x'_0, x'_3 , x'_4) =(0, 0, 0)$.
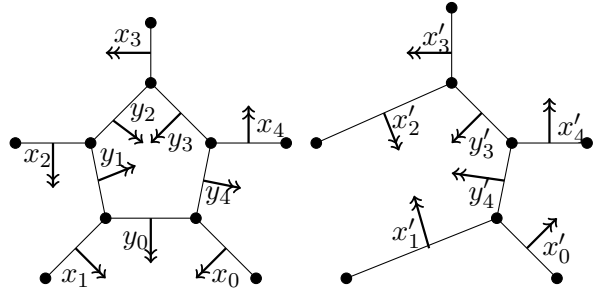
\begin{figure}[H]
  \centering
  \begin{tikzpicture}[scale=0.4]
  \filldraw [black] (1.5,-1.5) circle (5pt);
  \filldraw [black] (-1.5,-1.5) circle (5pt);
  \filldraw [black] (-2,1) circle (5pt);
  \filldraw [black] (2,1) circle (5pt);
  \filldraw [black] (0,3) circle (5pt);
  \filldraw [black] (3.5,-3.5) circle (5pt);
  \filldraw [black] (-3.5,-3.5) circle (5pt);
  \filldraw [black] (-4.5,1) circle (5pt);
  \filldraw [black] (4.5,1) circle (5pt);
  \filldraw [black] (0,5) circle (5pt);
  \draw (1.5,-1.5) -- (-1.5,-1.5);
  \draw (-1.5, -1.5) -- (-2,1);
  \draw (-2,1) -- (0,3);
  \draw (0,3) -- (2,1);
  \draw (2,1) -- (1.5,-1.5);
  \draw (1.5,-1.5) -- (3.5,-3.5);
  \draw (-1.5,-1.5) -- (-3.5,-3.5);
  \draw (0,3) -- (0,5);
  \draw (-2,1) -- (-4.5,1);
  \draw (2,1) -- (4.5,1);
  \draw[->>, line width=0.3mm] (2.5, -2.5) -- (1.5,-3.5);
  \draw[->>, line width=0.3mm] (-2.5,-2.5) -- (-1.5,-3.5);
  \draw[->>, line width=0.3mm] (-3.25,1) -- (-3.25,-0.5);
  \draw[->>, line width=0.3mm] (0,4) -- (-1.5,4);
  \draw[->>, line width=0.3mm] (3.25,1) -- (3.25,2.25);
  \draw[->>, line width=0.3mm] (0,-1.5) -- (0,-3);
  \draw[->>, line width=0.3mm] (-1.75,-0.25) -- (-0.5,0.25);
  \draw[->>, line width=0.3mm] (1.75,-0.25) -- (3,-0.5);
  \draw[->>, line width=0.3mm] (-1.25,1.75) -- (-0.25,1);
  \draw[->>, line width=0.3mm] (1,2) -- (0,1);
  \node at (2.5,-3.5) {$x_0$};
  \node at (-2.5,-3.5) {$x_1$};
  \node at (-3.75,0.5) {$x_2$};
  \node at (-0.75,4.5) {$x_3$};
  \node at (4,1.5) {$x_4$};
  \node at (-0.5,-2.25) {$y_0$};
  \node at (-1.25,0.5) {$y_1$};
  \node at (-0.25,2) {$y_2$};
  \node at (1,1) {$y_3$};
  \node at (2.25,-0.75) {$y_4$};
  \filldraw [black] (11.5,-1.5) circle (5pt);
  \filldraw [black] (12,1) circle (5pt);
  \filldraw [black] (10,3) circle (5pt);
  \filldraw [black] (13.5,-3.5) circle (5pt);
  \filldraw [black] (6.5,-3.5) circle (5pt);
  \filldraw [black] (5.5,1) circle (5pt);
  \filldraw [black] (14.5,1) circle (5pt);
  \filldraw [black] (10,5.5) circle (5pt);
  \draw (10,3) -- (12,1);
  \draw (12,1) -- (11.5,-1.5);
  \draw (11.5,-1.5) -- (13.5,-3.5);
  \draw (10,3) -- (10,5.5);
  \draw (12,1) -- (14.5,1);
  \draw (6.5,-3.5) -- (11.5,-1.5); 
  \draw (5.5,1) -- (10,3); 
  \draw[->>, line width=0.3mm] (12.5, -2.5) -- (13.5,-1.5);
  \draw[->>, line width=0.3mm] (10,4) -- (8.5,4);
  \draw[->>, line width=0.3mm] (13.25,1) -- (13.25,2.5);
  \draw[->>, line width=0.3mm] (9.25,-2.5) -- (8.75,-0.75); 
  \draw[->>, line width=0.3mm] (7.75,2) -- (8.25,0.75); 
  \draw[->>, line width=0.3mm] (11.75,-0.25) -- (10,0);
  \draw[->>, line width=0.3mm] (11,2) -- (10,1);
  \node at (13.5,-2.5) {$x'_0$};
  \node at (8.5,-2) {$x'_1$};
  \node at (8.5,1.75) {$x'_2$};
  \node at (9.5,4.5) {$x'_3$};
  \node at (14,1.5) {$x'_4$};
  \node at (11,1) {$y'_3$};
  \node at (11,-0.65) {$y'_4$};
  \end{tikzpicture}
  \captionsetup{justification=centering}
  \caption{$(x'_0,x'_3,x'_4) = (0,0,0)$}
\label{fig:geometric inside proof}
\end{figure}

For other cases, see Figure~\ref{geometry}. 

Finally, we note that  the proof is constructive.
When smaller graphs are defined for induction purposes,
the size of the smaller graph strictly decreases
and in the case when two smaller graphs are needed (as in the case
dealing with a chord or getting distinct $b_i$'s) the sum of
sizes of the smaller graphs is approximately that of the original graph.
Tracing through the proof it can be easily verified that 
a planar 3-way edge
matching can be found in polynomial time.
This completes the proof of Theorem~\ref{thm:P3EM}.
\end{proof}
\section{Dichotomy Theorem}

In this section we start the proof of Theorem~\ref{thm:main}.
When $\neg (f_0=f_3=0)$, (i.e., it is not the case that both $f_0=0$ and $f_1=0$),  by dividing a nonzero constant and possibly flipping
0 and 1 without changing the complexity of Holant, 
we can normalize the signature $[f_0,f_1,f_2,f_3]$ to be $[1,a,b,c]$. 
We first deal with a special case where $a=b$ and $c=1$.

\begin{lemma}\label{lemma:[1,a,a,1]}
\plholant{[1,a,a,1]}{(=_3)} is in $\operatorname{FP}$.
\end{lemma}


\begin{proof}
Perform the holographic transformation by the Hadamard matrix $H= \left[
\begin{smallmatrix}
    1 & 1 \\  1 & -1
\end{smallmatrix} \right]$, to $[1,a,a,1]$ on the left 
and $(=_3) = [1,0,0,1]$ on the right in the bipartite setting,  we get $$[1,a,a,1] H^{\otimes 3} = [2+6a,0,2-2a,0],~~~\mbox{and}~~~(H^{\otimes 3})^{-1} [1,0,0,1] = \tfrac{1}{4}[1,0,1,0].$$
Both transformed signatures are matchgate signatures~\cite{CaiL11} and thus the problem 
can be solved in polynomial time by the FKT algorithm.
\end{proof}

Another special case where $b=1$ and $a=c$  will be needed later.

\begin{lemma}\label{lemma:[1,a,1,a]}
\plholant{[1,a,1,a]}{(=_3)} is \numP-hard unless $a=0$ or $\pm 1$, in which cases it is in $\operatorname{FP}$.
\end{lemma}

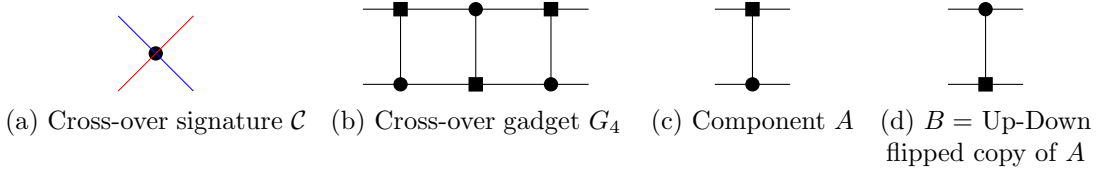
\begin{figure}[ht]
  \centering
  \begin{subfigure}[t]{0.25\textwidth}
  \centering
  \begin{tikzpicture}[scale=0.5]
  \filldraw (0,0) circle (5pt);
  \draw[red] (0,0) -- (-1,-1);
  \draw[blue] (0,0) -- (-1,1);
  \draw[red] (0,0) -- (1,1);
  \draw[blue] (0,0) -- (1,-1);
  \end{tikzpicture}
  \captionsetup{justification=centering}
  \caption{Cross-over signature $\mathcal{C}$}
  \label{cross-over signature}
  \end{subfigure}
  \begin{subfigure}[t]{0.25\textwidth}
  \centering
  \begin{tikzpicture}[scale=0.5]
  \draw[black, very thick, fill=black] (-2.15,0.85) rectangle (-1.85,1.15);
  \filldraw (-2,-1) circle (5pt);
  \draw[black, very thick, fill=black] (-0.15,-1.15) rectangle (0.15,-0.85);
  \filldraw (0,1) circle (5pt);
  \draw[black, very thick, fill=black] (1.85,0.85) rectangle (2.15,1.15);
  \filldraw (2,-1) circle (5pt);
  \draw (-3,1) -- (3,1);
  \draw (-2,1) -- (-2,-1);
  \draw (0,1) -- (0,-1);
  \draw (2,1) -- (2,-1);
  \draw (-3,-1) -- (3,-1);
  \end{tikzpicture}
  \captionsetup{justification=centering}
  \caption{Cross-over gadget $G_4$}
  \label{G4}
  \end{subfigure}
  \begin{subfigure}[t]{0.18\textwidth}
  \centering
  \begin{tikzpicture}[scale=0.5]
  \draw[black, very thick, fill=black] (-0.15,0.85) rectangle (0.15,1.15);
  \filldraw (0,-1) circle (5pt);
  \draw (0,-1) -- (0,1);
  \draw (1,1) -- (-1,1);
  \draw (-1,-1) -- (1,-1);
  \end{tikzpicture}
  \captionsetup{justification=centering}
  \caption{Component $A$}
  \label{G4A}
  \end{subfigure}
  \begin{subfigure}[t]{0.18\textwidth}
  \centering
  \begin{tikzpicture}[scale=0.5]
  \draw[black, very thick, fill=black] (-0.15,-1.15) rectangle (0.15,-0.85);
  \filldraw (0,1) circle (5pt);
  \draw (0,-1) -- (0,1);
  \draw (1,1) -- (-1,1);
  \draw (-1,-1) -- (1,-1);
  \end{tikzpicture}
  \captionsetup{justification=centering}
  \caption{$B = $ Up-Down flipped copy of $A$}
  \label{G4B}
  \end{subfigure}
  \captionsetup{justification=centering}
  \caption{Interpolate the cross-over signature $\mathcal{C}$}
\end{figure}

To prove Lemma~\ref{lemma:[1,a,1,a]}, let us define the cross-over  signature $\mathcal{C}$ of arity 4,  illustrated in Figure~\ref{cross-over signature}. It is 0-1 valued, and it takes value 1
iff the two  red dangling edges are equal and 
the two blue dangling edges are equal; furthermore, it is a straddled signature where the two top dangling edges are to be connected to RHS externally and the two bottom dangling edges are to be connected to LHS externally. 
In Figure~\ref{cross-over signature} only one vertex is present
pictorially.
However, when this signature is actually
implemented  or interpolated by some construction, the internal vertices that the two top  dangling edges are incident to
are LHS vertices; and
the  internal vertices that the  two bottom dangling edges 
are incident to
are RHS vertices.

For later convenience, we will write the  signature matrix for $\mathcal{C}$
with
rows  (resp.\ columns)   indexed by $(b_1, b_2) \in \{0, 1\}^2$
corresponding to the dangling edges on the leftside (resp. rightside)
as it appears in Figure~\ref{cross-over signature} (not the LHS, RHS
designation according to the bipartiteness), 
with $b_1$ for the top  edge.
The signature matrix is $\mathcal{C} = \left( \begin{smallmatrix}
    1 & 0 & 0 & 0 \\
    0 & 0 & 1 & 0 \\
    0 & 1 & 0 & 0 \\
    0 & 0 & 0 & 1 
\end{smallmatrix}\right)$. Note that the signature matrix of $\mathcal{C}$ is invariant under a cyclic rotation by 90$^\circ$ of the graph in Figure~\ref{cross-over signature}. The importance of this cross-over signature is conveyed in the following lemma.

We say a signature $f$ can be  planarly constructed or interpolated
if there is a polynomial time construction of planar gadgets or a sequence
of planar gadgets with external dangling edges conforming
to that of $f$ with respect to its bipartite specification,   such that 
the construction implements or interpolates $f$.
\begin{lemma}\label{lemma:cross-over}
For any signature sets $\mathcal{F}, \mathcal{G}$, if the cross-over signature $\mathcal{C}$ can be planarly constructed or interpolated, then $\holant{\mathcal{F}}{\mathcal{G}} \leq_{T} \plholant{\mathcal{F}}{\mathcal{G}}$.
\end{lemma}

\begin{proof}
Given any input signature grid of the problem $\holant{\mathcal{F}}{\mathcal{G}}$, we place it on the plane with possible edges intersecting each other at non-vertices. 
We may assume at most two edges intersect at any point, and the number of such intersections is  polynomially bounded.  
We replace each such intersection  by a copy of $\mathcal{C}$ as follows. 
Note that every
 edge connects a LHS
vertex with a RHS vertex.
Suppose an edge $e = \{u, v\}$ intersects consecutively $k \ge 1$  edges at non-vertices $x_1, \ldots, x_k$, 
where $u$ and $v$ are
from LHS and RHS, respectively. As we traverse from $u$ to $v$, for each $1 \le i \le k$ we label  R and L respectively
just as we enter and leave $x_i$. Labeling in this way for every edge having
such intersections, we find that locally at each intersection point,
cyclically two consecutive edges are labeled R and the other two consecutive edges are labeled L.
This is because at each local intersection point each pair of
incident edges that are \emph{not} cyclically consecutive are always labeled with distinct R $\ne$ L.
A moment reflection shows that the signature
$\mathcal{C}$ can always be used 
with a suitable 
 rotation at each intersection, while
respecting the bipartite structure.
We thus obtain an input of the problem $\plholant{\mathcal{F}}{\mathcal{G}}$ with Holant value unchanged.
\end{proof}



We are now ready to prove Lemma~\ref{lemma:[1,a,1,a]}.

\begin{proof}[Proof of Lemma~\ref{lemma:[1,a,1,a]}]
When $a=0$ or $\pm1$, the problem is in the affine class or degenerate, respectively, and thus in $\operatorname{FP}$ (see~\cite{cai2017complexity} for details of the algorithms). Assume $a \neq 0 $ and $a \neq \pm 1$.
The problem $\holant{[1,a,1,a]}{(=_3)}$ without the planar restriction
is shown to be \numP-hard in \cite{CaiFL21}. By Lemma~\ref{lemma:cross-over}, it suffices to show we can interpolate the cross-over signature $\mathcal{C}$. 

Consider
the gadget $G_4$ in Figure~\ref{G4} where we place 
the signature $[1,a,1,a]$ at the square vertices
and $=_3$ at the circle vertices. Note that $G_4$ is a straddled signature
with the two dangling edges at the top (reps. bottom) to be connected externally
to the RHS (reps.  LHS), just like the cross-over signature $\mathcal{C}$. After normalization 
by a constant $a+a^2 \ne 0$, the signature matrix of $G_4$  is  $\left(\begin{smallmatrix}
    z & 1 & 1 & 1 \\
    1 & 1 & z & 1 \\
    1 & z & 1 & 1 \\
    1 & 1 & 1 & z
\end{smallmatrix}\right)$,
where $z = \frac{1+a^3}{a+a^2} = a + a^{-1} - 1$.
As $a \ne \pm 1$, 
if $a >0$  then
$z =(a^{1/2} - a^{-1/2})^2 +1 > 1$, and if 
$a<0$  then $z = - (|a| + |a|^{-1}) - 1= -(|a|^{1/2} - |a|^{-1/2})^2 -3 
<-3$.
Here the rows  (resp. columns)  are  indexed by $(b_1, b_2) \in \{0, 1\}^2$ in lexicographic order
corresponding to the dangling edges on the leftside (resp. rightside,
as it appears in Figure~\ref{G4}), 
with $b_1$ for the top  edge.
This can be  verified by first computing the signatures for $A$ and $B$ in
Figures~\ref{G4A} and~\ref{G4B},
$A = \left( \begin{smallmatrix}
    1 & 0 & a & 0 \\
    0 & a & 0 & 1 \\
    a & 0 & 1 & 0 \\
    0 & 1 & 0 & a
\end{smallmatrix}
\right)$,
and  $B =\left(\begin{smallmatrix}
    1 & a & 0 & 0 \\
    a & 1 & 0 & 0 \\
    0 & 0 & a & 1 \\
    0 & 0 & 1 & a
\end{smallmatrix}
\right)$ is obtained from $A$ by exchanging both  middle two rows and 
 middle two columns. Then we have $G_4 = A \cdot B \cdot A$, as a matrix product.
Notice that the ``shape" of $G_4$ looks just like $\mathcal{C}$ if we replace 1 by 0, and $z$ by 1. We will exploit this remarkable coincidence
in our proof below.
Note also that the signature matrix of $G_4$ is invariant under cyclic rotations of the gadget.

Now we define a sequence of gadgets $\Gamma_{2s+1}$ of linear size,
which is a sequential composition of
$2s+1$ sub-gadgets, where for odd index $i=1, 3, \ldots, 2s+1$
we use  $G_4$, and for even index $i=2, 4, \ldots, 2s$
we use a  $180^\circ$-rotated copy of  $G_4$, and 
we merge the rightside two edges
of the $i$th sub-gadget with the leftside two edges
of the $(i+1)$th sub-gadget.
This  sequential composition  satisfies
the bipartite restriction. As the
rotated copy of  $G_4$ has the 
same signature matrix as that of $G_4$,
the signature matrix of $\Gamma_{2s+1}$  is $G_4^{2s+1}$,
the $(2s+1)$th power of $G_4$.
We can show
that it has the form (after normalization) 
$\Gamma_{2s+1} = \left(\begin{smallmatrix}
    x_s & 1 & 1& 1 \\
    1 & 1 & x_s & 1 \\
    1 & x_s & 1 & 1\\
    1 & 1 & 1 & x_s
\end{smallmatrix}\right)$,
where $\{x_s\}_{s \ge 0}$ are defined by a recurrence, with $x_0 = z$ and \[x_{s+1} = \frac{6+6z+3x_s+z^2\cdot x_s}{7+4z+z^2+2x_s+2z\cdot x_s}. \]

We are going to show that $x_s$'s are pairwise distinct. 
First suppose $a >0$. We have $x_{s+1} - 1
=\frac{(z-1)^2(x_s -1)}{7+4z+z^2+2x_s+2z\cdot x_s}$,
which
shows inductively that $x_s > 1$
 for all $s \in \mathbb{N}$, as the denominator is clearly positive. 
 Next, $\frac{x_{s+1} - 1}{x_s -1} = \frac{(z-1)^2}{(z+2)^2 + 3+2x_s(z+1)} <1$, as $x_s, z>1$.
 It follows that $x_{s+1}  < x_s$
 and hence pairwise distinct.
 Now suppose $a <0$. Inductively assume $x_s < -3$, which is true
 at $x_0 = z < -3$.
 We have $x_{s+1} + 3
 =\frac{(z+3)^2 (x_s +3)}{7+4z+z^2+2x_s+2z\cdot x_s}$.
 The denominator $7+4z+z^2+2x_s+2z\cdot x_s
 =(z+2)^2 +3 + 2 (z+1) x_s >0$, as $z <-3$ and inductively also $x_s <-3$.
 Then we  have $x_{s+1} + 3 <0$ since $x_s +3 < 0$.
 Now the denominator is $(z+2)^2 +3 + 2 (z+1) x_s > |z+2|^2 > |z+3|^2$
 as $z <-3$. Hence $\frac{x_{s+1} + 3}{x_s +3} =
 \frac{(z+3)^2}{(z+2)^2 +3 + 2 (z+1) x_s} <1$.
And so,  $x_{s+1} > x_s$, and  in particular they are pairwise distinct.


Note also that the number of bits required to represent $x_s$'s is polynomially bounded in the size of the input
because $x_s$'s come from, by definition, sums of at most $2^{n^{O(1)}}$ terms, each a product of $n^{O(1)}$ factors. 

Given any signature grid $\Omega$ where the cross-over signature $\mathcal{C}$ appears $n$ times, we construct signature grids $\Omega_s$, $0 \leq s \leq n$,
by replacing each copy of $\mathcal{C}$ by $\Gamma_{2s+1}$ while respecting
the bipartite restrictions. 
 We now stratify the assignments in the Holant sum for $\Omega$ according to the number $i$, $0 \leq i \leq n$, of total times that the input of $\mathcal{C}$ is (0,0,0,0), (1,0,1,0), (0,1,0,1), or (1,1,1,1) in cyclic order (these are the only inputs
 to  $\mathcal{C}$ with nonzero evaluations). Let $c_i$ be the sum over all corresponding assignments of the products from other signatures with this restriction of $i$. Then we have 
$\operatorname{Holant}(\Omega) = c_n$,
and
\begin{equation}\label{Vandermonde}
   \operatorname{Holant}(\Omega_s) = \sum\limits_{i=0}^n x^i_s \cdot c_i.
\end{equation}
Since  $x_s$'s are pairwise distinct, (\ref{Vandermonde}) is a full ranked
Vandermonde system, and we can solve for all $c_i$ in polynomial time, and 
in particular compute $c_n$, from the values of $\operatorname{Holant}(\Omega_s)$, $0 \leq i \leq n$.
\end{proof}

Hereafter, we say $[1,a,b,c]$ is \numP-hard or in FP to mean the problem $\plholant{[1,a,b,c]}{(=_3)}$ is \numP-hard or in FP. We shall invoke the following theorem in~\cite{MKJYC} when proving our results:
\begin{theorem}[Kowalczyk \& Cai] \label{previous 2-3}
Suppose $a,b \in \mathbb{C}$, and let $X = ab$, $Z = \left( \frac{a^3+b^3}{2} \right)^2$. Then $\plholant{[a,1,b]}{\left(=_3\right)}$ is \numP-hard except in the following cases, for which the problem is in $\operatorname{FP}$.

\begin{enumerate}
\item $X=1$;
\item $X=Z=0$;
\item $X=-1$ and $Z=0$;
\item $X=-1$ and $Z=-1$;
\item $X^3 = Z$.
\end{enumerate}
\end{theorem} 

By restricting Theorem~\ref{previous 2-3} to real numbers, we have the following corollary.

\begin{corollary} \label{2-3}
Suppose $a,b \in \mathbb{R}$, then $\plholant{[a,1,b]}{\left(=_3\right)}$ is \numP-hard except in the following cases, for which the problem is in $\operatorname{FP}$.

\begin{enumerate}
\item $ab=1$;
\item $a=1$ and $b=-1$;
\item $a=-1$ and $b=1$;
\item $a=b$.
\end{enumerate}
\end{corollary}

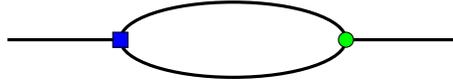
\begin{figure}[ht] 
\centering
  \begin{tikzpicture}
        \draw [very thick] (-1,0.5) -- (0.5,0.5);
        
        \draw[very thick] (2,0.5) ellipse (1.5cm and 0.5cm);
        \draw[very thick] (3.5,0.5) -- (5,0.5);
        \filldraw[fill= blue] (0.4,0.4) rectangle (0.6,0.6);
        \filldraw[fill=green] (3.5,0.5) circle (0.1cm);
    \end{tikzpicture}
    \captionsetup{justification=centering}
    \caption{Gadget $G_1$}
  \label{f1}
\end{figure}

Consider the binary straddled gadget $G_1$ in Figure \ref{f1}. 
Parallel edges are allowed. Its signature is $G_1=\left[\begin{smallmatrix} 1 & b \\ a & c \end{smallmatrix}\right]$, where  $G_1(i,j)$ (at row $i$ column $j$) is the value of this gadget when the left dangling edge (from the ``square") and the right dangling edge (from the ``circle" $(=_3)$) are assigned $i$ and $j$ respectively, for $i,j \in\{0,1\}$.
Iterating $G_1$ sequentially $k$ times is represented by the
matrix power $G_1^k$. It turns out that it is very useful
 either to produce directly  or to obtain by interpolation
a rank \emph{deficient} straddled signature, which would in most cases allow us to obtain unary signatures on either side.
With unary signatures we can connect to a ternary signature
to produce binary signatures on one side and then apply
Corollary~\ref{2-3}. 

The following lemma is proved in~\cite{CaiFL21}.

\begin{lemma} \label{3.1.1}
Given the binary straddled signature $G_1=\left[\begin{smallmatrix} 1 & b \\ a & c \end{smallmatrix}\right]$, we can interpolate the degenerate binary straddled signature $\left[\begin{smallmatrix}y & xy \\ 1 & x\end{smallmatrix}\right]$,
provided that $c\neq ab$, $a\neq 0$, $\Delta =\sqrt{(1-c)^2 + 4ab} \neq 0$ and $\frac{\lambda}{\mu}$ is not a root of unity, where
$\lambda=\frac{-\Delta+(1+c)}{2}$, $\mu=\frac{\Delta+(1+c)}{2}$ are the
 two eigenvalues,  
 and $x=\frac{\Delta-(1-c)}{2 a}$ and $y=\frac{\Delta+(1-c)}{2 a}$.
\end{lemma}

Given a degenerate binary straddled signature, we want to use it as unary signatures \emph{in a planar way}. It is only in this step that we need our P3EM theorem. More concretely, in the next lemma we show how we can \emph{essentially} separate a binary straddled signature to get a unary
signature.

\begin{lemma} \label{3.1.2}
For $\PlHolant( \, [1,a,b,c] \, | =_3)$,\ 
$a,b,c \in \mathbb{Q}$, $a\neq 0$,  with the availability of the binary degenerate straddled signature $\left[\begin{smallmatrix}y & xy \\ 1 & x\end{smallmatrix}\right]$ where 
$x=\frac{\Delta-(1-c)}{2 a}$,  $y=\frac{\Delta+(1-c)}{2 a}$ and $\Delta=\sqrt{(1-c)^2 + 4ab}$, we have the following reductions
\begin{enumerate}
    \item $\PlHolant ( \, [1+ax,a+bx,b+cx] \, | =_3) \leq_T \PlHolant ( \, [1,a,b,c] \, | =_3)$ except for 2 cases:
    $[1,a,a,1]$, $[1,a,-1-2a, 2+3a]$;
    \item $ \PlHolant ( \, [1,a,b,c] \, | [y,0,1]) \leq_T \PlHolant ( \, [1,a,b,c] \, | =_3)$.
    
\end{enumerate}  
\end{lemma}

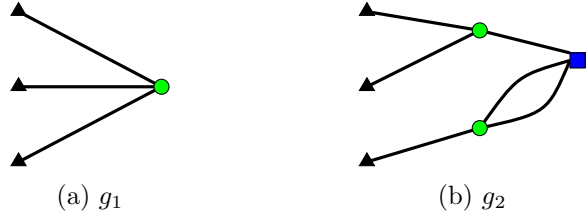
\begin{figure}[ht]
    \centering
    \begin{subfigure}[b]{0.3\textwidth}
    \centering
    \begin{tikzpicture}  
        \draw[very thick] (0.3,5.05) -- (2.2,4.05);
        \draw[very thick] (0.3,4.05) -- (2.2,4.05);
        \draw[very thick] (0.3,3.05) -- (2.2,4.05);
        \filldraw[fill=black] (0.3,5.17)--(0.2,5)--(0.4,5)--cycle;
        \filldraw[fill=black] (0.3,4.17)--(0.2,4)--(0.4,4)--cycle;
        \filldraw[fill=black] (0.3,3.17)--(0.2,3)--(0.4,3)--cycle;
        \filldraw[fill=green] (2.2, 4.05)circle(0.1);
    \end{tikzpicture}
    \captionsetup{justification=centering}
    \caption{$g_1$}
    \label{h1}
    \end{subfigure}
    \begin{subfigure}[b]{0.3\textwidth}
    \centering
    \begin{tikzpicture}  
        \draw[very thick] (0.3,5.05)--(1.8,4.8);
        \draw[very thick] (0.3,4.05)--(1.8,4.8);
        \draw[very thick] (0.3,3.05)--(1.8,3.5);
        \draw[very thick] (1.8,4.8) -- (3,4.5);
        \draw[very thick] (1.8,3.5).. controls (2.3, 4.2) ..(3,4.4);
        \draw[very thick] (1.8,3.5) .. controls (2.7, 3.7)..(3,4.4);
        \filldraw[fill=black] (0.3,5.17)--(0.2,5)--(0.4,5)--cycle;
        \filldraw[fill=black] (0.3,4.17)--(0.2,4)--(0.4,4)--cycle;
        \filldraw[fill=black] (0.3,3.17)--(0.2,3)--(0.4,3)--cycle;
        \filldraw[fill=green] (1.8,4.8) circle (0.1);
        \filldraw[fill=green] (1.8,3.5) circle (0.1);
        \filldraw[fill=blue] (3,4.3)rectangle(3.2,4.5);
    \end{tikzpicture}
    \captionsetup{justification=centering}
    \caption{$g_2$}
    \label{h2}
    \end{subfigure}
    \captionsetup{justification=centering}
    \caption{Two gadgets where each triangle represents the unary gadget $[y,1]$}
    \label{3g4y1}
\end{figure}

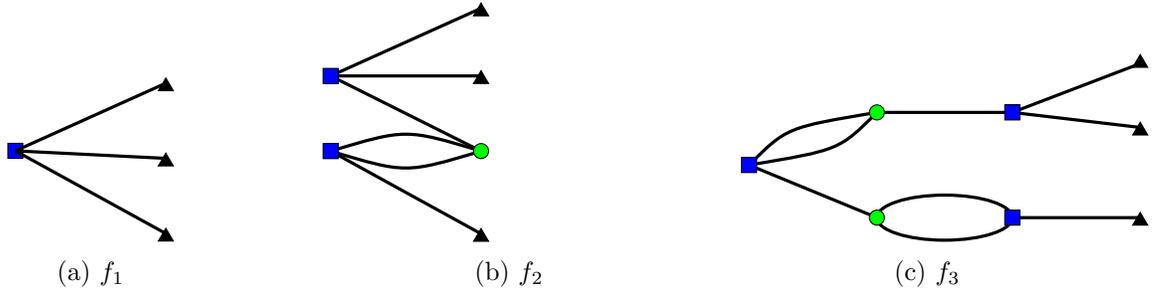
\begin{figure}[ht]
    \centering
    \begin{subfigure} [b] {0.3\textwidth}
    \centering
    \begin{tikzpicture}  
        \filldraw[fill=blue] (0.2,2.2) rectangle (0.4,2.4);
        \filldraw[fill=black] (2.3,3.27)--(2.2,3.1)--(2.4,3.1)--cycle ;
        \filldraw[fill=black] (2.3,2.27)--(2.2,2.1)--(2.4,2.1)--cycle;
        \filldraw[fill=black] (2.3,1.27)--(2.2,1.1)--(2.4,1.1)--cycle;
        \draw[very thick] (0.3,2.3)--(2.3,3.2);
        \draw[very thick] (0.3,2.3)--(2.3,2.2);
        \draw[very thick] (0.3,2.3)--(2.3,1.2);
    \end{tikzpicture}
    \captionsetup{justification=centering}
    \caption{$f_1$}
    \label{e1}
    \end{subfigure}
        \hspace{10pt}
    \begin{subfigure}  [b]{0.3\textwidth}
    \begin{tikzpicture}   
        \filldraw[fill=black] (2.3,5.27)--(2.2,5.1)--(2.4,5.1)--cycle;
        \filldraw[fill=black] (2.3,4.37)--(2.2,4.2)--(2.4,4.2)--cycle;
        \filldraw[fill=black] (2.3,2.27)--(2.2,2.1)--(2.4,2.1)--cycle;

        \draw[very thick] (0.3,4.3) -- (2.3,5.2);
        \draw[very thick] (0.3,4.3)--(2.3,4.3);
        \draw[very thick] (0.3,4.3)--(2.3,3.3);
        \draw[very thick] (0.3,3.3) ..controls (1.3,3.6)..  (2.3, 3.3);
        \draw[very thick] (0.3,3.3) .. controls (1.3, 3.0).. (2.3, 3.3);
        \draw[very thick] (0.3,3.3) -- (2.3, 2.2);
        \filldraw[fill=blue] (0.2,4.2) rectangle (0.4,4.4) ;  
        \filldraw[fill=blue] (0.2,3.2) rectangle (0.4,3.4) ;  
        \filldraw[fill=green] (2.3,3.3) circle (0.1cm) ;
    \end{tikzpicture}
    \captionsetup{justification=centering}
    \caption{$f_2$}
    \label{e2}
    \end{subfigure}
    \hspace{10pt}
    \begin{subfigure}  [b]{0.3\textwidth}
    \centering
    \begin{tikzpicture}  
        \draw[very thick] (0.3,4.3)..controls(0.8,4.8)..(2,5);
        \draw[very thick] (0.3,4.3)..controls(1.5, 4.5 )..(2,5);
        \draw[very thick] (0.3,4.3)--(2,3.6);
        \draw[very thick] (3.8,5)--(2,5);
        
        \draw[very thick] (3.8,5)--(5.5,5.65);
        \draw[very thick] (3.8,5)--(5.5,4.8);
        \draw[very thick] (3.8,3.6)--(5.5,3.6);
        
        \draw[very thick]  (2.9, 3.6) ellipse(0.9 and 0.3); 
        \filldraw[fill=blue] (0.2,4.2) rectangle(0.4,4.4); 
        \filldraw[fill=green] (2,5) circle (0.1cm);     
        \filldraw[fill=green] (2,3.6) circle(0.1cm);    
        \filldraw[fill=blue] (3.7,4.9) rectangle (3.9,5.1);   
        \filldraw[fill=blue] (3.7,3.5) rectangle (3.9, 3.7);  
        \filldraw[fill=black] (5.5,5.77)--(5.4,5.6)--(5.6,5.6)--cycle; 
        \filldraw[fill=black] (5.5,4.87)--(5.4,4.7)--(5.6,4.7)--cycle; 
        \filldraw[fill=black] (5.5,3.67)--(5.4,3.5)--(5.6,3.5)--cycle; 
    \end{tikzpicture}
    \captionsetup{justification=centering}
    \caption{$f_3$}
    \label{e3}
    \end{subfigure}
    \captionsetup{justification=centering}
    \caption{Three gadgets where each triangle represents the unary gadget $[1,x]$}
    \label{4g41x}
\end{figure}

\begin{proof}
The signature $[1+ax,a+bx,b+cx]$ is the binary signature obtained by connecting $[1,a,b,c]$ on LHS with $[1,x]$ on RHS. For simplicity, we denote $f : = [1,a,b,c]$ and $f^{\flat} := [1+ax,a+bx,b+cx]$.

To prove the first reduction $\PlHolant ( \,  f^{\flat}\, | =_3) \leq_T \PlHolant ( \, f \, | =_3)$, 
consider any input instance of the LHS  problem.  Let 
$G$ be its underlying 2-3 bipartite plane graph. 
We may assume
$G$ is connected, as the Holant value of $G$ is 
the product of the Holant values of its connected components.
We can view $G$ as the edge-vertex incidence graph of
a plane 3-regular graph 
$G'$, where every vertex of degree 2 in $G$ on the LHS is viewed
as an edge in $G'$. One can also
obtain $G'$ 
by merging the two edges
incident to every vertex of degree 2 in $G$.
If $G'$ is isomorphic to one of the two exceptions in Theorem~\ref{thm:P3EM}, then the size of $G$ is constant and we can compute the Holant value directly. Otherwise, we construct an input of the RHS problem as follows. 
We first obtain the degenerate binary straddled signature $D = \left[\begin{smallmatrix}y & xy \\ 1 & x\end{smallmatrix}\right]$
in $\plholant{[1,a,b,c]}{(=_3)}$. 
Then for every edge of $G'$, which is
assigned  the binary
signature $f^{\flat}$,
we replace it by a copy of $[1,a,b,c]$ and  connecting it
with
the edge of $D$ that corresponds to $[1,x]$.
This leaves 
1 dangling edge from each copy of $D$, each 
edge functionally equivalent to
a unary $[y,1]$ on LHS. They need to be connected to other $(=_3)$ signatures in a planar way.
Now we apply Theorem~\ref{thm:P3EM} to the
3-regular plane graph $G'$, which constructively
assigns every edge of $G'$ one of the two incident
faces such that we have a P3EM. We then add a suitable number
of $(=_3)$ and $f$ in each face
and connect them to \emph{exactly} 3 copies of
$[y,1]$ as shown in Figures \ref{h1} and  \ref{h2}.
Theorem~\ref{thm:P3EM} guarantees that
this can be done in a planar way.
Each connection produces a multiplicative factor
$g_1 =y^3 + 1$ in Figure \ref{h1} and a multiplicative factor $g_2 = y^3 + by^2+ay+c$ in Figure \ref{h2}. It can be directly checked that\footnotemark \footnotetext{We use Mathematica to solve the system of equation \begin{equation*}\begin{cases}
g_1 = y^3 + 1 = 0\\
g_2 = y^3 + by^2+ay+c = 0\\
y=\frac{\Delta+(1-c)}{2 a}
    \end{cases}
\end{equation*}
The empty solution of the system is proved by cylindrical decomposition, an algorithm for Tarski's theorem on real-closed fields.}, for $y=\frac{\Delta+(1-c)}{2 a}$, at least one of these factors is nonzero unless $y=-1$, and in that case the signature has the form $[1,a,a,1]$ or $[1,a,-2a-1,3a+2]$.
The proof of the  reduction is complete.

For the reduction $ \PlHolant ( \, [1,a,b,c] \, | [y,0,1]) \leq_T \PlHolant ( \, [1,a,b,c] \, | =_3)$, we use the same P3EM argument as above. Therefore it suffices to ``absorb'' those dangling unaries $[1,x]$ to produce some nonzero factor. We claim that at least one of the connection gadgets in Figures \ref{e1}, \ref{e2}, and \ref{e3} creates a nonzero global factor. 
The factors of these four gadgets are \begin{equation*}\begin{cases}
f_1 = cx^3+3bx^2+3ax+1 \\
f_2=(ab+c)x^3+ (3bc+2a^2+b)x^2 +(2b^2+ac + 3a)x+ab+1\\
\begin{split}f_3 = 
    & (ab+2abc+c^3)x^3+(2a^2+b+2a^2c+3ab^2+bc+3b^2c)x^2 \\
    & +(3a+3a^2b+ac+2b^2+2b^2c+ac^2)x+1+2ab+abc
\end{split}
    \end{cases}
\end{equation*} respectively.
 By setting the three formulae to be 0 simultaneously together with the condition $x=\frac{\Delta-(1-c)}{2 a}$, with $a\neq 0$, $a,b,c\in\mathbb{Q}$, we found that there is no common solution. The proof is now complete.
\end{proof}

We note that signatures of the form $[3x+y,-x-y,-x+y,3x-y]$ for some $x,y$ are exactly either of the form $[1,a,-2a-1, 3a+2]$ 
after normalization, or of the form
$[0,a,-2a,3a]$, for some $a$. 


\begin{remark}\label{remark8}
Just before Lemma~\ref{y11x} we stated that we could
\emph{essentially} separate a binary straddled signature to get a unary.
This statement is delicate. 
Getting unrestricted use of the unary $[1, x]$ on RHS would be
$\plholant{f}{(=_3), [1,x]}$.
The following two problems are equivalent.
$$\plholant{f, f^{\flat}}{(=_3)}
\equiv_T
\plholant{f}{(=_3), [1,x]}.
$$
This is because for the second problem, every occurrence of
$[1,x]$ is connected to $f$ to produce $f^{\flat}$,
and conversely for  the first problem, every occurrence of
$f^{\flat}$ can be replaced by connecting a copy of $f$  with
$[1,x]$.
However, we do not claim that the problem $\plholant{f}{(=_3), [1,x]}$ is reducible to
$\plholant{f}{(=_3)}$, which 
is the following
stronger reduction than what we showed:
$$\plholant{f, f^{\flat}}{(=_3)} \leq_{T} \plholant{f}{(=_3)}.$$
The issue is that now the 
 input graph for the LHS problem is not an edge-vertex incidence graph for a 3-regular plane graph, and so
we cannot apply Theorem~\ref{thm:P3EM} as before. 
If we merge the two incident
edges of all degree 2 vertices (assigned the binary
signature $f^{\flat}$) we do get a planar 3-regular graph.
But this graph may still have degree 3 vertices labeled $f$,
and not every edge comes from merging a degree 2 vertex that
was labeled $f^{\flat}$. Thus,
not every edge participates in a 3-way perfect matching.
In summary, a degenerate binary straddled signature is not completely equivalent to a unary $[1,x]$ on RHS.
We further remark  that, if this were true, we would have
a much simpler proof of Theorem~\ref{thm:[0,1,0,0]}.
\end{remark}

\begin{remark}
Reader should think of Lemma 9 mainly as an illustration of what we will call the \emph{P3EM argument}. The main take-away is that we can separate a degenerate binary straddled signature to get unaries so long as we use one of them on \emph{every} ternary signature on one side and the remaining dangling unaries can be absorbed to create a nonzero global factor. For example, in the proof of Theorem \ref{g1works1}, we are in fact using the gadget depicted in Figure~\ref{non_li}. For the sake of simplicity for presentation, we will say ``interpolate'' $[1,x]$ on RHS or $[y,1]$ on LHS hereafter while the reader is welcome to check the delicate issue in Remark~\ref{remark8} is taken care.
\end{remark}

The following proposition is proved in~\cite{CaiFL21}.

\begin{proposition} \label {g1-rou}
For $G_1 = \left[\begin{smallmatrix}1 & b \\ a & c\end{smallmatrix}\right]$, with $a,b,c \in \mathbb{Q}$, if it is non-singular (i.e., $c \ne ab$), then it has two nonzero eigenvalues $\lambda$ and $\mu$. The ratio  $\lambda/\mu$
is not a root of unity \emph{unless} at least one of the following conditions holds: \begin{equation}\label{c1eq} \begin{cases}

c+1=0\\
ab+c^2 + c + 1= 0\\
2ab + c^2 + 1= 0\\
3ab+c^2-c+1=0\\
4ab + c^2-2c+1=0\\
\end{cases}
\end{equation}
\end{proposition}

\begin{figure}[ht]
        \centering
        \begin{tikzpicture}  
            \draw[very thick] (0,4.3)--(4.3,4.3);
            \draw[very thick] (1.3,2.7)--(1.3,4.3);
            \draw[very thick] (3.3,2.7)--(3.3,4.3);
            \draw[very thick] (2.3,2.7)ellipse(1 and 0.4);
            \filldraw[fill=blue](1.2,4.2)rectangle(1.4,4.4);
            \filldraw[fill=green] (3.3,4.3)circle(0.1);
            \filldraw[fill=blue] (3.2,2.6)rectangle(3.4,2.8);
            \filldraw[fill=green] (1.3,2.7)circle(0.1);
        \end{tikzpicture}
        \captionsetup{justification=centering}
        \caption{Binary straddled signature $G_2$}
        \label{fig:g2}
\end{figure}
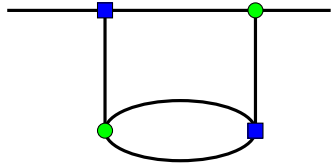

Now we introduce a new binary straddled signature $G_2$ as shown in Figure~\ref{fig:g2}. The signature matrix of 
$G_2$ is $
\left[\begin{smallmatrix}
w & b' \\ a' & c'
\end{smallmatrix}\right]$,
where $w=1+ab$, $a'=a+b^2$, $b'=a^2+bc$ and $c'=ab+c^2$. Similar to $G_1$, we have
$\Delta' = \sqrt{(w-c')^2+4a'b'}$, two eigenvalues $\lambda'=\frac{-\Delta' + (w+c')}{2}$ and $\mu' = \frac{\Delta' + (w+c')}{2}$. If $a'\ne 0$, we have $x'=\frac{\Delta'-(w-c')}{2a'}$, $y'=\frac{\Delta'+(w-c')}{2a'}$ and if further $\Delta' \ne 0$ we can write its Jordan Normal Form as 
\begin{equation}\label{eqn:G2-JNF}
G_2 = \left(\begin{array}{ll}w & b' \\ a' & c'\end{array}\right)=\left(\begin{array}{cc}-x' & y' \\ 1 & 1\end{array}\right)\left(\begin{array}{ll}\lambda' & 0 \\ 0 & \mu'\end{array}\right)\left(\begin{array}{cc}-x' & y' \\ 1 & 1\end{array}\right)^{-1}.
\end{equation}

Similar to Proposition \ref{g1-rou}, we have the following claim on $G_2$.
 \begin{proposition} \label{c3}
If the signature matrix of $G_2$  is non-degenerate, then
 the ratio  $\lambda'/\mu'$ of its eigenvalues 
is not a root of unity \emph{unless} at least one of the following conditions holds, where 
$A = w + c', B = (c'-w)^2 + 4a'b'$.
\begin{equation} \label{c3eq}\begin{cases}
A=0\\
B=0 \\
A^2+B=0\\
A^2+3B=0\\
3A^2 + B = 0\\
\end{cases}
\end{equation}
\end{proposition}

\begin{lemma} \label{y11x}
Suppose  $a,b,c \in \mathbb{Q}$, $a\neq0$ and $c\neq ab$ and $a,b,c$ do not satisfy any condition in  (\ref{c1eq}). Let $x=\frac{\Delta-(1-c)}{2 a}$,  $y=\frac{\Delta+(1-c)}{2 a}$ and $\Delta=\sqrt{(1-c)^2 + 4ab}$. Then 
for $\Holant ( \, [1,a,b,c] \, | =_3)$,
\begin{enumerate}
    \item we can  interpolate $[y,1]$ on LHS;
    \item we can  interpolate $[1,x]$ on RHS except for 2 cases:
    $[1,a,a,1]$, $[1,a,-1-2a, 2+3a]$.
\end{enumerate}  
\end{lemma}
\begin{proof}
This lemma follows from  Lemma \ref{3.1.1} and Lemma \ref{3.1.2}
using the binary straddled gadget $G_1$ with signature matrix $\left[\begin{smallmatrix}1 & b \\ a & c\end{smallmatrix}\right]$.  Note that $c\neq ab$ indicates that matrix $G_1$ is non-degenerate, and $\lambda/\mu$ not being a root of unity is equivalent to none of the equations in (\ref{c1eq}) holds.
\end{proof}

We have similar statements corresponding to $G_2$. When the signature matrix is non-degenerate and does not satisfy any condition in  (\ref{c3eq}), we can  interpolate the corresponding $[y',1]$ on LHS, and we can also interpolate the corresponding $[1,x']$ on RHS except when $y'=-1$.

\begin{definition} \label{d1}
For $\PlHolant ( \, [1,a,b,c] \, | =_3)$, with $a,b,c \in \mathbb{Q}$, $a \neq 0$, we say a binary straddled gadget $G$ \emph{works} if the signature matrix of $G$ is non-degenerate and the ratio of its two eigenvalues $\rlm$ is not a root of unity.
\end{definition}

\noindent
\begin{remark}
    
  Explicitly, the condition that $G_1$ \emph{works} is that $c \neq ab$ and $a,b,c$ do not satisfy any condition in (\ref{c1eq}), which is just the assumptions in Lemma \ref{y11x}. $G_1$ \emph{works} implies that it can be used to interpolate  $[y,1]$ on LHS, 
and to interpolate $[1,x]$ on RHS with two exceptions 
for which we already proved the dichotomy. The $x,y$ are as stated in Lemma \ref{y11x}.

Similarly, when the binary straddled gadget $G_2$ \emph{works}, for the corresponding values $x'$ and $y'$,
we can  interpolate $[y',1]$ on LHS,
and we can interpolate  $[1,x']$ on RHS except when $y'=-1$.
 \end{remark}

\begin{figure}[ht] 
\centering
\begin{tikzpicture}  
\draw[very thick] (0.4,7.3)--(6.2,7.3);
\draw[very thick] (1.3,7.3)--(3.3,3.8);
\draw[very thick] (5.3,7.3)--(3.3, 3.8);
\draw[very thick] (3.3, 3.8)--(3.3, 2.9);
\draw[very thick] (3.3,6.1)--(3.3,7.3);
\draw[very thick] (3.3,6.1)--(4.3, 5.55);
\draw[very thick] (3.3,6.1)--(2.3,5.55);
\filldraw[fill=blue] (1.2,7.2) rectangle(1.4,7.4);  
\filldraw[fill=blue] (5.2,7.2) rectangle(5.4,7.4);  
\filldraw[fill=blue] (3.2,3.7) rectangle(3.4,3.9); 
\filldraw[fill=green] (3.3,7.3)circle(0.1);
\filldraw[fill=green] (4.3,5.55)circle(0.1);
\filldraw[fill=green] (2.3,5.55)circle(0.1);
\filldraw[fill=blue] (3.2, 6.0) rectangle(3.4, 6.2);     
\end{tikzpicture}
\captionsetup{justification=centering}
  \caption{A ternary gadget $G_3$}
  \label{f4}
\end{figure}
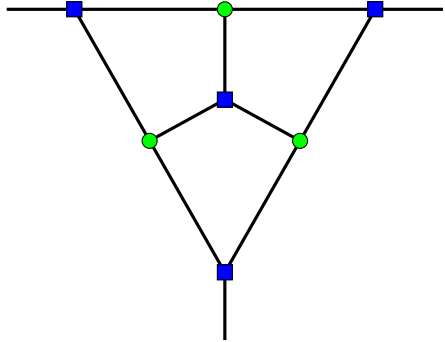

The ternary gadget $G_3$ in Figure \ref{f4} will be used in the proof here and later.

The unary signatures  $\Delta_{0}=[1,0]$ and $\Delta_{1}=[0,1]$
are called  the pinning signatures because they ``pin''
a variable to 0 or 1. Another useful unary signature is $\Delta_{2} = [1,1]$.
One good use of having unary signatures is that we can use Lemma~\ref{3.2} to get these three signatures. They 
are helpful as the following lemma shows.

\begin{lemma} \label{after-unary}
If $\Delta_{0}$, $\Delta_{1}$ and $\Delta_{2}$ can be interpolated on the RHS in
$\PlHolant ( \, [1,a,b,c] \, | =_3)$, 
where $a,b,c \in \mathbb{Q}$, $ab \neq 0$,
then  the problem is \numP-hard unless $[1,a,b,c]$ is affine or degenerate, in which cases it is in FP.
\end{lemma}
\begin{proof}
Connecting 
$[1,0]$, $[0,1]$ to $[1,a,b,c]$ on LHS respectively, we  get binary signatures $[1,a,b]$ and $[a,b,c]$. Then we can apply Corollary \ref{2-3},
and the problem is \numP-hard unless both $[1,a,b]$ and $[a,b,c]$ are in FP. When $ab\neq 0$, both $[1,a,b]$ and $[a,b,c]$ are in FP only when $[1,a,b,c]$ is (1) degenerate, i.e. $b=a^2$ and $c = a^3$, in which case the problem is in FP, (2) of the form $[1,a,1,a]$ which is resolved by Lemma~\ref{lemma:[1,a,1,a]} (when $a = \pm1$, $[1,a,1,a]$ is degenerate), (3) of the form $[1,a,a^2,a]$ which we will resolve later in this proof,  (4) $[1,a,b,c] = [1,1,1,-1]$ or $[1,-1,1,1]$ which we will resolve later in this proof, or (5) $[1,a,b,c] = [1,1,-1,-1]$ or $[1,-1,-1,1]$ which are affine and hence in FP.

For case (2), if we connect $[1,1]$ to $[1,a,a^2,a]$, we get a binary signature $[1+a,a+a^2,a+a^2]$ which after normalization (when $a\neq -1$) is $[1,a,a]$. This problem is \numP-hard by Corollary~\ref{2-3} unless $a=\pm1$, in which cases the problem $[1,a,b,c]$ is $[1,1,1,1]$ or $[1,-1,1,-1]$ which are degenerate and thus in FP.

For case (3), due to the symmetry by flipping 0 and 1 in the  signature, it suffices to consider only $f= [1,1,1,-1]$ and $g=[1,-1,1,1]$; they   are neither affine nor degenerate. For both  $f$ and $g$ we use the gadget $G_3$ to produce  ternary signatures  $f' = [1,1,3,3]$
and $g' = [1,1,-1,3]$ respectively.  Neither are among  the exceptional cases above.
So $\Holant ( \, f \, | =_3)$ and $\Holant ( \, g\, | =_3)$ are both \numP-hard.
\end{proof}

\vspace{.1in}

The following lemma lets us interpolate arbitrary unary signatures on RHS, in particular  $\Delta_{0}$, $\Delta_{1}$ and $\Delta_{2}$, from a binary gadget with a straddled signature and a suitable unary signature $s$ on RHS.  

\begin{lemma}[Vadhan, \cite{vadhan2001complexity}] \label{3.2}
Let $M\in \mathbb{R}^{2\times 2}$ be a non-singular signature matrix for a binary straddled gadget which is diagonalizable with distinct eigenvalues, and $s=[a,b]$ be a unary signature on RHS that is not a row eigenvector of $M$. Then $\{s\cdot M^j\}_{j\geq 0}$ can be used to interpolate any unary signature on RHS.
\end{lemma}

\subsection{Dichotomy for
$[1, a, b, c]$ when $ab\neq 0$ and $G_1$ works}  \label{large2}
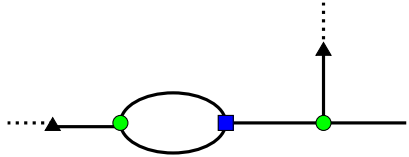
\begin{figure}[h!] 
\centering
\begin{tikzpicture}  
\draw[very thick] (2,3.3)ellipse(0.7 and 0.4);
\filldraw[fill=black] (0.4,3.37)--(0.3,3.2)--(0.5,3.2)--cycle; 
\draw[very thick] (0.34,3.25)--(1.3,3.25);
\draw[very thick] (2.7,3.3)--(5.1,3.3);
\draw[very thick] (4,4.3)--(4,3.3);
\filldraw[fill=green] (1.3, 3.3)circle(0.1);  
\filldraw[fill=blue] (2.6,3.2) rectangle(2.8,3.4);  
\filldraw[fill=green] (4,3.3)circle(0.1);
\filldraw[fill=black] (4,4.37)--(3.9,4.2)--(4.1,4.2)--cycle;
\draw[very thick, dotted] (0.4,3.3) -- (-0.2,3.3);
\draw[very thick, dotted] (4, 4.3) -- (4,4.9);
\end{tikzpicture}
\captionsetup{justification=centering}
  \caption{Non-linearity gadget, where a triangle represents the unary gadget $[y,1]$}
  \label{non_li}
\end{figure}

Let us introduce a \emph{non-linearity} gadget in Figure \ref{non_li}. If we place in the non-linearity gadget the binary degenerate straddled signature $D$ on triangles (in the way that respects the bipartite structure), $f$ on squares and $(=_3)$ on circles, we get its signature
 $[1,x] \otimes [1,x] \otimes [y^2+yb,ya+c]$.
Note that it is a ternary planar gadget on RHS.
The following two lemmas  will be used in the proof of  Theorem~\ref{g1works1}.

\begin{lemma} \label{1-1 X=1}
Let $a,b, c \in \mathbb{Q}$, $ab \neq 0$,  and
satisfy \\
\emph{(con1)} $a^3 - b^3 - ab(1-c) =0$ and\\
\emph{(con2)} $a^3+ab+2b^3=0$.\\
Then $\PlHolant ( \, [1,a,b,c] \, | =_3)$ is
\numP-hard unless it is $[1,-1,1,-1] = [1, -1]^{\otimes 3}$ which is degenerate, or a matchgate, in both cases the problem is in FP. 
\end{lemma}
\begin{proof}
If $a+b^2=0$  in addition to (\emph{con1}) and (\emph{con2})
with  $ab \neq 0$, then $[1,a,b,c] = [1,-1,1,-1]$ which is degenerate. Now we assume $a+b^2\ne0$. Here we use Gadget $G_2$. 

First assume $G_2$ \emph{works}. 
Using  $a+b^2\ne 0$
together with (\emph{con1}) and  (\emph{con2}), we 
can verify that 
$\Delta = \sqrt{4(a+b^2)(a^2+bc)+(c^2-1)^2} \not =0$, and
we can write the Jordan Normal Form  $$G_2 = \left(\begin{array}{ll}1+ab & a^2+bc \\ a+b^2 & ab+c^2
\end{array}\right)=\left(\begin{array}{cc}-x & y \\ 1 & 1\end{array}\right)\left(\begin{array}{ll}\lambda & 0 \\ 0 & \mu\end{array}\right)\left(\begin{array}{cc}-x & y \\ 1 & 1\end{array}\right)^{-1},$$ where
$\lambda = \frac{1+2ab+c^2-\Delta}{2}$, $\mu = \frac{1+2ab+c^2+\Delta}{2}$, $x = \frac{\Delta+c^2-1}{2(a+b^2)}$, $y = \frac{\Delta-c^2+1}{2(a+b^2)}$. Because $G_2$ works, $[y,1]$ on LHS is  available. Use this $[y,1]$
in the non-linearity gadget in Figure \ref{non_li}, we get the unary signature $\left[y^{2}+y b, y a+c\right]$ on the RHS. By Lemma \ref{3.2}, we can interpolate any unary signature, in particular $\Delta_{0}$, $\Delta_{1}$ and $\Delta_{2}$ on RHS
and apply Lemma~\ref{after-unary}, unless $\left[y^{2}+y b, y a+c\right]$ is proportional to a row eigenvector of $G_3$, namely $[1,-y]$ and $[1, x]$. Thus the exceptions are $ya+c = x(y^{2}+y b)$ and $ya+c = -y(y^2+yb)$. Notice that now $xy = \frac{a^2+bc}{a+b^2}$. The first equation implies $c = ab$ or $a + b^2= 0$ or $a^3-b^3c+ab(-1+c^2) = 0$. The second equation implies $a+b^2=0$ or $f_1 =0$ where $f_1= a^3+4a^6+3a^5b^2+a^3b^3-c-4a^3c+6a^4bc-6a^2b^2c-b^3c-3a^2b^5c-3a^3c^2-3abc^2-4b^3c2-a^3b^3c^2-6ab^4c^2-4b^6c^2+3c^3+4a^3c^3+6a^2b^2c^3+3b^3c^3+a^3c^4+3abc^4+4b^3c^4-3c^5-b^3c^5+c^7$. So there are four exceptional cases, \begin{equation} \label {g3works} \begin{cases}
    c = ab \\
    a + b^2= 0\\
    a^3-b^3c+ab(-1+c^2) = 0\\
    f_1 = 0\\ \end{cases}
\end{equation}

 For each of them, together with (\emph{con1}) and (\emph{con2}), we get 3 equations and can solve them using
Mathematica\texttrademark{}.
For rational $a,b,c$, when $ab\neq0$, there are only two possible results --- $[1,-1,1,-1]$ and $[1,-\frac{1}{3}, -\frac{1}{3}, 1]$. The first one violates $a+b^2 \ne 0$, 
and the second is a matchgate and thus in $\operatorname{FP}$. For all other cases when $G_2$ works, we have
the pinnng signatures $\Delta_{0}$, $\Delta_{1}$  and $\Delta_{2}$
on the RHS and then the lemma is proved by Lemma \ref{after-unary}.

Now suppose $G_2$ does not work. Then by Proposition \ref{c3}, we get at least one more condition, either one in (\ref{c3eq}) or $(a+b^2)(a^2+bc) = (1+ab)(ab+c^2) $ which indicates that $G_2$ is degenerate. For each of the 6 conditions, together with (\emph{con1}) and (\emph{con2}), we can solve them using
Mathematica\texttrademark{} for rational $a,b,c$. The only 
solution is $[1,-1,1,-1]$ which  violates $a+b^2 \ne 0$.
The proof of the  lemma is complete.
\end{proof}

\begin{lemma} \label{1-2 X=1}
Let $a,b, c \in \mathbb{Q}$, $ab \neq 0$,  and
satisfy\\
\emph{(con1)} $a^3 - b^3 - ab(1-c) =0$ and\\
\emph{(con2)}  $(a^4b+ab^4)^2 = (a^5+b^4)(b^5+a^4c)$.\\
Then $\PlHolant ( \, [1,a,b,c] \, | =_3)$ is
 \numP-hard unless it is $[1,a,a^2, a^3]$ which is degenerate, or a matchgate, in both cases the problem is in FP. 
\end{lemma}
\begin{proof}
Eliminating $c$ from  (\emph{con1}) and (\emph{con2})
we get 
 $a^{11}-a^9b+a^6b^4+a^5b^6-a^4b^5-a^3b^7+a^2b^9-b^{10}=0$, which, quite miraculously, can be factored as $(a^2-b) (a^9 + a^4 b^4 + a^3b^6 + b^9)=0$. If $b=a^2$, then with (\emph{con1}), we get $c=a^3$ thus the signature becomes $[1,a,a^2,a^3]$ which is degenerate. We  assume 
 $a^9 + a^4b^4 + a^3b^6+b^9=0$, then the rest  of the
 proof is essentially the same as Lemma \ref{1-1 X=1}.
\end{proof}

\begin{theorem} \label{g1works1}
For $a,b, c \in \mathbb{Q}$, $ab \neq 0$, 
 if $G_1$ works, then $\PlHolant ( [1,a,b,c]  | =_3)$ is
 \numP-hard unless it is in the tractable cases of Theorem~\ref{thm:main} and thus in FP.
\end{theorem}

\begin{proof}
If $[1,a,b,c]$ has the form $[1,a,a,1]$ or $[1,a,-1-2a,2+3a]$
then the problem is in FP.  We now assume the signature is not of these two forms.
By Lemma \ref{y11x}, when $G_1$ works, we can interpolate $[y,1]$ on LHS and also  $[1,x]$ on RHS.

Let us write down the Jordan Normal Form again: $$G_1 = \left(\begin{array}{ll}1 & b \\ a & c\end{array}\right)=\left(\begin{array}{cc}-x & y \\ 1 & 1\end{array}\right)\left(\begin{array}{ll}\lambda & 0 \\ 0 & \mu\end{array}\right)\left(\begin{array}{cc}-x & y \\ 1 & 1\end{array}\right)^{-1},$$  $\lambda=\frac{-\Delta+(1+c)}{2}$, $\mu=\frac{\Delta+(1+c)}{2}$, $x=\frac{\Delta-(1-c)}{2 a}$,  $y=\frac{\Delta+(1-c)}{2 a}$, $\Delta=\sqrt{(1-c)^2 + 4ab}$.

Using $[y,1]$ and the gadget in Figure \ref{non_li}, we get  $\left[y^{2}+y b, y a+c\right]$ on the RHS.
We can interpolate  $\Delta_{0}$, $\Delta_{1}$ and $\Delta_{2}$ on RHS unless $\left[y^{2}+y b, y a+c\right]$ is proportional to a row eigenvector of $G_1$, namely $[1,-y]$ or $[1, x]$, according to Lemma \ref{3.2}. Thus the exceptions are $y a+c = (y^{2}+y b)x$ or $y a+c = -y(y^{2}+y b)$. The first equation implies $a^{3}-b^{3}-a b(1-c)=0$ or $c=ab$. The second equation implies $c=ab$ or $c=1+a-b$. 

By assumption  $G_1$ works, so $c\ne ab$. Thus, we consider two exceptional cases.

\noindent 
\textbf{Case 1}: $a^{3}-b^{3}-a b(1-c)=0$

In this case, we have $1-c = \frac{a^3-b^3}{ab}$ and thus $\Delta = \sqrt{(1-c)^2+4ab} = |\frac{a^3+b^3}{ab}|$.  One condition ($4ab + c^2 - 2c+1 = 0$) in (\ref{c1eq}) is the same as $\Delta=0$.  Since $G_1$ works, we have $\Delta \neq 0$ and thus $a^3+b^3\neq 0$, which is equivalent to $a+b\ne 0$ when $a,b\in \mathbb{Q}$.

\begin{description}
    \item{Subcase 1: $ \frac{a^3+b^3}{ab} > 0$.}
%
 We have  $[1,x] = [1,\frac{\Delta - \left(1-c\right)}{2a}] = [1, \frac{b^2}{a^2}]$ on RHS. Connect $[1,x]$ to $[1,a,b,c]$ on LHS, we get the binary signature $[1+\frac{b^2}{a}, a+\frac{b^3}{a^2}, b+\frac{b^2c}{a^2}]$ on LHS. Note that $a+\frac{b^3}{a^2} \neq 0$ when $a+b \neq 0$.
%
It is \numP-hard (and thus the problem $[1,a,b,c]$ is \numP-hard) unless one of the tractable conditions in Corollary~\ref{2-3} holds. It turns out that the only possibilities are either (1) 
first case in Corollary \ref{2-3}, i.e. $(1+\frac{b^2}{a})(b+\frac{b^2c}{a^2}) = (a+\frac{b^3}{a^2})^2$, which becomes $\left(a^2-b \right)\left(a^3+ab+2b^3 \right) = 0$ after substituting $c = \frac{ab-a^3+b^3}{ab}$,
or (2) the second and third case in Corollary \ref{2-3}, which implies $(a = -b^2) \wedge (c= - b^3)$ and thus the problem $[1,a,b,c]$ becomes the problem $[1,-b^2,b,-b^3]$,
or (3) the fourth case in Corollary~\ref{2-3},
which implies $(a=b) \wedge (c=1)$ or $(a = -b^2) \wedge (c= - b^3)$, where in the former case $[1,a,b,c]$ is of the matchgate form $[1,a,a,1]$ and thus in FP. 

We now deal with the case (1).
When $a^2-b = 0$, together with $a^3-b^3=ab\left(1-c\right)$, we have $c = a^3$, and thus $[1,a,b,c]$ is degenerate. When $a^3+ab+2b^3 = 0$, together with $a^3-b^3=ab\left(1-c\right)$, by Lemma \ref{1-1 X=1}, $[1,a,b,c]$ is \numP-hard (with $a+b\ne0$ ruling out the exception). 

We now deal with the case $[1,-b^2,b,-b^3]$. If $b=0, \pm 1$, it is degenerate or affine. Now assume $b\neq 0, \pm 1$. $G_1 = \smm{1}{b}{-b^2}{b^3} = \smmv{1}{-b^2}\cdot\smmh{1}{b}$. Then we can get $[1,-b^2]$ on the  LHS.
Note that connecting three copies of $[1,b]$ with $[1,-b^2,b,-b^3]$ on LHS produces a global factor 
$1-b^6 \ne 0$.  Connect $[1,-b^2]$ twice to $[1,0,0,1]$ on RHS, and we get $[1,b^4]$ on RHS. Connect $[1,b^4]$ back to $[1,-b^2, b, -b^3]$ on LHS, and we get a binary signature $g=[1-b^6, -b^2 + b^5, b-b^7]$, which by Corollary \ref{2-3} is \numP-hard unless $b = \pm1$ which has been discussed, and therefore $[1,-b^2, b, -b^3]$ is also \numP-hard.

 \item{Subcase 2: $\frac{a^3+b^3}{ab} < 0$}. We have $[y,1] = [-\frac{b^2}{a^2},1]$ on LHS. Connecting two copies of $[y,1]$ to $(=_3)$ we get $[y^2, 1] = [\frac{b^4}{a^4}, 1]$ on RHS.
 Connecting it back to LHS, we  get a binary signature $[\frac{b^4}{a^4}+a,\frac{b^4}{a^3}+b,\frac{b^5}{a^4}+c]$ on LHS. 
It is \numP-hard  unless one of the tractable conditions in Corollary~\ref{2-3} holds. It turns out that the only possibilities are either
 $(a^4b+ab^4)^2 = (a^5+b^4)(b^5+a^4c)$, or $(a=b) \wedge (c=1)$ which is the matchgate case and thus in FP. We now deal with the former case. Together with $a^3-b^3=ab\left(1-c\right)$, by Lemma \ref{1-2 X=1}, $[1,a,b,c]$ is \numP-hard unless it is degenerate.
\end{description}

\noindent 
\textbf{Case 2}: $1+a-b-c=0$

In this case, $\Delta=|a+b|$, and since
$G_1$ works, one condition is $4ab+c^2-2c+1=0$ in (\ref{c1eq})
which says $\Delta \ne 0$,  and thus $a+b \neq 0$.

If $a+b > 0$, then 
$\Delta=a+b$, and $x = \frac{a+b-(1-c)}{2a} = 1$. Then we can interpolate $[1,x] = [1,1]$ on RHS (as $y = \frac{a+b+(1-c)}{2a} = \frac{b}{a}\ne -1$). Else,  $a+b < 0$,  then 
$\Delta=-(a+b)$, and $y = \frac{-a-b + (1-c)}{2a} = \frac{-a-b + (b-a)}{2a}= -1$. We can  get $[1,1]$ on RHS by connecting two copies of $[y,1]=[-1,1]$ to $[1,0,0,1]$. Then connecting $[1,1]$ to $[1,a,b,c]$ on LHS we get a binary signature $[a+1,a+b, a+c]$ on LHS. Again we can apply
 Corollary \ref{2-3} to it, and conclude that it is
 \numP-hard. It turns out that  
 the only feasible cases of tractability leads to 
  $[1,a,-1-2a,2+3a]$, and $(c=1)\wedge(a=b)$, in both cases the problem is in FP.
This proves the \numP-hardness of
$\Holant ( \, [1,a,b,c] \, | =_3)$.
%
%
%
%
\end{proof}

\subsection{Dichotomy for $[1,a,b,0]$}   \label{large3}

\begin{theorem}   \label{1ab0}
The problem  $[1,a,b,0]$ 
for $a,b \in \mathbb{Q}$ is
 \numP-hard
 unless it is in the tractable cases of Theorem~\ref{thm:main} and thus in FP.
\end{theorem}

\begin{proof}
If $ab\neq 0$ and $G_1$ works, then this is proved in
Theorem~\ref{g1works1}.
 If $a=b=0$, it is degenerate and in FP. We divide the rest into three cases: \begin{enumerate}
    \item $ab\neq 0$ and $G_1$ does not work;
    \item $f = [1,a,0,0]$ with $a\neq 0$;
    \item $f= [1,0,b,0]$ with $b\neq 0$.
\end{enumerate}

\noindent 
$\bullet$ Case 1:  $ab\neq 0$ in $f = [1,a,b,0]$ and $G_1$ does not work. Since $c=0 \ne ab$, this implies that at least one equation in (\ref{c1eq})
holds. After a simple derivation, we  have the following family of signatures to consider:   $[1,a,-\frac{1}{ka}, 0]$, for $k = 1, 2, 3, 4$. 

We  use $G_3$ to produce another symmetric ternary signature in each case. If the new signature is \numP-hard, then so is the given signature. We will describe the case $[1,a,-\frac{1}{a}, 0]$ in more detail; the other three types ($ k=2, 3, 4$) are similar.

For $k=1$, the gadget $G_3$ produces $g= [3 a^3+4, a^4-a-\frac{2}{a^2}, -a^2+\frac{1}{a}+\frac{1}{a^4}, a^3+3]$. 
For $a= -1$, this is $[1,0,-1,2]$, which has the form $[1,a',-1-2a', 2+3a']$ and is in FP. Below we assume $a \not = -1$. Then all entries of $g$ are nonzero.

We claim that the gadget $G_1$ works using $g$. Since $a\in \mathbb{Q}$, it can be checked that $g$ is non-degenerate since $(a^4-a-\frac{2}{a^2})( -a^2+\frac{1}{a}+\frac{1}{a^4}) = (3a^3 + 4)(a^3 + 3)$ has no solution, and that no equation in (\ref{c1eq}) has a solution applied to $g$. Hence, $G_1$ works using $g$ and we may apply Theorem \ref{g1works1} to $g$. Using the fact that $a\in\mathbb{Q}$, one can show that $g$ cannot be a Gen-Eq because it has no zero entry, nor can it be affine or degenerate. Also, it can be checked that there is no solution for $a$ if we were to impose the condition that $g$ is a matchgate, i.e. $(3 a^3+4 = a^3+3) \wedge (a^4-a-\frac{2}{a^2} = -a^2+\frac{1}{a}+\frac{1}{a^4})$, and also $ a = -1$ is the only solution for $g$ being in the form $[1,a',-1-2a',2+3a']$. Thus $[1,a,-\frac{1}{a}, 0]$ is \numP-hard.

\noindent 
$\bullet$ Case 2: $f= [1,a,0,0]$ with $a\neq 0$. 
The gadget $G_3$ produces $g'=[3a^3+1, a^4+a, a^2, a^3]$. Since $a\in \mathbb{Q}$, $3a^3+1\neq0$. 
If $a=-1$, $g'=[-2, 0,1,-1]$ and 
it suffices to consider $[1,-1,0,2]$ (which is $-1$ times the reversal $[-1,1,0,-2]$ 
obtained by swapping and roles of 0 and 1), in which case $G_2$ works where the matrix $G_2=\left[\begin{smallmatrix}1 & 1 \\ -1 & 4\end{smallmatrix}\right]$. We can interpolate $[1,x] = [1, -\frac{3 + \sqrt{5}}{2}]$ on RHS. Connect it back to $[1,-1,0,2]$ and get a binary signature $[\frac{5+\sqrt{5}}{2}, -1, -(3+\sqrt{5})]$ on LHS, which, by Corollary \ref{2-3}, is \numP-hard. Thus, $[1,-1, 0,2]$ is \numP-hard and so is $[1,a,0,0]$.

Else, $a\neq -1$.
We claim that the gadget $G_1$ works using $g'$.
The signature $g'$ is non-degenerate 
since $a\in \mathbb{Q}$  is nonzero
and thus  $(a^4+a)a^2 \not =(3a^3+1)a^3$.
Also 
 no equation in (\ref{c1eq}) has a solution applied to $g'$. Hence, $G_1$ works using $g'$ and we may apply Theorem \ref{g1works1} to $g'$. Using the fact that $a\in\mathbb{Q}$, one can show that $g'$ cannot be a Gen-Eq because it has no zero entry, nor can it be affine or degenerate. Also, there is no solution for $g'$ being a matchgate or in the form $[1,a',-1-2a',2+3a']$. Thus $[1,a,0, 0]$ is \numP-hard.


\noindent 
$\bullet$ Case 3: 
 $f = [1,0,b,0]$ with $b\neq 0$. The gadget $G_1$ produces a binary straddled signature $G_1=\left[\begin{smallmatrix}1 & b \\ 0 & 0\end{smallmatrix}\right] = \left[\begin{smallmatrix} 1 \\ 0 \end{smallmatrix}\right]\cdot \left[\begin{smallmatrix}1 & b \end{smallmatrix}\right]$ 
which decomposes into 
a unary signature $[1,b]$ on RHS and 
a unary signature $[1, 0]$ on LHS.
This gives us a reduction
$\PlHolant ( [1,b^2,b] | ( =_3)) \le_T
\PlHolant ( \, f \, | =_3)$ by the P3EM argument. The problem 
$\PlHolant ( [1,b^2, b]  | =_3)$ is \numP-hard except $b =\pm 1$,
by Corollary \ref{2-3}, which implies that $\PlHolant ( \, f \, | =_3)$ is also \numP-hard
when $b \ne \pm 1$. If $b = \pm 1$,
then $f$ is affine, and $\PlHolant ( \, f \, | =_3)$ is in FP.
%
%
%
\end{proof}

\subsection{Dichotomy for $[1,a,0,c]$}  \label{large4}
\begin{theorem} \label{1a0c}
The problem $[1,a,0,c]$ with $a,c\in \mathbb{Q}$ is \numP-hard unless $a=0$, in which case it is Gen-Eq and thus in FP.
\end{theorem}
\begin{proof}
When $a=0$, it is Gen-Eq and so is in FP. When $a\neq 0$, if $c=0$, it is \numP-hard by Theorem \ref{1ab0}. In the following we discuss $[1,a,0,c]$ with $ac\neq 0$.

 If $c=\pm1$, the signature is $[1,a,0,1]$ or $[1,a,0,-1]$. We use $G_2$ to produce a ternary signature $g=[3a^3+1, a^4+a, a^2, a^3 + 1]$ (both mapped to the same signature, surprisingly). If $a=-1$, it is $[1,0,-\frac{1}{2}, 0]$ after normalization, which by Theorem \ref{1ab0} is \numP-hard and so is the given signature $[1,-1,0,1]$. If $a\neq -1$, then $g$ has no zero entry. We then claim that the gadget $G_1$ works using $g$. It can be checked that $g$ is non-degenerate since $(a^4+a)a^2=(3a^3+1)(a^3+1)$ has no solution, and that no equation in (\ref{c1eq}) has a solution applied to $g$. Hence, $G_1$ works using $g$ and we may apply Theorem \ref{g1works1} to $g$. Using the fact that $a\in\mathbb{Q}$, one can show that $g$ cannot be a Gen-Eq because it has no zero entry, nor can it be affine or degenerate.  Also, it can be checked that the only solution for $a$ if we were to impose
 the condition that $g$ is a matchgate or in the form $[1,a',-1-2a',2+3a']$ for some $a'$ is $a=0$. Thus $[1,a,0,\pm 1]$ are both \numP-hard.


Now assume $c \ne 0, \pm 1$. We claim that the gadget $G_1$ works. It can be checked that for the non-degenerate matrix $G_1=\smm{1}{0}{a}{c}$, $\Delta = |1-c|$, $\rlm \in \{c, \frac{1}{c}\}$  is not a root of unity. 
Next we claim that we can obtain $[1,0]$ on RHS.
If $c<1$ by Lemma~\ref{y11x}
we can interpolate $[1,x]=[1,0]$  on RHS 
with two exceptions to which we already give a dichotomy (see
the Remark after Definition \ref{d1}). 
If $c>1$,  we can interpolate $[y,1] = [0,1]$ on LHS and so the gadget in Figure \ref{non_li} produces  $[0,c]$ on RHS, which is not
proportional to the 
 row eigenvectors $[1,-y] = [1,0]$ and 
 $[1,x] = [1,\frac{c-1}{a}]$ of  $G_1$. 
 By Lemma \ref{3.2}, we can interpolate any unary gadget on RHS, including $[1,0]$. Thus we can always get $[1,0]$ on RHS. Connect $[1,0]$ to $[1,a,0,c]$  and we will get a binary signature $[1,a,0]$ on LHS, which is \numP-hard by Corollary \ref{2-3}. Therefore $[1,a,0,c]$ is \numP-hard when $c\ne 0, \pm 1$. 
\end{proof}

\subsection{Dichotomy for $[1,a,b,c]$ when $abc\ne0$}

We need three lemmas to handle some special cases. Lemma~\ref{p1} is a part of Theorem~\ref{g1works1} 
(one verifies that $G_1$ works, in fact for $[1,-b^2, b, -b^3]$ the condition  (\ref{c1eq}) amounts to $b=1$, and the $b=0$ case is degenerate thus trivially in FP). For convenience,
we state it explicitly here.

\begin{lemma} \label{p1}
The problem $[1,-b^2, b, -b^3]$ with $b\in \mathbb{Q}$ is \numP-hard unless $b=0, \pm 1$, which is in FP.
\end{lemma}

The next lemma is not part of Theorem~\ref{g1works1}  since the condition  (\ref{c1eq}) fails for 
$[1, a, -\frac{1}{a}, -1]$.

\begin{lemma} \label{p2}
The problem $[1, a, -\frac{1}{a}, -1]$ with $a\in \mathbb{Q}$, $a\neq 0$ is \numP-hard unless $a=\pm 1$, in which case it is in FP.
\end{lemma}
\begin{proof}
If $a=\pm 1$, $[1,1,-1-1]$ is   affine and $[1,-1,1,-1]$ is degenerate, both of which are in FP. Now we assume $a\neq \pm 1$ (so the matrix $\smm{a^{-2}}{a}{1}{1}$ is invertible). We use the ternary gadget $G_3$ to get the signature $g = [3a^3+\frac{1}{a^3}+4, a^4-\frac{1}{a^2},-a^2+\frac{1}{a^4},a^3+\frac{3}{a^3}+4]$ on LHS. A direct computation (using Mathematica) shows that $G_1$ always works for $g$ unless $ a = -1$. Therefore, by Theorem~\ref{g1works1} we have $g$ is \numP-hard and then $[1,a,-\frac{1}{a},-1]$ is \numP-hard unless $g$ is in the tractable cases in Theorem~\ref{thm:main}. The only solution for $g$ being in the tractable cases in Theorem~\ref{thm:main} is $a = \pm1$. This completes our proof.
\end{proof}

\begin{lemma} \label{c=ab}
The problem $[1,a,b,ab]$ with $a,b\in \mathbb{Q}$ and $a,b\ne 0$ is \numP-hard unless it is degenerate or  affine, which is in FP.
\end{lemma}
\begin{proof}
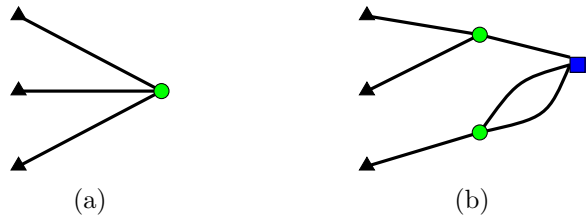
\begin{figure}[h!] 
    \centering
    \begin{subfigure}[b]{0.3\textwidth}
        \centering
        \begin{tikzpicture}  
            \draw[very thick] (0.3,5.05) -- (2.2,4.05);
            \draw[very thick] (0.3,4.05) -- (2.2,4.05);
            \draw[very thick] (0.3,3.05) -- (2.2,4.05);
            \filldraw[fill=black] (0.3,5.17)--(0.2,5)--(0.4,5)--cycle;
            \filldraw[fill=black] (0.3,4.17)--(0.2,4)--(0.4,4)--cycle;
            \filldraw[fill=black] (0.3,3.17)--(0.2,3)--(0.4,3)--cycle;
            \filldraw[fill=green] (2.2, 4.05)circle(0.1);

        \end{tikzpicture}
                 \captionsetup{justification=centering}
  \caption{}
        \label{i1}
    \end{subfigure}
    \begin{subfigure}[b]{0.3\textwidth}
        \centering
        \begin{tikzpicture}  
           \draw[very thick] (0.3,5.05)--(1.8,4.8);
        \draw[very thick] (0.3,4.05)--(1.8,4.8);
        \draw[very thick] (0.3,3.05)--(1.8,3.5);
        \draw[very thick] (1.8,4.8) -- (3,4.5);
        \draw[very thick] (1.8,3.5).. controls (2.3, 4.2) ..(3,4.4);
        \draw[very thick] (1.8,3.5) .. controls (2.7, 3.7)..(3,4.4);
        \filldraw[fill=black] (0.3,5.17)--(0.2,5)--(0.4,5)--cycle;
        \filldraw[fill=black] (0.3,4.17)--(0.2,4)--(0.4,4)--cycle;
        \filldraw[fill=black] (0.3,3.17)--(0.2,3)--(0.4,3)--cycle;
        \filldraw[fill=green] (1.8,4.8) circle (0.1);
        \filldraw[fill=green] (1.8,3.5) circle (0.1);
        \filldraw[fill=blue] (3,4.3)rectangle(3.2,4.5);

        \end{tikzpicture}
                \captionsetup{justification=centering}
  \caption{}
        \label{i2}
    \end{subfigure}
    \captionsetup{justification=centering}
    \caption{Two gadgets where each triangle represents the unary gadget $[1,a]$}
    \label{2g41a}
\end{figure}

If $a=-1$ and $b= \pm 1$
then  problem is in FP. Indeed,  $[1,-1,1,-1]$ is degenerate, and  $[1,-1,-1,1]$ can be transformed to matchgate;
both problems are  in FP. We therefore now assume it is not the case that both $a=-1$ and $b= \pm 1$.

We first use $G_3$ to construct a ternary signature $h = [-b^4+3b^2-2, b^4-2b^3-b^2+2b, 2b^3-b^2-2b+1, -2b^4+3b^2-1]$ on LHS, which can be normalized to $h = [2-b^2,b^2-2b,2b-1,1-2b^2]$ after dividing $(b+1)(b-1)$. 
Using gadget $G_1$, we have a degenerate matrix   $G_1 = \smm{1}{b}{a}{ab} =  \smmv{1}{a}\cdot \smmh{1}{b}$. We  get $[1,b]$ on RHS if $[1,a]$ can appropriately form some nonzero global factor.

Figure \ref{2g41a} indicates two different ways of ``absorbing'' $[1,a]$ on LHS. Importantly, we place 
$g$ instead of $[1,a,b,ab]$ on the square vertex.  Figure~\ref{i1} provides a factor $1+a^3$ which is nonzero if $a \neq -1$. When $a = -1$,  Figure~\ref{i2} provides a factor $2b^2-4b+2 = 2(b-1)^2$ which is nonzero unless $b=1$. Therefore, we can interpolate $[1,b]$ on RHS. Connect $[1,b]$ back to $[1,a,b,ab]$ on LHS and we get the binary signature $g=[1+ab, a+b^2, b+ab^2]$. If $a + b^2 = 0$, the given signature is $[1, -b^2, b, -b^3]$ which, according to Lemma \ref{p1}, is \numP-hard (the exceptions in Lemma \ref{p1} do not apply as $b \ne 0$ and if $b = \pm 1$ then $a =-b^2 = -1$ which is excluded.)
Now we assume $a+b^2 \neq 0$. Normalize $g$ by dividing $a+b^2$, we have the binary signature $[\frac{1+ab}{a+b^2},1,\frac{b+ab^2)}{a+b^2}]$ on LHS. 
Applying Corollary \ref{2-3} to $g$, it is \numP-hard (and so is the given signature $[1,a,b,ab]$) unless \begin{enumerate}
    \item $(1+ab)(b+ab^2) = (a+b^2)^2$. This implies $\left(a^2-b \right)\left(b^3-1 \right)=0$. If $a^2-b=0$, the given signature is $[1,a,a^2,a^3]$ and  is degenerate. If $b^3-1=0$, since $b\in\mathbb{Q}$, we have $b=1$ and  $a
    \ne  -b^2 = -1$, and the given signature is $[1,a,1,a]$. This is resolved by Lemma~\ref{lemma:[1,a,1,a]}.
    \item $\frac{1+ab}{a+b^2} = 1$ and $\frac{b+ab^2}{a+b^2}=-1$. 
    Dividing the two expressions gives $b=-1$. This implies $a=0$ 
    and therefore violates our assumption that $a \neq 0$.
    \item $\frac{1+ab}{a+b^2} = -1$ and $\frac{b+ab^2}{a+b^2}= 1$. No $(a,b)$ pair solves these two equations, under the assumption $a +b^2 \ne 0$.
    %
    \item $\frac{1+ab}{a+b^2} = \frac{b+ab^2}{a+b^2}$. We have either $b=1$ or $ab=-1$. If $b = 1$, then $[1,a,b,ab]$ becomes $[1,a,1,a]$ which is \numP-hard by Lemma~\ref{lemma:[1,a,1,a]} unless $a = 0$ (this violates our assumption) or $a= \pm 1$, in which cases the signature is affine and therefore the problem is in FP. If $ab=-1$, then  $[1,a,b,ab]$ becomes $[1,a,-\frac{1}{a},-1]$ which is \numP-hard by Lemma~\ref{p2} unless $a=\pm1$, in which cases the signature is affine and therefore the problem is in FP.
    \end{enumerate}
Note that since $a,b\ne 0$, $[1,a,b,ab]$ cannot be Gen-Eq. The lemma is  proved.
\end{proof}

Now we prove
%
\begin{theorem} \label{large5}
The problem $[1,a,b,c]$ with $a,b,c \in \mathbb{Q}$, $abc\ne0$, is \numP-hard unless it is in the tractable cases in Theorem~\ref{thm:main}.
\end{theorem}
\begin{proof}
 By Proposition~\ref{g1-rou}, Theorem \ref{g1works1} and Lemma \ref{c=ab} it suffices to consider the case when the ratio of two eigenvalues in $G_1=\smm{1}{b}{a}{c}$ is a root of unity and $c \neq ab$. If the ratio of eigenvalues of $G_1$ is a root of unity, we know at least one condition in (\ref{c1eq}) holds. For  convenience, we list the conditions in (\ref{c1eq}) here and label them as $R_i$ where $i=1,2,3,4,5$:
\begin{equation} \label{j1}
    R = \bigvee_{i=1}^{5} R_i,
    \ \ \text{where} \begin{cases}
     R_1: c=-1\\
     R_2: ab + c^2+c+1=0\\
     R_3: 2ab+c^2+1=0\\
     R_4: 3ab+c^2-c+1=0\\
     R_5: 4ab+c^2-2c+1=0\\
    \end{cases}
\end{equation}


We apply $G_{3}$ on $[1,a,b,c]$, i.e.\ placing squares to be $[1,a,b,c]$ and circles to be $=_3$, to produce a ternary signature $[w,x,y,z]=[1+3a^3+3a^2b^2+b^3c, a+a^4+2a^2b+a^2bc+2ab^3+b^2c^2, a^2+ab^2+2a^3b+b^4+2ab^2c+bc^3, a^3+3a^2b^2+3b^3c+c^4]$. If $w \neq 0$ and $G_1$ works on $[w,x,y,z]$, by Theorem \ref{g1works1} we have $[w,x,y,z]$ is \numP-hard and thus $[1,a,b,c]$ is \numP-hard unless at least one condition $S_i$ listed below holds, where $i=1,2,3,4,5,6,7$:
\begin{equation} \label{j2}
    S = \bigvee_{i=1}^{7} S_i
    \ \ \text{where} \begin{cases}
    S_1 : x^2=wy \wedge y^2=xz\ (\text{degenerate form})\\
    S_2 : x=0 \wedge y=0\ (\text{Gen-Eq form})\\
    S_3 : w=y \wedge x=0 \wedge z=0\ (\text{affine form $[1,0,1,0]$}) \\
    S_4 : w+y=0 \wedge x = 0 \wedge z=0\ (\text{affine form $[1,0,-1,0]$})\\
    S_5: w = z \wedge x = y\ (\text{matchgate-realizable form $[a',b',b',a']$}) \\
    S_6: w = -z \wedge x = -y\ (\text{matchgate-realizable form $[a',b',-b',-a']$}) \\
    S_7: -w-2x = y \wedge 2w + 3x = z\ (\text{form $[3a'+b',-a'-b',-a'+b',3a'-b']$})
    \end{cases}
\end{equation} 
\noindent Note that the affine forms $[1,1,-1,-1]$ and $[1,-1,-1,1]$ are special forms of  $[a,b,b,a]$ and  $[a,b,-b,-a]$. Solve the equation system $R \wedge S$ for variables $a,b,c \in \mathbb{Q}$, we have the following solutions:
\begin{itemize}
    \item $a=c=-1, b=1$; the problem $[1,-1,1,-1]$ is in FP since it is degenerate;
    \item $a=1, b=c=-1$; the  problem $[1,1,-1,-1]$ is in FP since it is affine;
    \item $a=c=1, b=-1$; the problem $[1,1,-1,1]$ is \numP-hard (use the gadget $G_3$ to produce $[1,1,-1,3]$ after flipping 0's and 1's, then use it again to produce $[1,1,-5,19]$ which is \numP-hard by Theorem~\ref{g1works1}. Note that we need to apply
    $G_3$ twice in order that the condition that
    $G_1$ works
    in 
    Theorem~\ref{g1works1} is satisfied for the newly created ternary signature);
    \item $a = -b, c= -1$; the problem $[1,a,-a,-1]$ is matchgate-transformable and thus in FP.
\end{itemize}

Continuing the discussion for the ternary signature $[w,x,y,z]$, it remains to consider the case when $w=0$ or $G_1$ does not work on $[w,x,y,z]$. For $ w \ne 0$ we  normalize $[w,x,y,z]$ to be $[1,\frac{x}{w}, \frac{y}{w}, \frac{z}{w}]$ and substituting $\frac{x}{w}, \frac{y}{w}, \frac{z}{w}$ into $a,b,c$ respectively in (\ref{c1eq}), we get at least one condition $T_i$ listed below, where $i=1,2,3,4,5,6$:
\begin{equation} \label{j3}
    T = \bigvee_{i=1}^{6} T_i,
    \ \ \text{where} \begin{cases}
    T_1 : zw+w^2=0\\
    T_2 : xy+z^2+zw+ w^2=0\\
    T_3 : 2xy+z^2+w^2=0\\
    T_4 : 3xy+z^2-zw+w^2=0\\
    T_5 : 4xy+z^2-2zw+w^2=0\\
    T_6 : xy=wz\\
    \end{cases}
\end{equation}

Note that $T_1$ incorporates the case when $w=0$. So we have the condition $R \wedge T$.
We now apply  $G_{3}$ once
again using $[w,x,y,z]$ to produce another new ternary signature $[w_2,x_2,y_2,z_2]$ where $w_2 = w^4 + 3wx^3 + 3x^2y^2+y^3z$, $x_2=w^3x+2wx^2y+x^4+2xy^3+x^2yz+y^2z^2$, $y_2= w^2x^2+wxy^2+2x^3y+y^4+2xy^2z+yz^3$, $z_2=wx^3+3x^2y^2+3y^3z+z^4$. Similarly as the previous argument, if $w_2 \neq 0$ and $G_1$ works on $[w_2,x_2,y_2,z_2]$, we know $[w_2,x_2,y_2,z_2]$ is \numP-hard and thus $[1,a,b,c]$ is \numP-hard unless at least one condition $U_i$  listed below holds, where $i=1,2,3,4,5,6,7$:
\begin{equation} \label{j4}
    U = \bigvee_{i=1}^{7} U_i,
    \ \ \text{where} \begin{cases}
    U_1 : x_2^2=w_2y_2 \wedge y_2^2=x_2z_2\ (\text{degenerate form})\\
    U_2 : x_2=0 \wedge y_2=0\ (\text{Gen-Eq form})\\
    U_3 : w_2=y_2 \wedge x_2=0 \wedge z_2=0\ (\text{affine form $[1,0,1,0]$}) \\
    U_4 : w_2+y_2=0 \wedge x_2= 0 \wedge z_2=0\ (\text{affine form $[1,0,-1,0]$})\\
    S_5: w_2 = z_2 \wedge x_2 = y_2\ (\text{matchgate-realizable form $[a',b',b',a']$}) \\
    S_6: w_2 = -z_2 \wedge x_2 = -y_2\ (\text{matchgate-realizable form $[a',b',-b',-a']$}) \\
    S_7: -w_2-2x_2 = y_2 \wedge 2w_2 + 3x_2 = z_2\ (\text{form $[3a'+b',-a'-b',-a'+b',3a'-b']$})
    \end{cases}
\end{equation}

Solve the equation system $R \wedge T \wedge U$ for rational-valued variables $a,b,c$, we have the following solutions:
\begin{itemize}
    \item $a=c=-1, b=1$; the problem $[1,-1,1,-1]$ is in FP since it is degenerate;
    \item $a=1, b=c=-1$; the problem $[1,1,-1,-1]$ is in FP since it is affine;
    \item $a=-1, b=c=1$; the problem $[1,-1,1,1]$ is \numP-hard (use the gadget $G_4$ to produce $[1,1,-1,3]$, use it again to produce $[1,1,-5,19]$ which is \numP-hard by Theorem~\ref{g1works1});
    \item $a=c=1, b=-1$; the problem $[1,1,-1,1]$ is \numP-hard
    (this is the reversal of $[1,-1, 1, 1]$);
\end{itemize}

Otherwise, we know $w_2 = 0$ or $G_1$ does not work on $[w_2,x_2,y_2,z_2]$. Similarly, we know at least one condition $V_i$  listed below holds, where $i=1,2,3,4,5, 6$:
\begin{equation}\label{j5}
    V =    \bigvee_{i=1}^{6} V_i,
    \ \ \text{where} \begin{cases}
    V_1 : z_2w_2+w_2^2=0\\
    V_2 : x_2y_2+z_2^2+z_2w_2+ w_2^2=0\\
    V_3 : 2x_2y_2+z_2^2+w_2^2=0\\
    V_4 : 3x_2y_2+z_2^2-z_2w_2+w_2^2=0\\
    V_5 : 4x_2y_2+z_2^2-2z_2w_2+w_2^2=0\\
    V_6 : x_2y_2=w_2z_2\\
    \end{cases}
\end{equation}

Finally, solve the equation system $R\wedge T\wedge V$ for variables $a,b,c \in \mathbb{Q}$, we have the following solutions:
\begin{itemize}
    \item $a=-1, b=c=1$; the problem $[1,-1,1,1]$ is \numP-hard (see the case above for $R \wedge T \wedge U$);
    \item $a=c=1, b=-1$; the problem $[1,1,-1,1]$ is \numP-hard 
    (this is the reversal of $[1,-1,1,1]$);
    \item $a=-b, c=-1$; the problem $[1,a,-a,-1]$ is matchgate-transformable and thus in FP.
\end{itemize}

The proof of Theorem~\ref{large5} is now complete.
\end{proof}

\section{Dichotomy for $[0,a,b,0]$} 
We now finish the discussion for $[0,a,b,0]$ with the help of previous theorems on $[1,a,b,c]$. 
\begin{theorem} \label{0ab0}
The problem $[0,a,b,0]$ with $a,b \in \mathbb{Q}, ab\neq 0$ is \numP-hard unless $a=\pm b$.
\end{theorem}
\begin{proof}
 We apply the gadget $G_3$ on $[0,a,b,0]$ to produce the ternary signature $g=[3a^2b^2, a(a^3+2b^3), b(2a^3+b^3), 3a^2b^2]$. 
 
We can normalize $g$  to be the form $[1, a', b', c']$. Since $a, b, c \in \mathbb{Q}$, we have
 $a'b'c' \ne 0$. By Theorem \ref{large5}, we know $[1,a',b',c']$ is \numP-hard (and so is $[0,a,b,0]$) unless it is in the tractable cases in Theorem~\ref{thm:main}. However, the problem $[1,a',b',c']$ in FP implies $a = \pm b$. This finishes the proof.
 \end{proof}

If now $a=b=0$, then the Holant value is 0 and the problem is trivially in $\operatorname{FP}$. Suppose exactly one of $a$ and $b$ is 0. In this case, by normalizing and possibly flipping 0 and 1 in the input, it suffices to consider the ternary signature $[0,1,0,0]$. 

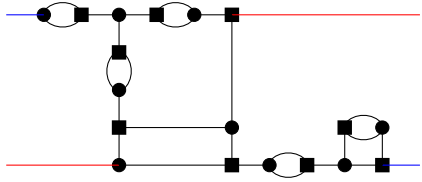
\begin{figure}[ht]
  \centering
  \begin{tikzpicture}[scale=0.5]
  \filldraw (0,0) circle (5pt);
  \draw[black, very thick, fill=black] (-0.15,0.85) rectangle (0.15,1.15);
  \filldraw (0,2) circle (5pt);
  \draw[black, very thick, fill=black] (-0.15,2.85) rectangle (0.15,3.15);
  \filldraw (0,4) circle (5pt);
  \draw[black, very thick, fill=black] (-1.15,3.85) rectangle (-0.85,4.15);
  \filldraw (-2,4) circle (5pt);
  \draw[black, very thick, fill=black] (2.85,-0.15) rectangle (3.15,0.15);
  \draw[black, very thick, fill=black] (0.85,3.85) rectangle (1.15,4.15);
  \filldraw (2,4) circle (5pt);
  \draw[black, very thick, fill=black] (2.85,3.85) rectangle (3.15,4.15);
  \filldraw (3,1) circle (5pt);
  \filldraw (4,0) circle (5pt);
  \draw[black, very thick, fill=black] (4.85,-0.15) rectangle (5.15,0.15);
  \filldraw (6,0) circle (5pt);
  \draw[black, very thick, fill=black] (6.85,-0.15) rectangle (7.15,0.15);
  \draw[black, very thick, fill=black] (5.85, 0.85) rectangle (6.15,1.15);
  \filldraw (7,1) circle (5pt);
  \draw[red] (-3,0) -- (0,0);
  \draw[red] (3,4) -- (8,4);
  \draw[blue] (-2,4) -- (-3,4);
  \draw[blue] (7,0) -- (8,0);
  \draw (0,0) -- (4,0);
  \draw (0,0) -- (0,2);
  \draw (0,3) -- (0,4);
  \draw (-1,4) -- (1,4);
  \draw (2,4) -- (3,4);
  \draw (0,1) -- (3,1);
  \draw (3,0) -- (3,4);
  \draw (5,0) -- (7,0);
  \draw (6,0) -- (6,1);
  \draw (7,0) -- (7,1);
  \draw (-0.1,2) .. controls (-0.4,2.25) and (-0.4,2.75) .. (-0.1,3);
  \draw (0.1,2) .. controls (0.4,2.25) and (0.4,2.75) .. (0.1,3);
  \draw (1,3.9) .. controls (1.25,3.6) and (1.75,3.6) .. (2,3.9);
  \draw (1,4.1) .. controls (1.25,4.4) and (1.75,4.4) .. (2,4.1);
  \draw (4,-0.1) .. controls (4.25,-0.4) and (4.75,-0.4) .. (5,-0.1);
  \draw (4,0.1) .. controls (4.25,0.4) and (4.75,0.4) .. (5,0.1);
  \draw (6,0.9) .. controls (6.25,0.6) and (6.75,0.6) .. (7,0.9);
  \draw (6,1.1) .. controls (6.25,1.4) and (6.75,1.4) .. (7,1.1);
  \draw (-1,4.1) .. controls (-1.25,4.4) and (-1.75,4.4) .. (-2,4.1);
  \draw (-1,3.9) .. controls (-1.25,3.6) and (-1.75,3.6) .. (-2,3.9);
  \end{tikzpicture}
  \captionsetup{justification=centering}
  \caption{The cross-over-pinned-0 gadget $\mathcal{P}$}
  \label{a nice gadget}
\end{figure}
\begin{theorem}\label{thm:[0,1,0,0]}
The problem $\plholant{[0,1,0,0]}{(=_3)}$ is \numP-complete.
\end{theorem}

\begin{proof}
In~\cite{DyerF86}, Dyer and Frieze 
proved the problem Planar-X3C  NP-complete: An input is a
collection $\cal S$ of 3-element subsets of a set $U$, where the 
bipartite incidence graph is planar, and we ask for
an exact cover of $U$ by some ${\cal S'} \subseteq {\cal S}$.
Their reduction in fact produces instances 
where every $x \in U$ appears in exactly two or three sets of $\cal S$.
One can further verify that their reduction is parsimonious.
Thus, their proof yields the \#P-completeness for
$\plholant{[0,1,0,0],[0,1,0]}{(=_3)}$.
We prove our theorem by a reduction
\begin{equation}\label{reduction-in-thm4}
\plholant{[0,1,0,0],[0,1,0]}{(=_3)}\leq_{T} \plholant{[0,1,0,0]}{(=_3)}.
\end{equation}

Note that a unary pin-0 signature $\Delta_0$
connected to
$[0,1,0,0]$  produces $[0,1,0]$. If we replace  each $[0,1,0]$ in  $\Omega$ by 
 $[0,1,0,0]$ connected with  $\Delta_0$, the Holant value is unchanged.  
 So  if 
we can produce 
 $\Delta_0$ on the RHS in  $\plholant{[0,1,0,0]}{(=_3)}$, then (\ref{reduction-in-thm4}) follows. 
 But, in any bipartite 3-regular
problem 
\emph{provably} no construction can produce
individual unary signatures.  
Next, 
 note that in any signature grid $\Omega$ of  $\plholant{[0,1,0,0],[0,1,0]}{(=_3)}$,
 the number of  appearances of $[0,1,0]$ is congruent to $0 \bmod 3$, by counting the total degrees of LHS and RHS. Then, our idea is to create  triples of $\Delta_0$ so that we can apply them, one triple 
 at a time.

There remains the difficulty of how to construct  triples of $\Delta_0$ on the RHS in the  setting
 $\plholant{[0,1,0,0]}{(=_3)}$, and more importantly, not only the construction must be
 planar but also we must be able to apply them in $\Omega$, three at a time,
 in a planar fashion.  Notice that the appearances of $[0,1,0]$ in  $\Omega$ 
 generally do not allow this planar grouping (and indeed the output instances
 in~\cite{DyerF86} do not have this property).
 
The following construction accomplishes all these requirements in one fell swoop!
The planar cross-over-pinned-0 gadget $\mathcal{P}$ is illustrated in Figure~\ref{a nice gadget}, where we place  $[0,1,0,0]$ at the squares and  $=_3$ at the circles. It has the following properties:
\begin{enumerate}
    \item externally the two left dangling edges are to be connected to LHS, and two right dangling edges are to be connected to RHS;
    \item the two blue dangling edges are pinned to be 0;
    \item the two red dangling edges can be assigned to either 0 or 1, but must be the same value, and either choice
    induces  a unique assignment for the internal edges.
\end{enumerate}
Note also that if we ``flip'' $\mathcal{P}$ along the ``axis'' of the two blue edges,
thereby exchange the two red edges, we have a reflected copy of $\mathcal{P}$,
call it  $\mathcal{P}'$,
where the North-East red edge connects externally to LHS, and the
South-West red edge connects externally to RHS, exactly the opposite of  $\mathcal{P}$.

Thus, the gadget  $\mathcal{P}$   allows ``passing over'' one crossing edge (the pair of
red edges will  take its place) 
while the two end blue edges are pinned to 0. 
We can link any $k \ge 1$ copies of   $\mathcal{P}$ or  $\mathcal{P}'$
by the blue edges to ``pass over''  $k$ crossings. Note that the linking of the 
two end blue edges respects the bipartite structure, and  $\mathcal{P}$ or  $\mathcal{P}'$
allow any individual bipartite orientation of the crossed edge.
We call this a linked  $\mathcal{P}$ gadget.

Let $n=3m$ be the number of $[0,1,0]$ in $\Omega$ for some integer $m\ge 0$.  We now add $m$ new vertices on RHS assigned the signature $=_3$. We then use three copies of the linked
 $\mathcal{P}$ gadgets to connect  this $=_3$ to three occurrences of $[0,1,0]$ in $\Omega$, while 
replacing  the signatures there by $[0,1,0,0]$.
(If  some passage from $=_3$ to  $[0,1,0]$ does not encounter any crossing edge, we will
artificially introduce two such crossings!)
This defines a signature grid $\Omega'$ in $\plholant{[0,1,0,0]}{(=_3)}$
with  $\operatorname{Holant}(\Omega') = \operatorname{Holant}(\Omega)$.
\end{proof}

\begin{remark}
One can easily construct a degenerate ternary signature $[1,0,0,0] = [1,0] \otimes [1,0] \otimes [1,0]$ on RHS by placing circles to be $[0,1,0,0]$ and squares to be $(=_3)$ on $G_2$. However, we cannot  apply Theorem~\ref{thm:P3EM} and conclude that 
$\plholant{[0,1,0,0],[0,1,0]}{(=_3)} \leq_{T} 
\plholant{[0,1,0,0]}{(=_3)}.$
See also Remark~\ref{remark8}.
\textbf{The use of the gadget $\mathcal{P}$ is essential.}
\end{remark}


\section{Main Theorem}
We are now ready to prove our main theorem. At the end
of the proof there is a flowchart of the logical structure for this proof of
Theorem~\ref{thm:main}.

 \begin{proof}[Proof of Theorem~\ref{thm:main}]
 First, if $f_0 = f_3 = 0$, we separate the discussion into whether $f_1 f_2 \neq 0$. If $f_1 f_2 \neq 0$, by Theorem \ref{0ab0}  we know that it is  \numP-hard unless $f_1=f_2$, in which case it is matchgate-transformable and thus in FP. If $f_1 f_2 =0$, then by Theorem~\ref{thm:[0,1,0,0]} we know that it is  \numP-hard unless $f_1=f_2=0$, in which case the problem is trivially in FP. This finishes the case when $f_0 = f_3 = 0$.
  
 Assume now at least one of $f_0$ and $f_3$ is not $0$. 
 By considering the reversal of the signature, 
 we can assume $f_0\ne 0$, then the signature becomes $[1,a,b,c]$ after normalization. 
 
 If $c=0$, the dichotomy for $[1,a,b,0]$ is proved in Theorem \ref{1ab0}.
 
 If in $[1,a,b,c]$, $c\ne 0$, then $a$ and $b$ are symmetric by flipping. Now if $ab=0$, we can assume $b=0$ by the afore-mentioned symmetry, i.e., the signature becomes $[1,a,0,c]$. By Theorem \ref{1a0c}, it is \numP-hard unless $a=0$, in which case it is Gen-Eq. 
 
 Finally, for the problem $[1,a,b,c]$ where $abc\ne0$, Theorem \ref{large5} proves the dichotomy that it is \numP-hard unless the signature is in the tractable cases of Theorem~\ref{thm:main}.
 \end{proof}

\bigskip  
\newpage
\textbf{Flowchart of proof structure: }

\medskip

\tikzstyle{startstop} = [rectangle, rounded corners, minimum width=2cm, minimum height=1cm,text centered, draw=black, fill=red!20]
\tikzstyle{io} = [trapezium, trapezium left angle=70, trapezium right angle=110, minimum width=3cm, minimum height=1cm, text centered, draw=black, fill=blue!20]
\tikzstyle{process} = [rectangle, minimum width = 2cm, minimum height=1cm, text centered,text width=2cm, draw=black, fill=yellow!20]
\tikzstyle{decision} = [diamond, minimum width=3cm, minimum height=1cm, text centered, draw=black, fill=green!20]
\tikzstyle{arrow} = [thick,->,>=stealth,  text centered, text width=2.2cm]

\begin{tikzpicture}[node distance=4cm]
    \node (s1) [startstop] {$[f_0,f_1,f_2,f_3]$};
    \node (s2) [process, left of=s1, xshift=-1.8cm] {Dichotomy for $[0,a,b,0]$, Theorem~\ref{0ab0}};
    \node (s3) [startstop, below of=s1, yshift=1.0cm] {$[1,a,b,c]$};
    \node (s4) [process, left of=s3,xshift=-1.8cm] {Dichotomy for $[1,a,b,0]$, Theorem~\ref{1ab0}};
    \node (s5) [process, right of=s3, xshift=1.8cm] {Dichotomy for $[1,a,0,c]$, Theorem~\ref{1a0c}};
    \node (s6) [process, below of=s3, yshift=1.0cm, text width=3cm] {Dichotomy in Theorem~\ref{large5}};
    \node (s7) [process, right of=s1, xshift = 1.8cm] {Hardness for $[0,1,0,0]$, Theorem~\ref{thm:[0,1,0,0]}};
    
    \draw [arrow] (s1) -- node[anchor=south] {if $f_0=f_3=0$ and $f_1 f_2 \neq 0$} (s2);
    \draw [arrow] (s1) -- node[anchor=east] {else, possibly by flipping} (s3);
    \draw [arrow] (s3) -- node[anchor=south] {if $c=0$} (s4);
    \draw [arrow] (s3) -- node[anchor=south] {if $c \ne 0$ and $ab=0$ } (s5);
    \draw [arrow] [text width =3.5cm](s3) -- node[anchor=east] {else (i.e., $abc\ne 0$)} (s6);
    \draw [arrow] (s1) -- node[anchor=south] {if $f_0=f_3=0$ and $f_1 f_2 = 0$} (s7);
\end{tikzpicture}



\newpage

\bibliographystyle{alpha}
\bibliography{reference.bib}

\newpage

\input{Appendix}

\end{document}